\documentclass[11pt,english]{article}
\usepackage[T1]{fontenc}
\usepackage[latin9]{inputenc}
\usepackage{geometry}
\usepackage{color}
\usepackage{babel}
\usepackage{array}
\usepackage{longtable}
\usepackage{url}
\usepackage{amsmath}
\usepackage{amsthm}
\usepackage{amssymb}
\usepackage{graphicx}
\usepackage{setspace}
\usepackage[authoryear]{natbib}
\usepackage[unicode=true,pdfusetitle,
 bookmarks=true,bookmarksnumbered=false,bookmarksopen=false,
 breaklinks=false,pdfborder={0 0 1},backref=false,colorlinks=true]
 {hyperref}
\hypersetup{
 urlcolor=blue, citecolor=blue}

\makeatletter

\providecommand{\tabularnewline}{\\}
\usepackage{tikz}
\usetikzlibrary{calc}

\theoremstyle{remark}
\newtheorem*{notation*}{\protect\notationname}
\theoremstyle{definition}
 \newtheorem{example}{\protect\examplename}
\theoremstyle{definition}
\newtheorem{defn}{\protect\definitionname}
\theoremstyle{plain}
\newtheorem{thm}{\protect\theoremname}
\theoremstyle{plain}
\newtheorem{cor}{\protect\corollaryname}
\theoremstyle{plain}
\newtheorem{assumption}{\protect\assumptionname}
\theoremstyle{definition}
\newtheorem*{example*}{\protect\examplename}
\theoremstyle{plain}
\newtheorem{lem}{\protect\lemmaname}
\theoremstyle{plain}
\newtheorem*{prop*}{\protect\propositionname}
\theoremstyle{plain}
\newtheorem{prop}{\protect\propositionname}

\usepackage{dcolumn}
\usepackage{booktabs} 
\usepackage{adjustbox}
\usepackage{pdfpages}
\usepackage[para,online,flushleft]{threeparttable}
\usepackage{graphicx}
\usepackage{subcaption}
\usepackage[bottom]{footmisc}
\usepackage{setspace}

\makeatother

\providecommand{\assumptionname}{Assumption}
\providecommand{\corollaryname}{Corollary}
\providecommand{\definitionname}{Definition}
\providecommand{\examplename}{Example}
\providecommand{\lemmaname}{Lemma}
\providecommand{\notationname}{Notation}
\providecommand{\propositionname}{Proposition}
\providecommand{\theoremname}{Theorem}

\begin{document}
\title{\onehalfspacing{}Stable Outcomes and Information in Games: \\
An Empirical Framework\thanks{I am indebted to Sokbae Lee, Bernard Salanié, and Qingmin Liu for
their guidance and support. I would like to thank Matthew Backus,
Yeon-koo Che, Pierre-André Chiappori, Evan Friedman, Duarte Gonçalves,
Gautam Gowrisankaran, Miho Hong, Hiroaki Kaido, Dilip Ravindran, and
all seminar participants at Boston University and Columbia University.
I also thank an associate editor and two anonymous referees for constructive
comments that improved this paper. The views expressed in this article
are those of the author and do not necessarily reflect those of the
Federal Trade Commission or any individual Commissioner. This paper
is based on the first chapter of my PhD dissertation. All errors are
mine.}}
\author{\onehalfspacing{}Paul S. Koh\thanks{Federal Trade Commission. Email: \texttt{pkoh@ftc.gov}.}}
\date{May 18, 2023}
\maketitle
\begin{abstract}
\begin{onehalfspace}
Empirically, many strategic settings are characterized by stable outcomes
in which players\textquoteright{} decisions are publicly observed,
yet no player takes the opportunity to deviate. To analyze such situations
in the presence of incomplete information, we build an empirical framework
by introducing a novel solution concept that we call \emph{Bayes stable
equilibrium} and computationally tractable approaches for estimation
and inference\emph{.} Our framework allows the researcher to be agnostic
about players' information and the equilibrium selection rule. In
an application, we study the strategic entry decisions of McDonald's
and Burger King in the US. While the Bayes stable equilibrium identified
set is always (weakly) tighter than the Bayes correlated equilibrium
identified set, our results show that the former can be substantially
tighter in practice. In a counterfactual experiment, we examine the
impact of increasing access to healthy food on the market structures
in Mississippi food deserts.
\end{onehalfspace}
\end{abstract}

\section{Introduction}

In dynamic strategic settings where firms can react after observing
their opponents' choices, our intuitions suggest that firms' actions
would change over time. Interestingly, we often see firms reach a
certain steady state in which no firm changes its decision even when
it can. For example, major exporters' decisions to export products
to specific markets remain unchanged for a long period \citep{ciliberto2021superstar}.
Airline firms' decisions to operate between cities tend to be persistent
\citep{ciliberto2009marketstructure}. Food-service retailers operate
in a local market over a long horizon, knowing precisely the identities
of the competitors operating nearby. In all these examples, each firm's
action constitutes a best response to the \emph{observed} actions
of the opponents.

The prevalence of incomplete information in the real world makes the
phenomenon particularly interesting. If opponents' actions are observable
at the steady state, rational firms will use the observation to update
their beliefs. For example, while a coffee chain's own research might
find a given neighborhood unattractive, observing that Starbucks\textemdash a
chain known to have leading market research technology\textemdash enter
the neighborhood may make it think twice.\footnote{According to Tom O'Keefe, the founder of Tully's Coffee, Tully's early
business expansion strategy was to ``open across the street from
every Starbucks'' because ``they do a great job at finding good
locations.'' \citep{goll2000seattles}.} If there is no further revision of actions, it must be that each
firm holds beliefs refined by their observations of opponents' actions,
but no further updating is possible.

Although stable outcomes in the presence of information asymmetries
are common in the real world, it is not straightforward to model the
data generating process. The main difficulty arises from requiring
that the firms' beliefs and actions be consistent with each other.
On the one hand, firms' beliefs must support the realized actions
as optimal. On the other hand, each firm's beliefs must incorporate
its private information \emph{as well as} the information extracted
from observing its opponents' decisions. Static Bayes Nash equilibrium,
although a popular modeling choice, does not account for the possible
revision of actions after opponents' actions are observed. Modeling
convergence to stable outcomes via a dynamic games framework may be
feasible but is likely non-trivial and reliant on ad hoc assumptions.
In this paper, we develop a tractable equilibrium notion that satisfies
the consistency requirement and facilitates econometric analysis when
the analyst observes a cross-section of stable outcomes at some point
in time.

We propose a solution concept dubbed \emph{Bayes stable equilibrium}
as a basis for analyzing stable outcomes in the presence of incomplete
information. Bayes stable equilibrium is described as follows. A \emph{decision
rule} specifies a distribution over action profiles for each realization
of the state of the world and players' private signals. Suppose that,
after the state of the world and private signals are realized, an
action profile is drawn from the decision rule, and the action profile
is \emph{publicly} recommended to the players. The decision rule is
a Bayes stable equilibrium if the players always find no incentives
to deviate from the publicly recommended action profile after observing
their private signals and the action profile.

We justify Bayes stable equilibrium using a version of rational expectations
equilibrium à la \citet{radner1979rational}. First, we argue that
rational expectations equilibrium, appropriately defined for our setting,
provides a simple approach to rationalizing stable outcomes under
incomplete information. We define rational expectations equilibrium
by adopting the ``outcome function'' approach of \citet{liu2020stability},
who uses a similar approach to define the notion of stability in two-sided
markets with incomplete information. Next, we show that Bayes stable
equilibrium characterizes the implications of rational expectations
equilibria when the analyst can only specify the minimal information
available to the players. Thus, Bayes stable equilibrium is useful
as it allows the analyst to be ``informationally robust'' in the
same sense as the Bayes correlated equilibrium of \citet{bergemann2016bayescorrelated}.
The informational robustness property is attractive since it is often
difficult for the analyst to know the true information structure governing
the data generating process.

Assuming that the analyst observes a cross-section of stable outcomes,
we characterize the identified set of parameters using Bayes stable
equilibrium as a solution concept. The corresponding identified set
has a number of attractive properties. First, the identified set is
valid for arbitrary equilibrium selection rules and robust to the
possibility that the players actually observed more information than
specified by the analyst. We let the model be ``incomplete'' in
the sense of \citet{tamer2003incomplete}, and the parameters are
typically partially identified. Second, when strong assumptions on
information are made, the Bayes stable equilibrium identified set
collapses to the pure strategy Nash equilibrium identified set studied
in \citet*{beresteanu2011sharpidentification} and \citet{galichon2011setidentification}.
Third, everything else equal, the Bayes stable equilibrium identified
set is (weakly) tighter than the Bayes correlated equilibrium identified
set studied in \citet{magnolfi2021estimation}. While Bayes stable
equilibrium and Bayes correlated equilibrium both facilitate estimation
of games with weak assumptions on players' information, the former
is stronger as it leverages the assumption that players' actions are
observable to each other at equilibrium situations.

We propose a computationally tractable approach to estimation and
inference. We show that checking whether a candidate parameter enters
the identified set (asking whether we can find an equilibrium consistent
with data) solves a linear program. Furthermore, we propose a simple
approach to constructing confidence sets for the identified set by
leveraging the insights from \citet{horowitz2021inference}. The key
idea is to construct convex confidence sets for the conditional choice
probabilities, which are the only source of sampling uncertainty.
Checking whether a candidate parameter belongs to the confidence set
amounts to solving a convex program. 

As an empirical application, we use our framework to analyze the strategic
entry decisions of McDonald's and Burger King in the US. We estimate
the model parameters using Bayes stable equilibrium and explore the
role of informational assumptions on identification. We also use the
model to simulate the impact of increasing access to healthy food
in Mississippi food deserts. We find that popular assumptions on players'
information may be too strong, as the corresponding identified set
can be empty. On the other hand, making no assumption on players'
information produces an identified set that is too large, indicating
that some assumptions on information are necessary to produce informative
results. We show that an informative identified set can be obtained
under an intermediate assumption, which is also credible; this specification
assumes that McDonald's has accurate information about its payoff
shocks while Burger King may minimally observe nothing. We also compute
the identified sets under the Bayes correlated equilibrium assumption
and find that the Bayes stable equilibrium identified sets are substantially
tighter under the same assumptions on players' information: the volume\textemdash measured
as the product of the projection intervals\textemdash under Bayes
stable equilibrium is at most 5\% of that under Bayes correlated equilibrium.

\subsubsection*{Related Literature}

Our work adds to the literature on the econometric analysis of game-theoretic
models by designing a framework that applies to a class of situations
characterized by stable outcomes (see \citet{depaula2013econometric}
and \citet{aradillas-lopez2020theeconometrics} for recent surveys).\footnote{In his survey on the econometrics of static games, \citet{aradillas-lopez2020theeconometrics}
classifies existing papers around five criteria: (i) Nash equilibrium
versus weaker solution concepts; (ii) the presence of multiple solutions;
(iii) complete- versus incomplete-information games; (iv) correct
versus incorrect beliefs; (v) parametric versus nonparametric models.
To place our work in these categories, this paper (i) develops a new
solution concept that is weaker than complete information pure strategy
Nash equilibrium but stronger than Bayes correlated equilibrium; (ii)
admits a set of equilibria; (iii) allows a general form of incomplete
information which accommodate standard assumptions as special cases;
(iv) assumes that players have correct beliefs; (v) imposes parametric
assumptions on the payoff functions and the distribution of unobservables.} Our framework would be well-suited when (i) it is reasonable to assume
that the realized actions represent best responses to the \emph{observed}
decisions of the opponents, (ii) the stability of outcomes is not
driven by high costs of revising actions, and (iii) the analyst observes
cross-sectional data of firms' stable decisions at some point in time.\footnote{This idea behind cross-sectional analysis of games is accentuated
in \citet{ciliberto2009marketstructure}: ``\emph{The idea behind
cross-section studies is that in each market, firms are in a long-run
equilibrium. The objective of our econometric analysis is to infer
long-run relationships between the exogenous variables in the data
and the market structure that we observe at some point in time, without
trying to explain how firms reached the observed equilibrium.}''
(pp.1792-1793).}

Our framework differs from the usual Nash framework. To account for
stable outcomes, we assume players can observe opponents' actions
and react. In contrast, in static Nash frameworks, players are not
allowed to change their ``one-shot'' actions and therefore may be
subject to regrets after observing the realized actions of their opponents.\footnote{The empirical literature has been aware that the Nash framework is
subject to ex-post regret when information is incomplete or players
are using mixed strategies. See, for example, the discussions in \citet{draganska2008discrete},
\citet{einav_not_2010}, and \citet{ellickson2011structural}. } Furthermore, we are not aware of dynamic models (e.g., frameworks
based on Markov perfect equilibrium) that can straightforwardly handle
stable outcomes in incomplete information environment. 

Bayes stable equilibrium allows the researcher to work with weak assumptions
on players' information. An early work in this spirit is \citet{grieco2014discrete},
which considers a parametric class of information structures that
nests standard assumptions. Our work is most closely related to recent
papers that use Bayes correlated equilibrium as a basis for informationally
robust econometric analysis: \citet{magnolfi2021estimation} applies
Bayes correlated equilibrium to static entry games (which are also
considered in this paper), \citet*{syrgkanis2021inference} to auctions,
and \citet{gualdani2020identification} to static, single-agent models.\footnote{There is also a strand of literature that studies the possibility
that firms might have biased beliefs (see \citet{aguirregabiria2020identification}
and \citet{aguirregabiria2020firmstextquoteright} for a review).
The works in this literature assume that the econometrician knows
the true information structure of the game but firms may not have
correct beliefs. In contrast, we assume that firms have correct beliefs
but the econometrician does not know the true information structure.}

We contribute to the literature on the econometrics of moment inequality
models by proposing a simple approach to constructing confidence sets
based on the idea of \citet{horowitz2021inference}.\footnote{Recent development in inference with moment inequality models has
introduced many alternative approaches for constructing confidence
sets (see \citet{ho2017partial}, \citet{canay2017practical}, and
\citet{molinari2020microeconometrics2} for recent surveys). However,
to the best of our knowledge, most are not directly applicable to
our setup, primarily due to the presence of a high-dimensional nuisance
parameter and a large number of inequalities. A feasible strategy
for inference is the subsampling approach of \citet*{chernozhukov2007estimation},
which is also used in \citet{magnolfi2021estimation} and \citet*{syrgkanis2021inference}.} Our approach is new in the context of econometric analysis of game-theoretic
models and applicable under alternative solution concepts such as
pure strategy Nash equilibrium or Bayes correlated equilibrium.

Our work also relates to the game theory literature in two dimensions.
First, our solution concept adopts the idea of rational expectations
equilibrium pioneered by \citet{radner1979rational} to capture how
players refine their information based on market observables in equilibrium
situations. Our approach closely follows the logic in \citet{liu2020stability},
which uses the same idea to define the notion of stability in two-sided
markets with incomplete information. Compared to other works that
study solution concepts based on rational expectations equilibrium
(e.g., \citet{green1987posterior}, \citet{minehart1999expost}, \citet{minelli2003information},
and \citet{kalai2004largerobust}), we do not assume that actions
are generated by a product of individual strategy mappings nor that
players' types are fully revealed after actions are realized. Second,
our solution concept also adds to the recent literature that studies
solution concepts with informational robustness properties (e.g.,
Bergemann and Morris \citeyearpar{bergemann2013robustpredictions,bergemann2016bayescorrelated,bergemann2017belieffree}
and \citet{doval2020sequential}). 

 Finally, our empirical application contributes to the literature
on entry competition in the fast-food industry. Existing empirical
works that study strategic entries by the top burger chains include
\citet{toivanen2005marketstructure}, \citet{thomadsen2007product},
\citet{yang2012burgerking}, \citet{gayle2015choosing}, \citet{igami2016unobserved},
\citet{yang2020learning}, and \citet{aguirregabiria2020identification}.
In particular, \citet{yang2020learning}, which studies strategic
entries in the Canadian hamburger industry, shares a similar motivation
that players extract information from the opponents' actions, but
uses a dynamic games framework to explicitly model the learning process.
Our empirical work is distinguished by the use of novel datasets and
its focus on exploring the role of informational assumptions. To the
best of our knowledge, we are the first to study the impact of local
food environment on burger chains' strategic entry decisions.\footnote{For a list of works in economics that study issues related to food
deserts, see \citet{allcott2019fooddeserts} and the references cited
therein.}

 The rest of the paper is organized as follows. Section \ref{sec:Model}
introduces the notion of Bayes stable equilibrium in a general finite
game of incomplete information and studies its property. Section \ref{sec:Identification}
sets up the econometric model and provides identification results.
Section \ref{sec:Estimation and Inference} provides econometric strategies
for computationally tractable estimation and inference. Section \ref{sec:Empirical-Application}
applies our framework to the entry game played by McDonald's and Burger
King in the US. Section \ref{sec:Conclusion} concludes. All proofs
are in Appendix \ref{sec:Proofs}.

\begin{notation*}
Throughout the paper, we will use the following notation to express
discrete probability distributions in a compact manner. When $\mathcal{Y}$
is a finite set, and $p\left(y\right)$ denotes the probability of
$y\in\mathcal{Y}$, we will use $p_{y}\equiv p\left(y\right)$. Similarly,
$q_{y\vert x}\equiv q\left(y\vert x\right)$ will be used to denote
conditional probability of $y$ given $x$. We let $\Delta_{y}\equiv\Delta\left(\mathcal{Y}\right)$
denote the probability simplex on $\mathcal{Y}$, so that $p\in\Delta_{y}$
if and only if $p_{y}\geq0$ for all $y\in\mathcal{Y}$ and $\sum_{y\in\mathcal{Y}}p_{y}=1$.
Similarly, we let $\Delta_{y\vert x}$ denote the set of all probability
distributions on $\mathcal{Y}$ conditional on $x$, so that $q\in\Delta_{y\vert x}$
if and only if $q_{y\vert x}\geq0$ for all $y$ and $\sum_{y\in\mathcal{Y}}q\left(y\vert x\right)=1$.
We also use the convention that writes an action profile as $a=\left(a_{1},...,a_{I}\right)=\left(a_{i},a_{-i}\right)$.
\end{notation*}

\section{Model\label{sec:Model}}

We consider empirical settings characterized by two properties. First,
the setting is \emph{dynamic} in the sense that players can revise
their actions after observing the opponents' actions.\footnote{We use ``dynamic'' to mean that each player can react to the realized
actions of the opponents. However, we do not introduce standard dynamic
games assumptions (e.g., finite number of periods, timing of moves,
etc.) to model players' interactions. } Second, players' actions are readily and publicly observed by others.
 In other words, we focus on certain ``steady-state'' situations
in which all players publicly observe each other's realized actions,
yet no deviation occurs even when they have the opportunity to do
so. Our objective is to describe such situations as a static equilibrium.
When conducting econometric analysis, we will assume that the analyst
observes a cross-section of stable outcomes.

In this section, we introduce Bayes stable equilibrium as a solution
concept that solves the consistency problem and facilitates econometric
analysis while allowing for weak assumptions on players' information.
Throughout the paper, we assume that the state of the world remains
persistent enough to abstract away from transitions over time, and
that the costs of revising actions are sufficiently low so that we
can ignore them.\footnote{The zero adjustment cost assumption is not essential for the key
ideas of this paper. In the real world, the costs of revising actions
are not zero. However, the relevant question is whether high adjustment
costs drive stable outcomes. We treat adjustment costs as negligible
compared to the long-run profits obtained at stable outcomes. This
is in the same spirit as the empirical matching models surveyed in
\citet{chiappori2016theeconometrics}; the stable matching condition
abstracts away from the costs of entering into or exiting a marriage.
Similar assumptions assumptions are commonly used for econometric
models of network formation although forming networks can be costly
in reality (see, e.g., \citet{de2020econometric}). The assumption
is useful for motivating an alternative to the Nash framework and
simplifying exposition. See Appendix \ref{sec:Adjustment-Costs-in}
for further discussion.} We formalize the idea in a general class of discrete games of incomplete
information, following the notation of \citet{bergemann2016bayescorrelated}.

We proceed as follows. In Section \ref{subsec:Discrete-Games-of},
we lay out the game environment. In Section \ref{subsec:Stable-Outcomes},
we formalize the notion of stable outcomes and motivate our solution
concept. In Section \ref{subsec:Rational-Expectations-Equilibrium},
we argue that rational expectations equilibrium à la \citet{radner1979rational}
can be used as a baseline solution concept for rationalizing stable
outcomes in the presence of incomplete information. In Section \ref{subsec:Bayes-Stable-Equilibrium},
we introduce Bayes stable equilibrium. Then, in Section \ref{subsec:Informational-Robustness-of BSE},
we show that Bayes stable equilibrium characterizes the implications
of rational expectations equilibria when the players might observe
more information than assumed by the analyst. In Section \ref{subsec:Relationship-to-Other},
we compare the proposed solution concepts to pure strategy Nash equilibrium
and Bayes correlated equilibrium. Finally, in Section \ref{subsec:Existence-and-Uniqueness},
we discuss issues around the existence and uniqueness of the proposed
solution concepts.

\subsection{Discrete Games of Incomplete Information\label{subsec:Discrete-Games-of}}

Let $\mathcal{I}=\left\{ 1,2...,I\right\} $ be the set of players.
The players interact in a finite game of incomplete information $\left(G,S\right)$.\footnote{Throughout this paper, we assume that the state space is finite. The
assumption simplifies the notation. In addition, even though continuous
state space can be used, we will eventually need to discretize the
space for feasible estimation. \citet{magnolfi2021estimation} and
\citet*{syrgkanis2021inference} take similar discretization approaches
for estimation with Bayes correlated equilibria. } A \emph{basic game} $G=\langle\mathcal{E},\psi,\left(\mathcal{A}_{i},u_{i}\right)_{i=1}^{I}\rangle$
specifies the payoff-relevant primitives: $\mathcal{E}$ is a finite
set of unobserved states; $\psi\in\Delta\left(\mathcal{E}\right)$
is a common prior distribution with full support; $\mathcal{A}_{i}$
is a finite set of actions available to player $i$, and $\mathcal{A}\equiv\times_{i=1}^{I}\mathcal{A}_{i}$
is the set of action profiles; $u_{i}:\mathcal{A}\times\mathcal{E}\to\mathbb{R}$
is player $i$'s von Neumann\textendash Morgenstern utility function.
An \emph{information structure} $S=\langle\left(\mathcal{T}_{i}\right)_{i=1}^{I},\pi\rangle$
specifies the information-related primitives: $\mathcal{T}_{i}$ is
a finite set of signals (or types), and $\mathcal{T}\equiv\times_{i=1}^{I}\mathcal{T}_{i}$
is the set of signal profiles; $\pi:\mathcal{E}\to\Delta\left(\mathcal{T}\right)$
is a signal distribution, which allows players' signals to be arbitrarily
correlated. The interpretation is that the state of the world $\varepsilon\in\mathcal{E}$,
which is drawn from the prior $\psi$, is not directly observed by
the players, but each player $i$ receives a private signal $t_{i}\in\mathcal{T}_{i}$
whose informativeness about $\varepsilon$ depends on the signal distribution
$\pi$. The game is common knowledge to the players. As highlighted
by \citet{bergemann2016bayescorrelated}, the separation between the
basic game and the information structure facilitates the analysis
on the role of information structures.

In empirical applications, there is a finite set of exogenous observable
covariates $\mathcal{X}$. We can augment the notation and let $\left(G^{x},S^{x}\right)$
describe the game in markets with characteristics $x\in\mathcal{X}$.
Indexing each game by $x$ is justified by assuming that $x$ is common-knowledge
to the players and that the game primitives are functions of $x$.
We suppress the dependence on $x$ for now.

The following two-player entry game serves as a running example as
well as a baseline model for our empirical application.
\begin{example}[Two-player entry game]
\label{exa:1} The basic game $G$ is described as follows. There
are two players, $i=1,2$. The state of the world $\varepsilon\in\mathcal{E}$
is a vector of payoff shocks, $\varepsilon=\left(\varepsilon_{1},\varepsilon_{2}\right)\in\mathbb{R}^{2}$,
where $\varepsilon_{i}$ enters player $i$'s payoff. Assume $\varepsilon\sim\psi$
for some distribution $\psi$, e.g., bivariate normal; $\varepsilon_{i}$'s
may be correlated. Firm $i$'s action set is $\mathcal{A}_{i}=\left\{ 0,1\right\} $,
where $a_{i}=1$ represents staying in the market and $a_{i}=0$ represents
staying out. The payoff function is $u_{i}\left(a_{i},a_{j},\varepsilon_{i}\right)=a_{i}\left(\beta_{i}+\kappa_{i}a_{j}+\varepsilon_{i}\right)$,
where $\beta_{i}\in\mathbb{R}$ is the intercept and $\kappa_{i}\in\mathbb{R}$
is the ``spillover effect'' parameter; $\kappa_{i}$ may be negative
or positive depending on the nature of competition. Then, $\beta_{i}+\varepsilon_{i}$
is the monopoly profit, $\beta_{i}+\kappa_{i}+\varepsilon_{i}$ is
the duopoly profit, and the profit from staying out is zero.

Next, we provide examples of information structures to which we will
pay special attention in our empirical application: \label{exa:special information structures}
\begin{itemize}
\item In $S^{complete}$, each player observes the realization of $\varepsilon$.
Formally, we have $\mathcal{T}_{i}\equiv\mathcal{E}$ for all player
$i$, and $\pi\left(t_{1}=\varepsilon,t_{2}=\varepsilon\vert\varepsilon\right)=1$
for each $\varepsilon$;
\item In $S^{private}$, $\varepsilon_{i}$ is private information to player
$i$. We have $\mathcal{T}_{i}\equiv\mathcal{E}_{i}$ for all player
$i$, and $\pi\left(t_{1}=\varepsilon_{1},t_{2}=\varepsilon_{2}\vert\varepsilon\right)=1$
for each $\varepsilon$;
\item In $S^{1P}$, player 1 observes $\varepsilon_{1}$, but player 2 observes
nothing. We have $\mathcal{T}_{1}\equiv\mathcal{E}_{1}$, $\mathcal{T}_{2}\equiv\left\{ 0\right\} $,
and $\pi\left(t_{1}=\varepsilon_{1},t_{2}=0\vert\varepsilon\right)=1$
for each $\varepsilon$. Player 2's signal is uninformative;
\item Finally, in $S^{null}$, both players observe nothing. We have $\mathcal{T}_{1}\equiv\mathcal{T}_{2}\equiv\left\{ 0\right\} $.
\end{itemize}
Note that the information structures described above can be ordered
from the most informative to the least informative: $S^{complete}$,
$S^{private}$, $S^{1P}$, $S^{null}$. For example, $S^{complete}$
is ``more informative'' than $S^{private}$ since each player is
allowed to ``observe more.'' We will formally define a partial order
on the set of information structures following \citet{bergemann2016bayescorrelated}
in Section \ref{subsec:Informational-Robustness-of BSE}. $\blacksquare$

\end{example}

\subsection{Stable Outcomes\label{subsec:Stable-Outcomes}}

Let us formalize the notion of stable outcomes and motivate our solution
concept.\footnote{The term ``stability'' has been used in different ways in the theory
literature depending on the context. Our notion of stability is the
closest to ``stable matching'' defined in \citet{liu2020stability}
under incomplete information matching games (the canonical complete
information stable matching is a special case). There is also ``hindsight
(or ex-post) stability'' of \citet{kalai2004largerobust}, whose
motivation is very similar to ours but differs in that it also requires
players' types to be revealed after the play. To the best of our knowledge,
the term ``Bayes stable equilibrium'' has not been used in the literature.} Suppose that, at some point in time, the state of the world is $\varepsilon$,
the private signals are $t=\left(t_{1},...,t_{I}\right)$, and the
players' decisions are $a=\left(a_{1},...,a_{I}\right)$. Assume that
each player $i$ observes her private signal $t_{i}$ as well as the
outcome $a$. What are the conditions for having no deviation at this
situation? A necessary condition is that each player $i$ holds a
belief $\mu^{i}\in\Delta\left(\mathcal{E}\right)$ that gives no incentive
to deviate from the status quo outcome $a$ unilaterally.
\begin{defn}[Stable outcome]
 An outcome $a=\left(a_{1},...,a_{I}\right)$ is \emph{stable} with
respect to a system of beliefs $\mu=\left(\mu^{i}\right)_{i=1}^{I}$
if, for each player $i=1,...,I$, 
\begin{equation}
\mathbb{E}_{\varepsilon\sim\mu^{i}}\left[u_{i}\left(a,\varepsilon\right)\right]\ge\mathbb{E}_{\varepsilon\sim\mu^{i}}\left[u_{i}\left(a_{i}',a_{-i},\varepsilon\right)\right]\label{eq:stable outcome}
\end{equation}
for all $a_{i}'\in\mathcal{A}_{i}$.

\end{defn}
In addition to actions being optimal with respect to the beliefs,
a sensible equilibrium would require that the beliefs reflect each
player's private information as well as the information revealed from
observing opponents' decisions. 

But how do these beliefs arise? In general, static Bayes Nash equilibrium
will not generate stable outcomes and stable beliefs; players may
have incentives to revise their actions after directly observing opponents'
decisions and updating beliefs by inverting the equilibrium strategies.\footnote{There may be special classes of games where ex post regret does not
arise or can be limited in a Bayes Nash equilibrium. \citet{kalai2004largerobust}
studies hindsight stability in a special class of games with many
players. \citet{Mathevet2022} study a class of information structures
called ``single meeting schemes'' in which a subset of players participant
in a meeting and get equally informed while the non-participants stay
uninformed. In this case, the informed players are not subject to
regret in a pure Bayes Nash equilibrium because they have (weakly)
more information than others and thus can predict all players' actions
(although the uninformed players may regret their actions after observing
the actions of the informed players).} While it is natural to ask whether we can use a noncooperative dynamic
game to model convergence to a pair of stable decisions and stable
beliefs, such route is likely to be non-trivial and dependent on ad
hoc assumptions. In the following sections, we propose a simple and
pragmatic approach to the problem.

\subsection{Rational Expectations Equilibrium\label{subsec:Rational-Expectations-Equilibrium}}

Before introducing Bayes stable equilibrium, which will be the solution
concept we take to econometric analysis, we define a version of rational
expectations equilibrium à la \citet{radner1979rational} that offers
a simple conceptual framework for rationalizing stable outcomes in
the presence of incomplete information. To define rational expectations
equilibrium appropriately in our setting, we follow \citet{liu2020stability}
and use the ``outcome function'' approach described as follows.\footnote{An analog of an outcome function in \citet{liu2020stability} is the
\emph{matching function} that maps players' types to an observable
match. In noncooperative games settings, \citet{minehart1999expost}
and \citet{minelli2003information} have made similar attempts to
connect rational expectations equilibrium to games without price.
While their definition of rational expectations equilibrium refers
to strategy profiles, we take a ``cooperative'' approach and use
outcomes functions, which are not necessarily the product of individual
strategy mappings.} 

Let a game $\left(G,S\right)$ be given. Let $\delta:\mathcal{T}\to\Delta\left(\mathcal{A}\right)$
be an \emph{outcome function} in $\left(G,S\right)$; an outcome function
specifies a probability distribution over action profiles at each
realization of players' signals. 
\begin{example}[Continued]
 Let us provide an example of an outcome function. Suppose that
the information structure is given by $S^{private}$ so that player
$i$ observes $\varepsilon_{i}$. Let $\delta\left(a_{1},a_{2}\vert t_{1},t_{2}\right)\in\mathbb{R}$
denote the probability of outcome $\left(a_{1},a_{2}\right)$ when
player 1 observes $t_{1}$ and player 2 observes $t_{2}$. Let 
\[
\delta\left(\varepsilon_{1},\varepsilon_{2}\right)\equiv\delta\left(\left(0,0\right),\left(1,0\right),\left(0,1\right),\left(1,1\right)\vert t_{1}=\varepsilon_{1},t_{2}=\varepsilon_{2}\right)\in\mathbb{R}^{4}
\]
 be the corresponding probability vector. An example of an outcome
function is given by
\[
\delta\left(\varepsilon_{1},\varepsilon_{2}\right)=\begin{cases}
\left(1,0,0,0\right) & \varepsilon_{1}<\bar{\varepsilon}_{1},\varepsilon_{2}<\bar{\varepsilon}_{2}\\
\left(0,1,0,0\right) & \varepsilon_{1}\geq\bar{\varepsilon}_{1},\varepsilon_{2}<\bar{\varepsilon}_{2}\\
\left(0,0,1,0\right) & \varepsilon_{1}<\bar{\varepsilon}_{1},\varepsilon_{2}\geq\bar{\varepsilon}_{2}\\
\left(0,0,0,1\right) & \varepsilon_{1}\geq\bar{\varepsilon}_{1},\varepsilon_{2}\geq\bar{\varepsilon}_{2}
\end{cases}
\]
where $\bar{\varepsilon}_{i}\in\mathbb{R}$, $i=1,2$, represents
some threshold value. The above outcome function dictates that player
$i$ is present in the market if $\varepsilon_{i}$ is above $\bar{\varepsilon}_{i}$
and absent otherwise. $\blacksquare$
\end{example}
Assume that $\delta$ is common knowledge to the players. Suppose
that, after the state of the world $\varepsilon\in\mathcal{E}$ and
the signal profile $t\in\mathcal{T}$ are realized according to the
prior distribution $\psi\left(\cdot\right)$ and the signal distribution
$\pi\left(\cdot\vert\varepsilon\right)$, an action profile $a\in\mathcal{A}$
is drawn from the outcome function $\delta\left(\cdot\vert t\right)$,
and the players publicly observe $a$. Each player $i$, having observed
his private signal and the realized action profile $\left(t_{i},a_{i},a_{-i}\right)$,
updates his beliefs about the state of the world $\varepsilon$ using
Bayes' rule, and decides whether to adhere to the observed outcome
(play $a_{i}$) or not (deviate to $a_{i}'\neq a_{i}$). If $\delta$
is such that the players always find the realized action profiles
optimal, we call it a rational expectations equilibrium of $\left(G,S\right)$.
Let $\mathbb{E}_{\varepsilon}^{\delta}\left[u_{i}\left(a_{i}',a_{-i},\varepsilon\right)\vert t_{i},a_{i},a_{-i}\right]$
denote the expected payoff to player $i$ from choosing $a_{i}'$
conditional on observing private signal $t_{i}$ and action profile
$\left(a_{i},a_{-i}\right)$.
\begin{defn}[Rational expectations equilibrium]
 An outcome function $\delta$ is a \emph{rational expectations equilibrium}
for $\left(G,S\right)$ if, for each $i=1,...,I$, $t_{i}\in\mathcal{T}_{i}$,
$\left(a_{i},a_{-i}\right)\in\mathcal{A}$ such that $\text{Pr}^{\delta}\left(t_{i},a_{i},a_{-i}\right)>0$,
we have
\begin{equation}
\mathbb{E}_{\varepsilon}^{\delta}\left[u_{i}\left(a_{i},a_{-i},\varepsilon\right)\vert t_{i},a_{i},a_{-i}\right]\geq\mathbb{E}_{\varepsilon}^{\delta}\left[u_{i}\left(a_{i}',a_{-i},\varepsilon\right)\vert t_{i},a_{i},a_{-i}\right]\label{eq:REE condition in expectation form}
\end{equation}
for all $a_{i}'\in\mathcal{A}_{i}$.
\end{defn}
The outcome function $\delta:\mathcal{T}\to\Delta\left(\mathcal{A}\right)$
represents a reduced-form relationship between players' information
and the outcome of the game. We are agnostic about the details on
how $\delta$ came about. However, it is assumed that the players
agree on $\delta$, and use it to infer opponents' information after
observing the realized decisions. Thus, $\delta$ serves as the players'
``model'' for connecting the uncertainties to the observables.

There is nothing conceptually new; we simply apply the idea of rational
expectations equilibrium to our setting. Rational expectations equilibrium
refers to a mapping from agents' information to observable market
outcomes such that the agents do not have incentives to deviate after
observing the realized market outcomes. The key idea is that if the
final market outcome is observable and depends on agents' signals
about the state of the economy, then the agents must be able to learn
others' information based on their observation of the market outcome.
The agents are said to have ``rational expectations'' because they
refine their information based on the information available at the
equilibrium situation. In \citet{radner1979rational}, there is a
price function (or a forecast function) that maps agents' signals
to market price. The agents use their observation of their price to
not only calculate their budget but also to infer others' information
via the price function. In \citet{liu2020stability}, there is a matching
function that maps agents' signals to a (two-sided) match. The agents
use their observation of a match to infer others' information before
assessing whether they have (unilateral and pairwise) incentives to
deviate from a given match. Although the exact definition varies by
economic environment\textemdash depending on endogenous outcomes of
the model and agents' optimality conditions\textemdash the logic is
parallel.

In a rational expectations equilibrium, outcomes and beliefs are determined
simultaneously such that the stability condition (\ref{eq:stable outcome})
is satisfied. If the environment\textemdash the state of the world
and the players' signals\textemdash stays unchanged and the outcomes
are generated by a rational expectations equilibrium, the realized
decisions persist over time. In the econometric analysis, we assume
that the analyst observes these decisions at some point in time.

\subsection{Bayes Stable Equilibrium\label{subsec:Bayes-Stable-Equilibrium}}

Let us introduce Bayes stable equilibrium. Let $\left(G,S\right)$
be given. A \emph{decision rule} in $\left(G,S\right)$ is a mapping
$\sigma:\mathcal{E}\times\mathcal{T}\to\Delta\left(\mathcal{A}\right)$
that specifies a probability distribution over action profiles at
each realization of state and signals. Assume that $\sigma$ is common
knowledge to the players. Suppose the data generating process is described
as follows. First, the state of the world $\varepsilon\in\mathcal{E}$
is drawn from $\psi\left(\cdot\right)$, and the profile of private
signals $t\in\mathcal{T}$ is drawn from $\pi\left(\cdot\vert\varepsilon\right)$.
Next, an action profile $a\in\mathcal{A}$ is drawn from $\sigma\left(\cdot\vert\varepsilon,t\right)$
and publicly observed by the players. Then, each player $i$, having
observed her private signal and the realized action profile $\left(t_{i},a_{i},a_{-i}\right)$,
updates her belief about the state of the world $\varepsilon$ using
Bayes' rule and decides whether to adhere to the observed outcome
(play $a_{i}$) or not (deviate to $a_{i}'\neq a_{i}$). If the players
always have no incentives to deviate from the realized action profiles,
we call $\sigma$ a Bayes stable equilibrium.
\begin{defn}[Bayes Stable Equilibrium]
 A decision rule $\sigma$ is a \emph{Bayes stable equilibrium} for
$\left(G,S\right)$ if, for each $i=1,...,I$, $t_{i}\in\mathcal{T}_{i}$,
$\left(a_{i},a_{-i}\right)\in\mathcal{A}$ such that $\text{Pr}^{\sigma}\left(t_{i},a_{i},a_{-i}\right)>0$,
we have
\begin{equation}
\mathbb{E}_{\varepsilon}^{\sigma}\left[u_{i}\left(a_{i},a_{-i},\varepsilon\right)\vert t_{i},a_{i},a_{-i}\right]\geq\mathbb{E}_{\varepsilon}^{\sigma}\left[u_{i}\left(a_{i}',a_{-i},\varepsilon\right)\vert t_{i},a_{i},a_{-i}\right]\label{eq:BSE condition in expectation form}
\end{equation}
for all $a_{i}'\in\mathcal{A}_{i}$.
\end{defn}
It is helpful to interpret $\sigma$ as the recommendation strategy
of an omniscient mediator. The mediator commits to $\sigma$ and announces
it to the players at the beginning of the game. Then, after observing
the realized $\left(\varepsilon,t\right)$, the mediator draws an
action profile $a$ from $\sigma\left(\cdot\vert\varepsilon,t\right)$
and publicly recommends it to the players. The Bayes stable equilibrium
condition requires that the publicly recommended action profiles are
always incentive compatible to the players.

Note that an outcome function $\delta$ does not depend on the state
of the world $\varepsilon$ whereas a decision rule $\sigma$ can.
The measurability of an outcome function with respect to players'
information reflects the requirement that if any outcome is to be
achieved, it cannot depend on what they do not know. On the other
hand, a decision rule allows the realized action profiles to be correlated
with the unobserved state. In the next section, we show that the correlation
arises because Bayes stable equilibrium captures the implications
of rational expectations equilibria when the players might observe
extra signals about the state of the world that are unknown to the
analyst.

We can simplify the obedience condition (\ref{eq:BSE condition in expectation form})
so that the equilibrium conditions are linear in the decision rule.
Given that player $i$ observes signal $t_{i}$ and recommendation
$\left(a_{i},a_{-i}\right)$, the expected payoff from choosing $a_{i}'$
is
\begin{align*}
\mathbb{E}_{\varepsilon}^{\sigma}\left[u_{i}\left(a_{i}',a_{-i},\varepsilon\right)\vert t_{i},a_{i},a_{-i}\right] & =\sum_{\varepsilon}u_{i}\left(a_{i}',a_{-i},\varepsilon\right)\text{Pr}^{\sigma}\left(\varepsilon\vert t_{i},a_{i},a_{-i}\right)\\
 & =\sum_{\varepsilon}u_{i}\left(a_{i}',a_{-i},\varepsilon\right)\left(\frac{\sum_{t_{-i}}\psi\left(\varepsilon\right)\pi\left(t_{i},t_{-i}\vert\varepsilon\right)\sigma\left(a_{i},a_{-i}\vert\varepsilon,t_{i},t_{-i}\right)}{\sum_{\tilde{\varepsilon},\tilde{t}_{-i}}\psi\left(\tilde{\varepsilon}\right)\pi\left(t_{i},\tilde{t}_{-i}\vert\tilde{\varepsilon}\right)\sigma\left(a_{i},a_{-i}\vert\tilde{\varepsilon},t_{i},\tilde{t}_{-i}\right)}\right).
\end{align*}
Then, after cancelling out the denominator, which is constant across
all possible realizations of $\varepsilon\in\mathcal{E},t_{-i}\in\mathcal{T}_{-i}$,
the obedience condition (\ref{eq:BSE condition in expectation form})
can be rewritten as\footnote{Using a similar argument, we can express the rational expectation
equilibrium conditions for an outcome function $\delta$ in $\left(G,S\right)$
as: $\sum_{\varepsilon,t_{-i}}\psi_{\varepsilon}\pi_{t\vert\varepsilon}\delta_{a\vert t}u_{i}\left(a,\varepsilon\right)\geq\sum_{\varepsilon,t_{-i}}\psi_{\varepsilon}\pi_{t\vert\varepsilon}\delta_{a\vert t}u_{i}\left(a_{i}',a_{-i},\varepsilon\right),\ \forall i,t_{i},a,a_{i}'.$} 
\begin{equation}
\sum_{\varepsilon,t_{-i}}\psi_{\varepsilon}\pi_{t\vert\varepsilon}\sigma_{a\vert\varepsilon,t}u_{i}\left(a,\varepsilon\right)\geq\sum_{\varepsilon,t_{-i}}\psi_{\varepsilon}\pi_{t\vert\varepsilon}\sigma_{a\vert\varepsilon,t}u_{i}\left(a_{i}',a_{-i},\varepsilon\right),\quad\forall i\in\mathcal{I},t_{i}\in\mathcal{T}_{i},a\in\mathcal{A},a_{i}'\in\mathcal{A}_{i}.\label{eq:BSE condition}
\end{equation}
Since $\sigma$ enters the expression linearly, finding a Bayes stable
equilibrium solves a linear feasibility program, a feature that renders
estimation computationally tractable.

\subsection{Informational Robustness of Bayes Stable Equilibrium\label{subsec:Informational-Robustness-of BSE}}

In Section \ref{subsec:Rational-Expectations-Equilibrium}, we have
argued that an analyst can use rational expectations equilibrium as
a description of stable outcomes under incomplete information situations.
More often than not, however, it is difficult for the analyst to know
the true information structure governing the data generating process.
Attempts to characterize all feasible predictions (joint distribution
on states, signals, and actions) of a model by a direct enumeration
over all possible information structures are likely to be futile since
the set of information structures is large. How might the analyst
proceed without making strong assumptions on players' information?

We show that Bayes stable equilibrium provides a tractable characterization
of all rational expectations equilibrium predictions that can arise
when the players might observe more information than assumed by the
analyst. Thus, Bayes stable equilibrium serves as a tool for analyzing
stable outcomes with weak assumptions on players' information. The
informational robustness property closely resembles that of Bayes
correlated equilibrium (established in Theorem 1 of \citet{bergemann2016bayescorrelated}),
namely that Bayes correlated equilibrium provides a shortcut to charactering
all Bayes Nash equilibrium predictions that can arise when the players
might observe more information than specified by the analyst.

We formalize the idea as follows. First, to capture the idea that
players observe more information under one information structure than
under another, we introduce the notion of \emph{expansion} defined
in \citet{bergemann2016bayescorrelated}.
\begin{defn}[Expansion]
 Let $S=\left(\mathcal{T},\pi\right)$ be an information structure.
$S^{*}=\left(\mathcal{T}^{*},\pi^{*}\right)$ is an \emph{expansion}
of $S$, or $S^{*}\succsim_{E}S$, if there exists $\left(\tilde{\mathcal{T}}_{i}\right)_{i=1}^{I}$
and $\lambda:\mathcal{E}\times\mathcal{T}\to\Delta\left(\tilde{\mathcal{T}}\right)$
such that $\mathcal{T}_{i}^{*}=\mathcal{T}_{i}\times\tilde{\mathcal{T}_{i}}$
for all $i=1,...,I$ and $\pi^{*}\left(t,\tilde{t}\vert\varepsilon\right)=\pi\left(t\vert\varepsilon\right)\lambda\left(\tilde{t}\vert\varepsilon,t\right)$.
\end{defn}
Intuitively, $S^{*}$ is an expansion of $S$ if each player is allowed
to observe more signals under $S^{*}$ than under $S$. In other words,
in $S$, each player $i$ observes a private signal $t_{i}$, whereas
in $S^{*}$, each $i$ gets to observe an additional signal $\tilde{t}_{i}$
generated by an augmenting signal distribution $\lambda$. The notion
of expansion defines a partial order $\succsim_{E}$ on the set of
information structures.
\begin{example}[Continued]
 Clearly, $S^{complete}\succsim_{E}S^{private}\succsim_{E}S^{1P}\succsim_{E}S^{null}$.
For example, to show $S^{private}\succsim_{E}S^{1P}$, take $\mathcal{T}_{1}^{private}=\mathcal{E}_{1}$,
$\mathcal{T}_{2}^{private}=\mathcal{E}_{2}$, $\mathcal{T}_{1}^{1P}=\mathcal{E}_{1}$,
$\mathcal{T}_{2}^{1P}=\left\{ 0\right\} $, $\tilde{\mathcal{T}}_{1}=\left\{ 0\right\} $,
$\tilde{\mathcal{T}}_{2}=\mathcal{E}_{2}$, and $\lambda\left(\tilde{t}_{1}=0,\tilde{t}_{2}=\varepsilon_{2}\vert\varepsilon_{2}\right)=1$,
i.e., in $S^{private}$, Player 2 receives an extra signal that informs
him the realization of $\varepsilon_{2}$. $\blacksquare$
\end{example}
Let $\mathcal{P}_{\varepsilon,t,a}^{BSE}\left(G,S\right)$ be the
set of joint distributions on $\mathcal{E}\times\mathcal{T}\times\mathcal{A}$
that can arise in a Bayes stable equilibrium of $\left(G,S\right)$.
Let $\mathcal{P}_{\varepsilon,t,a}^{REE}\left(G,S\right)$ be defined
similarly. Note that if $S^{*}\succsim_{E}S$, a joint distribution
on $\mathcal{E}\times\mathcal{T}^{*}\times\mathcal{A}$ induce a marginal
on $\mathcal{E}\times\mathcal{T}\times\mathcal{A}$. The following
theorem states that by considering Bayes stable equilibrium of $\left(G,S\right)$,
we can capture all joint distributions on $\mathcal{E}\times\mathcal{T}\times\mathcal{A}$
that can arise in a rational expectations equilibrium under some information
structure that is more informative than $S$.
\begin{thm}[Informational robustness]
\label{thm:BSE and REE connection} For any basic game $G$ and information
structure $S$, $\mathcal{P}_{\varepsilon,t,a}^{BSE}\left(G,S\right)=\bigcup_{S^{*}\succsim_{E}S}\mathcal{P}_{\varepsilon,t,a}^{REE}\left(G,S^{*}\right)$.
\end{thm}
The proof of the theorem closely follows that of \citet{bergemann2016bayescorrelated}
Theorem 1. The ``$\subseteq$'' direction is established by taking
the equilibrium decision rule $\sigma:\mathcal{E}\times\mathcal{T}\to\Delta\left(\mathcal{A}\right)$
as an augmenting signal function that generates a ``public signal''
$a$ that is commonly observed by the agents. We then construct a
trivial outcome function $\delta$ that places unit mass on the recommended
outcome, i.e., $\delta\left(\tilde{a}\vert a\right)=1$ if and only
if $\tilde{a}=a$. Then the rational expectations equilibrium condition
for $\delta$ in the game with augmented information structure is
implied by the obedience condition for $\sigma$. Conversely, the
``$\supseteq$'' direction is established by integrating out the
``extra signals'' $\tilde{t}_{i}$ from the rational expectations
equilibrium condition, which directly implies the obedience condition
for the induced decision rule $\sigma\left(a\vert\varepsilon,t\right)\equiv\sum_{\tilde{t}}\lambda\left(\tilde{t}\vert\varepsilon,t\right)\delta\left(a\vert t,\tilde{t}\right)$.

Theorem \ref{thm:BSE and REE connection} can be framed in terms of
marginal distributions on the action profiles. This characterization
is more relevant for econometric analysis since typical data contain
information on players' decisions but not the signals nor the state
of the world. Let $\mathcal{P}_{a}^{BSE}\left(G,S\right)$ be the
set of marginal distributions on $\mathcal{A}$ that can arise in
a Bayes stable equilibrium of $\left(G,S\right)$. Let $\mathcal{P}_{a}^{REE}\left(G,S\right)$
be defined similarly.
\begin{cor}[Observational equivalence]
\label{cor:CCP for BSE and REE} For any basic game $G$ and information
structure $S$, $\mathcal{P}_{a}^{BSE}\left(G,S\right)=\bigcup_{S^{*}\succsim_{E}S}\mathcal{P}_{a}^{REE}\left(G,S^{*}\right)$.
\end{cor}
Intuitively, allowing more information to the players should shrink
the set of equilibria because it tightens the obedience constraints.
The following corollary formalizes the idea.
\begin{cor}
\label{cor:more info leads to tighter set} For any basic game $G$
and information structures $S$ and $S'$ such that $S\succsim_{E}S'$,
$\mathcal{P}_{\varepsilon,t,a}^{BSE}\left(G,S\right)\subseteq\mathcal{P}_{\varepsilon,t,a}^{BSE}\left(G,S'\right)$.
\end{cor}

\subsection{Relationship to Other Solution Concepts\label{subsec:Relationship-to-Other}}

 In this section, we compare our solution concepts to other existing
solution concepts that have been frequently employed for empirical
analysis. First, we compare rational expectations equilibrium and
pure strategy Nash equilibrium, which are decentralized solution concepts.
We show that our framework attains pure strategy Nash equilibrium
as a special case. Second, we compare Bayes stable equilibrium and
Bayes correlated equilibrium, which are centralized solution concepts
that rely on a mediator analogy. We show that Bayes stable equilibrium
refines Bayes correlated equilibrium as the former imposes stronger
restrictions than the latter. We also provide an overview of how our
framework relates to the literature.

\subsubsection{Comparison to Pure Strategy Nash Equilibrium}

The following theorem says that pure strategy Nash equilibrium arises
as a special case of rational expectations equilibrium (or Bayes stable
equilibrium) when strong assumptions on players' information are made.
\begin{thm}[Relationship to pure strategy Nash equilibrium]
\label{thm:Relationship to PSNE} 
\begin{enumerate}
\item \label{enu:REE PSNE 1}Let $G$ be an arbitrary basic game and let
$S^{complete}$ be an information structure in which the state of
the world $\varepsilon$ is publicly observed by the players. An outcome
function $\delta:\mathcal{E}\to\Delta\left(\mathcal{A}\right)$ is
a rational expectations equilibrium of $\left(G,S^{complete}\right)$
if and only if, for every $\varepsilon\in\mathcal{E}$, $\delta_{\tilde{a}\vert\varepsilon}>0$
implies $\tilde{a}$ is a pure-strategy Nash equilibrium action profile
at $\varepsilon$. Furthermore, $\delta$ is a rational expectations
equilibrium of $\left(G,S^{complete}\right)$ if and only if it is
a Bayes stable equilibrium of $\left(G,S^{complete}\right)$.
\item \label{enu:REE PSNE 2}Suppose that the basic game $G$ is such that
$\varepsilon=\left(\varepsilon_{1},...,\varepsilon_{I}\right)$ and
$u_{i}\left(a,\varepsilon\right)=u_{i}\left(a,\varepsilon_{i}\right)$,
and let $S^{private}$ be an information structure in which each player
$i$ observes $\varepsilon_{i}$. Then an outcome function $\delta:\mathcal{E}\to\Delta\left(\mathcal{A}\right)$
is a rational expectations equilibrium of $\left(G,S^{private}\right)$
if and only if it is a rational expectations equilibrium of $\left(G,S^{complete}\right)$.
Furthermore, $\delta$ is a rational expectations equilibrium of $\left(G,S^{private}\right)$
if and only if it is a Bayes stable equilibrium of $\left(G,S^{private}\right)$.
\end{enumerate}
\end{thm}
Theorem \ref{thm:Relationship to PSNE}.\ref{enu:REE PSNE 1} states
that when information is complete, rational expectations equilibrium
is observationally equivalent to pure strategy Nash equilibrium. A
rational expectations equilibrium outcome function $\delta$ is just
a selection device over pure strategy Nash outcomes. It also implies
that, when players have complete information, a rational expectations
equilibrium exists if and only if there is at least one pure strategy
Nash equilibrium action profile at each $\varepsilon\in\mathcal{E}$
(on the support of $\psi$). 

Theorem \ref{thm:Relationship to PSNE}.\ref{enu:REE PSNE 2} states
that when $\varepsilon$ is simply a vector of player-specific payoff
shocks\textemdash a common assumption for empirical models of discrete
games\textemdash we can use weaker informational assumptions to rationalize
pure strategy Nash outcomes. Intuitively, when each player $i$ observes
his type $\varepsilon_{i}$ and an outcome $a$ in an equilibrium
situation, opponents' types $\varepsilon_{-i}$ are \emph{payoff-irrelevant}.
In a pure strategy Nash equilibrium, $i$ uses its knowledge of $\varepsilon_{-i}$
to \emph{predict} $a_{-i}$. However, in a rational expectations equilibrium,
$i$ \emph{observes} $a_{-i}$, so $\varepsilon_{-i}$ plays no role
for $i$. Therefore, under the rational expectations equilibrium assumption,
it is sufficient that player $i$ observes $\varepsilon_{i}$ in order
to support pure strategy Nash outcomes.

Note that under the assumptions in the theorem, there is no material
difference between an outcome function and a decision rule because
players' signals exhaust information about the state of the world,
so Bayes stable equilibrium and rational expectations equilibrium
are equivalent. 

\subsubsection{Comparison to Bayes Correlated Equilibrium}

Bayes stable equilibrium refines Bayes correlated equilibrium because
equilibrium conditions for the former are stronger. To describe Bayes
correlated equilibrium, suppose that an omniscient mediator commits
to a decision rule $\sigma:\mathcal{E}\times\mathcal{T}\to\Delta\left(\mathcal{A}\right)$
in $\left(G,S\right)$ and announces it to the players so that $\sigma$
is common knowledge to the players. After the state $\varepsilon$
and signal profile $t$ are drawn from $\psi$ and $\pi$ respectively,
the mediator observes $\left(\varepsilon,t\right)$ and draws an action
profile $a$ from the decision rule $\sigma\left(\cdot\vert\varepsilon,t\right)$.
Then, the mediator \emph{privately} recommends $a_{i}$ to each player
$i$. Each player $i$, having observed his private signal $t_{i}$
and the privately recommended action $a_{i}$, decides whether to
follow the recommendation (play $a_{i}$) or not (deviate to $a_{i}'\neq a_{i}$).
If the players are always obedient, then the decision rule is a Bayes
correlated equilibrium of $\left(G,S\right)$.

Formally, a decision rule $\sigma:\mathcal{E}\times\mathcal{T}\to\Delta\left(\mathcal{A}\right)$
in $\left(G,S\right)$ is a \emph{Bayes correlated equilibrium} if
for each $i\in\mathcal{I}$, $t_{i}\in\mathcal{T}_{i}$, and $a_{i}\in\mathcal{A}_{i}$,
we have
\[
\mathbb{E}_{\left(\varepsilon,a_{-i}\right)}^{\sigma}\left[u_{i}\left(a_{i},a_{-i},\varepsilon\right)\vert t_{i},a_{i}\right]\geq\mathbb{E}_{\left(\varepsilon,a_{-i}\right)}^{\sigma}\left[u_{i}\left(a_{i}',a_{-i},\varepsilon\right)\vert t_{i},a_{i}\right]
\]
for all $a_{i}'\in\mathcal{A}_{i}$ whenever $\text{Pr}^{\sigma}\left(t_{i},a_{i}\right)>0$,
or more compactly, 
\begin{equation}
\sum_{\varepsilon,t_{-i},a_{-i}}\psi_{\varepsilon}\pi_{t\vert\varepsilon}\sigma_{a\vert\varepsilon,t}u_{i}\left(a_{i},a_{-i},\varepsilon\right)\geq\sum_{\varepsilon,t_{-i},a_{-i}}\psi_{\varepsilon}\pi_{t\vert\varepsilon}\sigma_{a\vert\varepsilon,t}u_{i}\left(a_{i}',a_{-i},\varepsilon\right),\quad\forall i,t_{i},a_{i},a_{i}'.\label{eq:BCE condition}
\end{equation}

The only difference between Bayes stable equilibrium and Bayes correlated
equilibrium is that the former assumes each player $i$ observes $\left(a_{i},a_{-i}\right)$
whereas the latter assumes each $i$ observes only $a_{i}$, but not
$a_{-i}$. While the Bayes correlated equilibrium conditions (\ref{eq:BCE condition})
integrate out opponents' actions $a_{-i}$ since each player $i$
needs to form expectation over $a_{-i}$, Bayes stable equilibrium
conditions (\ref{eq:BSE condition}) condition on $a_{-i}$ because
$a_{-i}$ is observed to $i$ at the equilibrium situation. The following
is immediate.
\begin{thm}[Relationship to Bayes correlated equilibrium]
\label{thm:BSE and BCE} If a decision rule $\sigma$ is a Bayes
stable equilibrium of $\left(G,S\right)$, it is a Bayes correlated
equilibrium of $\left(G,S\right)$.
\end{thm}
Outcomes generated by a Bayes correlated equilibrium may be subject
to regret; a player who observes the realized decisions of the opponents
might want to revise her action. In contrast, Bayes stable equilibrium
explicitly requires that such regret is absent. When information is
complete, Bayes correlated equilibrium reduces to the canonical correlated
equilibrium, whereas Bayes stable equilibrium reduces to pure strategy
Nash equilibrium in the sense described in Theorem \ref{thm:Relationship to PSNE}.
When there is a single player, the two solution concepts are identical
because there is no informational feedback from observing opponents'
actions. 

\subsubsection{Relationship to the Literature}

Although the relationship between our solution concepts and static
Nash equilibrium can be gleaned from the above theorems, we provide
a compact review and discuss connections to the related literature
for the readers.

Let $\mathcal{P}_{\varepsilon,a}^{SC}\left(G,S\right)$ denote the
set of predictions (joint distributions on $\mathcal{E}\times\mathcal{A}$)
that can arise in game $\left(G,S\right)$ under solution concept
$SC$. We use $\mathcal{P}_{\varepsilon,a}^{NE}\left(G,S\right)$
and $\mathcal{P}_{\varepsilon,a}^{PSNE}\left(G,S\right)$ to represent
the set of (mixed-strategy) Bayes Nash equilibrium predictions and
pure strategy Bayes Nash equilibrium predictions respectively. Note
that $\mathcal{P}_{\varepsilon,a}^{PSNE}\left(G,S\right)\subseteq\mathcal{P}_{\varepsilon,a}^{NE}\left(G,S\right)$
since pure strategies are special cases of mixed strategies. The predictions
of various solution concepts are related to each other in the following
way.
\begin{cor}[Relationships among predictions]
\label{cor:Predictions} Let $G$ be an arbitrary basic game. 
\begin{enumerate}
\item For any information structure $S$, $\mathcal{P}_{\varepsilon,a}^{BSE}\left(G,S\right)=\bigcup_{\tilde{S}\succsim_{E}S}\mathcal{P}_{\varepsilon,a}^{REE}\left(G,\tilde{S}\right)$,
and $\mathcal{P}_{\varepsilon,a}^{BCE}\left(G,S\right)=\bigcup_{\tilde{S}\succsim_{E}S}\mathcal{P}_{\varepsilon,a}^{NE}\left(G,\tilde{S}\right)$.
\item If $S\succsim_{E}S'$, then $\mathcal{P}_{\varepsilon,a}^{BSE}\left(G,S\right)\subseteq\mathcal{P}_{\varepsilon,a}^{BSE}\left(G,S'\right)$,
and $\mathcal{P}_{\varepsilon,a}^{BCE}\left(G,S\right)\subseteq\mathcal{P}_{\varepsilon,a}^{BCE}\left(G,S'\right)$.
However, if $S\neq S'$ (even if one is an expansion of the other),
no clear relationship exists between $\mathcal{P}_{\varepsilon,a}^{REE}\left(G,S\right)$
and $\mathcal{P}_{\varepsilon,a}^{REE}\left(G,S'\right)$, nor between
$\mathcal{P}_{\varepsilon,a}^{NE}\left(G,S\right)$ and $\mathcal{P}_{\varepsilon,a}^{NE}\left(G,S'\right)$. 
\item For any information structure $S$, $\mathcal{P}_{\varepsilon,a}^{BSE}\left(G,S\right)\subseteq\mathcal{P}_{\varepsilon,a}^{BCE}\left(G,S\right)$.
\item If $S=S^{complete}$, then $\mathcal{P}_{\varepsilon,a}^{BSE}\left(G,S\right)=\mathcal{P}_{\varepsilon,a}^{REE}\left(G,S\right)=\mathcal{P}_{\varepsilon,a}^{PSNE}\left(G,S\right)$,
but $\mathcal{P}_{\varepsilon,a}^{BCE}\left(G,S\right)$ is the set
of complete information correlated equilibrium predictions. 
\end{enumerate}
\end{cor}
In the literature on econometric models of games, it has been common
to assume that the unobserved state variable $\varepsilon$ is a vector
of $\varepsilon_{i}$'s that only enter firm $i$'s payoff (see our
running example and Theorem \ref{thm:Relationship to PSNE}.\ref{enu:REE PSNE 2}).
Under this structure on payoffs and states, most papers have assumed
that the data generating process can be described by a Nash equilibrium
with information structure set to either $S^{complete}$ or $S^{incomplete}$
(an important exception is \citet{grieco2014discrete} who considers
a flexible, but parametric, information structure that nests both).
Examples of works that use pure strategy Nash equilibrium under $S^{complete}$,
which is a special case of our framework, include Bresnahan and Reiss
(\citeyear{bresnahan_entry_1990,bresnahan1991empirical,bresnahan_entry_1991}),
\citet{tamer2003incomplete}, \citet{ciliberto2009marketstructure},
\citet{bajari2010identification}, \citet{kline_identification_2015},
and \citet{aradillas2022inference}. Examples of works that use pure
strategy (Bayes) Nash equilibrium under $S^{private}$ include \citet{seim_empirical_2006},
\citet{pesendorfer_asymptotic_2008}, \citet{sweeting_strategic_2009},
\citet{aradillas-lopez_semiparametric_2010}, and \citet{bajari2010estimating}.
Since generally $\mathcal{P}_{\varepsilon,a}^{PSNE}\left(G,S^{complete}\right)\neq\mathcal{P}_{\varepsilon,a}^{PSNE}\left(G,S^{private}\right)$,
the two sets of papers rely on different model predictions. 

\citet{magnolfi2021estimation} motivate their analysis by arguing
that researchers often do not know whether the true information structure
is $S^{complete}$ or $S^{private}$ and propose using Bayes correlated
equilibrium under $S^{private}$. Bayes correlated equilibrium summarizes
the implications of Nash equilibrium with unknown information structure
in the sense that $\mathcal{P}_{\varepsilon,a}^{BCE}\left(G,S^{private}\right)=\bigcup_{\tilde{S}\succsim S^{priavte}}\mathcal{P}_{\varepsilon,a}^{NE}\left(G,\tilde{S}\right)$,
as established by \citet{bergemann2016bayescorrelated}. \citet*{syrgkanis2021inference}
apply Bayes correlated equilibrium to common-value and private-value
auctions. 

There has been no work that develops an empirical framework to tackle
the regret problem associated with Nash equilibrium.\footnote{\citet{yang2020learning} allows for information updating after observing
opponents' actions in the context of fast-food industry, but models
the interaction as a dynamic game. In contrast to his framework that
requires panel data, our framework allows the researcher to work with
cross-sectional data. } Rational expectations equilibrium provides a simple framework for
capturing steady state situations in which players observe opponents'
actions but do not deviate. Similarly to Bayes correlated equilibrium,
Bayes stable equilibrium provides robustness to informational assumptions
because $\mathcal{P}_{\varepsilon,a}^{BSE}\left(G,S\right)=\bigcup_{\tilde{S}\succsim_{E}S}\mathcal{P}_{\varepsilon,a}^{REE}\left(G,\tilde{S}\right)$.
Complete information pure strategy Nash equilibrium arises as a special
case of our framework because $\mathcal{P}_{\varepsilon,a}^{PSNE}\left(G,S^{complete}\right)=\mathcal{P}_{\varepsilon,a}^{BSE}\left(G,S^{complete}\right)\subseteq\mathcal{P}_{\varepsilon,a}^{BSE}\left(G,S\right)$
for any information structure $S$. However, the predictions under
rational expectations equilibrium or Bayes stable equilibrium are
generally unrelated to incomplete information Nash equilibrium predictions,
e.g., $\mathcal{P}_{\varepsilon,a}^{PSNE}\left(G,S^{private}\right)\not\in\mathcal{P}_{\varepsilon,a}^{BSE}\left(G,S^{private}\right)$.
Bayes stable equilibrium predictions are tighter than Bayes correlated
equilibrium predictions ($\mathcal{P}_{\varepsilon,a}^{BSE}\left(G,S\right)\subseteq\mathcal{P}_{\varepsilon,a}^{BCE}\left(G,S\right)$
for any $\left(G,S\right)$). In the empirical application, we show
that Bayes stable equilibrium can lead to a substantially tighter
identified set compared to Bayes correlated equilibrium; leveraging
the assumption that market outcomes are readily observed by the players
can add substantial identifying power.

\subsection{Existence and Uniqueness of Bayes Stable Equilibrium\label{subsec:Existence-and-Uniqueness}}

The reader may wonder about high-level conditions for the existence
and uniqueness of Bayes stable equilibrium. Unfortunately, we do not
have results applicable to a large class of games relevant to empirical
work. The existence and uniqueness of Bayes stable equilibrium are
generally not guaranteed. For instance, in the matching pennies game,
there is no Bayes stable equilibrium because there is always one player
who wants to deviate. In the battle of the sexes game, there is a
continuum of Bayes stable equilibria because any decision rule that
represents a mixture over the two pure strategy Nash equilibrium action
profiles corresponds to a Bayes stable equilibrium. 

The task is actually non-trivial even for complete information pure
strategy Nash equilibrium (which Bayes stable equilibrium boils down
to when information is complete) in discrete games environment because
we cannot use the standard fixed point theorems that depend on continuity,
convexity, and compactness. Existence and uniqueness of pure strategy
Nash equilibria in discrete games are typically checked numerically
(e.g., \citet{ciliberto2009marketstructure}'s algorithm enumerates
over all action profiles at each state to find all pure strategy Nash
equilibria). Similarly, the existence and uniqueness of Bayes stable
equilibrium should be checked on a case-by-case basis.

Fortunately, this paper provides positive results. First, knowledge
about the existence and uniqueness of complete information pure strategy
Nash equilibrium can be used to infer the existence and uniqueness
of Bayes stable equilibrium (see Theorem \ref{thm:Relationship to PSNE}).
The researcher can apply this result when dealing with a class of
games for which the existence and uniqueness of pure strategy Nash
equilibrium are well-understood (e.g., two-player entry games). Second,
numerically checking for existence of a Bayes stable equilibrium (or
a rational expectations equilibrium) can be done quickly by solving
a linear program. To the best of our knowledge, a linear programming
approach to checking the existence of complete information pure strategy
Nash equilibrium is new.

\section{Econometric Model and Identification\label{sec:Identification}}

In this section, we describe the econometric model. We characterize
the identified set under the assumption that data are generated by
a Bayes stable equilibrium and discuss its properties.

\subsection{Setup}

Let us denote observable market covariates as $x\in\mathcal{X}$ where
$\mathcal{X}$ is a finite set; $x$ is common knowledge to the players
and observed by the econometrician. At each $x\in\mathcal{X}$, the
player interact in a game $\left(G^{x,\theta},S^{x}\right)$ where
$G^{x,\theta}=\langle\mathcal{E},\psi^{x,\theta},\left(\mathcal{A}_{i},u_{i}^{x,\theta}\right)_{i=1}^{I}\rangle$
is the basic game, $S^{x}=\langle\left(\mathcal{T}_{i}\right)_{i=1}^{I},\pi^{x}\rangle$
is the information structure, and $\theta\in\Theta$ is a finite-dimensional
parameter the analyst wish to identify.\footnote{It is without loss to assume that $\mathcal{E}$ and $\mathcal{T}$
do not depend on $x$ because we can use $\mathcal{E}\equiv\cup_{x}\mathcal{E}^{x}$
and $\mathcal{T}\equiv\cup_{x}\mathcal{T}^{x}$. In principle, we
can also let $\theta$ enter the information structures, which would
make the information structures be part of the objects the econometrician
wants to identify. In this paper, however, we focus on the case where
$\theta$ only enters the payoff functions and the distribution of
the payoff shocks.} We maintain the assumption that the set $\mathcal{E}$ is finite
in order to make estimation feasible.\footnote{If the benchmark distribution of unobservables is continuous, it will
be discretized. Increasing the number of points in $\mathcal{E}$
can make the discrete approximation more accurate at the expense of
increased computational burden. See Appendix \ref{sec:Computational-Details}
for the details on how we make discrete approximations to continuous
distributions.} The parameter $\theta$ enters the prior distributions $\psi^{x,\theta}\in\Delta\left(\mathcal{E}\right)$
and the payoff functions $u_{i}^{x,\theta}:\mathcal{A}\times\mathcal{E}\to\mathbb{R}$.
As standard in the empirical literature, we assume that the state
of the world is a vector of player-specific payoff shocks, i.e., $\varepsilon=\left(\varepsilon_{1},\varepsilon_{2},...,\varepsilon_{I}\right)$
and $u_{i}^{x,\theta}\left(a,\varepsilon\right)=u_{i}^{x,\theta}\left(a,\varepsilon_{i}\right)$.

The data $\left\{ \left(a_{m},x_{m}\right)\right\} _{m=1}^{n}$ represent
a cross-section of action profiles and covariates in markets $m=1,...,n$
that are independent from each other. Let $\phi^{x}\in\Delta\left(\mathcal{A}\right)$
denote the \emph{conditional choice probabilities} that represent
the probability of observing each action profile conditional on covariate
value $x$. We assume that the econometrician can identify $\phi^{x}$
at each $x\in\mathcal{X}$ as $n\to\infty$. The set of baseline assumptions
for identification analysis is summarized below.
\begin{assumption}[Baseline assumption for identification]
\label{assu:baseline assumption for identification}
\begin{enumerate}
\item The set of covariates $\mathcal{X}$ and the set of states $\mathcal{E}$
are finite.
\item The prior distribution $\psi^{x,\theta}\in\Delta\left(\mathcal{E}\right)$
and the payoff functions $u_{i}^{x,\theta}\left(\cdot\right)$ are
known up to a finite-dimensional parameter $\theta$.
\item \label{enu:games with private values}The state of the world is a
vector of player-specific payoff shocks, i.e., $\varepsilon=\left(\varepsilon_{1},...,\varepsilon_{I}\right)$
and $u_{i}^{x,\theta}\left(a,\varepsilon\right)=u_{i}^{x,\theta}\left(a,\varepsilon_{i}\right)$.
\item \label{enu:CCPs are identified from data}Conditional choice probabilities
$\phi^{x}\in\Delta\left(\mathcal{A}\right)$, $x\in\mathcal{X}$,
are identified from the data.
\end{enumerate}
\end{assumption}
\begin{example*}
(Continued) In the baseline example, there are no observable covariates.
The econometrician assumes that the prior distribution is $\varepsilon_{i}\overset{\text{iid}}{\sim}N\left(0,1\right)$
(which will be discretized). The payoff function is $u_{i}^{\theta}\left(a_{i},a_{j},\varepsilon_{i}\right)=a_{i}\left(\kappa_{i}a_{j}+\varepsilon_{i}\right)$
where $\theta=\left(\kappa_{1},\kappa_{2}\right)\in\mathbb{R}^{2}$
is the parameter of interest. The econometrician observes the conditional
choice probabilities $\phi=\left(\phi_{\left(0,0\right)},\phi_{\left(0,1\right)},\phi_{\left(1,0\right)},\phi_{\left(1,1\right)}\right)$
whose elements represent the probability of each action profile, e.g.,
$\phi_{\left(1,0\right)}$ is the probability that firm 1 stays in
$\left(a_{1}=1\right)$ but firm 2 stays out $\left(a_{2}=0\right)$.
$\blacksquare$
\end{example*}
Given Assumption \ref{assu:baseline assumption for identification},
the identified set of parameters can be defined when the solution
concept and the information structure are specified. For any game
$\left(G^{x,\theta},S^{x}\right)$, let $\mathcal{P}_{a}^{SC}\left(G^{x,\theta},S^{x}\right)$
be the set of feasible probability distributions over action profiles
under solution concept $SC$. 
\begin{defn}[Identified set of parameters]
 Given Assumption \ref{assu:baseline assumption for identification},
a solution concept $SC$, and information structures $\tilde{S}=\left(\tilde{S}^{x}\right)_{x\in\mathcal{X}}$,
the identified set of parameters is defined as:
\[
\Theta_{I}^{SC}\left(\tilde{S}\right)\equiv\left\{ \theta\in\Theta:\ \forall x\in\mathcal{X},\ \phi^{x}\in\mathcal{P}_{a}^{SC}\left(G^{x,\theta},\tilde{S}^{x}\right)\right\} .
\]
\end{defn}
In words, a candidate parameter $\theta$ enters the identified set
$\Theta_{I}^{SC}\left(\tilde{S}\right)$ if at each $x\in\mathcal{X}$,
the observed conditional choice probabilities $\phi^{x}$ can arise
under some equilibrium.

\subsection{Identification and Informational Robustness}

Let us translate the observational equivalence between rational expectations
equilibrium and Bayes stable equilibrium (Corollary \ref{cor:CCP for BSE and REE})
in terms of identified sets. Consider the following assumption.
\begin{assumption}[Identification under rational expectations equilibrium]
\label{assu:Identification under REE} In each market with covariates
$x\in\mathcal{X}$, the data are generated by a rational expectations
equilibrium of $\left(G^{x,\theta_{0}},\tilde{S}^{x,0}\right)$ for
some information structure $\tilde{S}^{x,0}$ that is an expansion
of $S^{x}$ ($\tilde{S}^{x,0}\succsim_{E}S^{x}$).
\end{assumption}
Assumption \ref{assu:Identification under REE} says that there is
a true parameter $\theta_{0}$ underlying the data generating process,
and that at each $x\in\mathcal{X}$, the true information structure
is some $\tilde{S}^{x,0}$ that is an expansion of $S^{x}$. In practice,
we will consider a scenario where the econometrician knows the \emph{baseline
information structure} $S^{x}$, which describes the \emph{minimal}
information available to the players but not the true information
structure $\tilde{S}^{x,0}$. Then, under Assumptions \ref{assu:baseline assumption for identification}
and \ref{assu:Identification under REE}, the econometrician has
to admit all information structures that are expansions of $S^{x}$.
This approach contrasts with the traditional approach that assumes
the econometrician knows the true information structure exactly. 

However, directly working with Assumption \ref{assu:Identification under REE}
is computationally infeasible because it requires searching over the
set of information structures, which is large. We show that Assumption
\ref{assu:Identification under REE} can be replaced with the following
assumption, which does not rely on unknown information structures.
\begin{assumption}[Identification under Bayes stable equilibrium]
\label{assu:Identification under BSE} In each market with covariates
$x\in\mathcal{X}$, the data are generated by a Bayes stable equilibrium
of $\left(G^{x,\theta_{0}},S^{x}\right)$.
\end{assumption}
The following theorem is the consequence of Corollary \ref{cor:CCP for BSE and REE};
Assumption \ref{assu:Identification under REE} and Assumption \ref{assu:Identification under BSE}
are observationally equivalent.

\begin{thm}[Equivalence of identified sets]
\label{thm:Equivalence of identified sets}  The identified set under
Assumptions \ref{assu:baseline assumption for identification} and
\ref{assu:Identification under REE} is equal to the identified set
under Assumptions \ref{assu:baseline assumption for identification}
and \ref{assu:Identification under BSE}.
\end{thm}
Theorem \ref{thm:Equivalence of identified sets} says that in order
to compute the identified set when the data are generated by some
rational expectations equilibrium but with an unknown information
structure, we can proceed as if the data are generated by a Bayes
stable equilibrium with known information structure.

\citet{magnolfi2021estimation} and \citet*{syrgkanis2021inference}
develop a similar approach for informationally robust estimation of
games, but use Bayes correlated equilibrium as the solution concept.
They assume that the underlying data generating process is described
by Bayes Nash equilibria, whereas we rely on rational expectations
equilibria. Also see \citet{gualdani2020identification} for the single-agent
case.

Our identification results make no assumptions on the equilibrium
selection rule. The Bayes stable equilibrium identified set under
Assumptions \ref{assu:baseline assumption for identification} and
\ref{assu:Identification under BSE} is valid even when the data are
generated from a mixture of multiple equilibria. The convexity of
the set of Bayes stable equilibria (readily verified from the equilibrium
conditions (\ref{eq:BSE condition}) since $\sigma$ enters the expression
linearly) makes the single equilibrium assumption innocuous. For example,
if the data are generated by two equilibria $\sigma^{1}$ and $\sigma^{2}$
with mixture probability $\lambda$ and $\left(1-\lambda\right)$,
then since $\sigma^{\lambda}\equiv\lambda\sigma^{1}+\left(1-\lambda\right)\sigma^{2}$
is another equilibrium that generates the same joint distributions,
it is as if the data were generated by a single equilibrium $\sigma^{\lambda}$.\footnote{\citet*{syrgkanis2021inference} Lemma 2 presents a general argument
on why it is without loss to assume that the data are generated by
a single equilibrium if the set of predictions is convex.}

\subsection{Relationship Between Identified Sets}

Recall from Example \ref{exa:1} that in $S^{complete}$ each player
$i$ observes the realization of $\varepsilon$, and in $S^{private}$
each player $i$ observes the realization of $\varepsilon_{i}$. We
let $\Theta_{I}^{SC}\left(S^{complete}\right)$ denote the identified
set when $S^{x}=S^{complete}$ at every $x\in\mathcal{X}$; $\Theta_{I}^{SC}\left(S^{private}\right)$
is defined similarly. Finally, we write $S^{1}\succsim_{E}S^{2}$
if and only if $S^{1,x}\succsim_{E}S^{2,x}$ at every $x\in\mathcal{X}$.
The following theorem shows the relationship between identified sets.
\begin{thm}[Relationship between identified sets]
\label{thm:Relationship between identified sets} Suppose Assumption
\ref{assu:baseline assumption for identification} holds.
\begin{enumerate}
\item \label{enu:Relationship between identified sets 1}If $S\succsim_{E}S'$,
then $\Theta_{I}^{BSE}\left(S\right)\subseteq\Theta_{I}^{BSE}\left(S'\right)$.
\item \label{enu:Relationship between identified sets 2}$\Theta_{I}^{BSE}\left(S^{complete}\right)=\Theta_{I}^{PSNE}\left(S^{complete}\right)=\Theta_{I}^{BSE}\left(S^{private}\right)$.
\item \label{enu:Relationship between identified sets 3}For any information
structure $S$, $\Theta_{I}^{BSE}\left(S\right)\subseteq\Theta_{I}^{BCE}\left(S\right)$.
\end{enumerate}
\end{thm}
First, Theorem \ref{thm:Relationship between identified sets}.\ref{enu:Relationship between identified sets 1}
says that a stronger assumption on information leads to a tighter
identified set. The result directly follows from Corollary \ref{cor:more info leads to tighter set},
which says that the feasible set of equilibria shrinks when more information
is available to the players. A consequence of Theorem \ref{thm:Relationship between identified sets}.\ref{enu:Relationship between identified sets 1}
is that we will have $\Theta_{I}^{BSE}\left(S^{complete}\right)\subseteq\Theta_{I}^{BSE}\left(\tilde{S}\right)\subseteq\Theta_{I}^{BSE}\left(S^{null}\right)$
for any $\tilde{S}$, i.e., the tightest identified set is obtained
when $S^{complete}$ is assumed and the loosest identified set is
obtained when $S^{null}$ is assumed. Note that $\Theta_{I}^{BSE}\left(S^{null}\right)$
corresponds to the identified set that makes no assumption on players'
information.

Second, Theorem \ref{thm:Relationship between identified sets}.\ref{enu:Relationship between identified sets 2},
which is a consequence of Theorem \ref{thm:Relationship to PSNE},
says that Bayes stable equilibrium and pure strategy Nash equilibrium
are observationally equivalent when $S^{complete}$ is assumed.\footnote{When Assumption \ref{assu:baseline assumption for identification}.\ref{enu:games with private values}
is imposed, rational expectations equilibrium and Bayes stable equilibrium
are identical under $S^{private}$ and $S^{complete}$. This is because
a profile of players' signals is equal to the state of the world,
so conditioning on players' information is equivalent to conditioning
on the state of the world.} Furthermore, due to Assumption \ref{assu:baseline assumption for identification}.\ref{enu:games with private values},
Bayes stable equilibrium can deliver the same identified set under
$S^{private}$ which is weaker than $S^{complete}$. Thus, if the
researcher takes Bayes stable equilibrium (or rational expectations
equilibrium) to be a reasonable notion for the given empirical setting,
pure strategy Nash equilibrium outcomes can be rationalized with informational
assumptions that are weaker than the complete information assumption.

Finally, Theorem \ref{thm:Relationship between identified sets}.\ref{enu:Relationship between identified sets 3},
which follows from Theorem \ref{thm:BSE and BCE}, says that for any
baseline assumption on players' information, the Bayes stable equilibrium
identified set is a subset of the Bayes correlated equilibrium identified
set.

\subsection{Identifying Power of Informational Assumptions}

We use a two-player entry game (our running example) to numerically
illustrate the identifying power of various informational assumptions
in the spirit of \citet{aradillas-lopez_identification_2008}. We
also compare the identifying power to that of Bayes correlated equilibrium
studied in \citet{magnolfi2021estimation}.

Each player's payoff function is $u_{i}^{\theta}\left(a_{i},a_{j},\varepsilon_{i}\right)=a_{i}\left(\kappa_{i}a_{j}+\varepsilon_{i}\right)$.
We assume $\left(\varepsilon_{1},\varepsilon_{2}\right)$ follows
a bivariate normal distribution with zero mean, unit variance, and
zero correlation. As a discrete approximation to the prior distribution,
we use a grid of 30 points for each $\mathcal{E}_{i}$ and a Gaussian
copula to assign appropriate probability mass on each grid point $\left(\varepsilon_{1},\varepsilon_{2}\right)$.\footnote{Computational details can be found in Appendix \ref{sec:Computational-Details}.}
We set $\theta=\left(\kappa_{1},\kappa_{2}\right)=\left(-1.0,-1.0\right)$
and generate choice probabilities using the pure strategy Nash equilibrium
assumption with arbitrary selection rule.\footnote{Specifically, we generate population choice probability by finding
a feasible $\sigma:\mathcal{E}\to\Delta\left(\mathcal{A}\right)$
which satisfies the inequalities in (\ref{eq:PSNE feasibility}) as
described in Section \ref{subsec:A-Linear-Programming}.}

To construct the identified sets, we take the distribution of unobservables
as known, and collect all points $\left(\kappa_{1},\kappa_{2}\right)$
compatible with the given solution concept and informational assumptions.
We plot the convex hulls of the identified sets in Figure \ref{fig:Identified-sets}.

\begin{figure}[h]
\caption{\label{fig:Identified-sets} Convex Hulls of Identified Sets}

\centering{}%
\begin{minipage}[t]{0.45\columnwidth}%
\begin{center}
(a) BSE
\par\end{center}
\begin{center}
\includegraphics[scale=0.4]{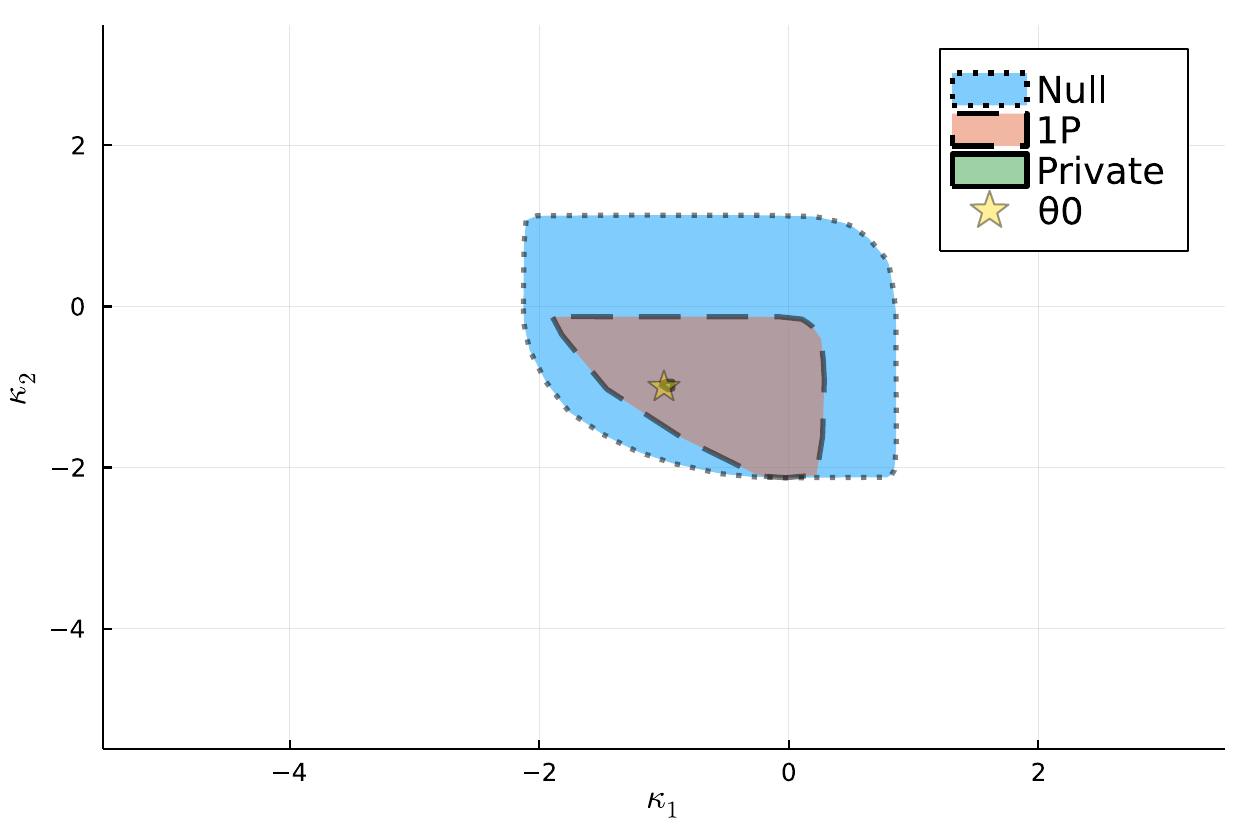}
\par\end{center}%
\end{minipage}\hfill{}%
\begin{minipage}[t]{0.45\columnwidth}%
\begin{center}
(b) BCE
\par\end{center}
\begin{center}
\includegraphics[scale=0.4]{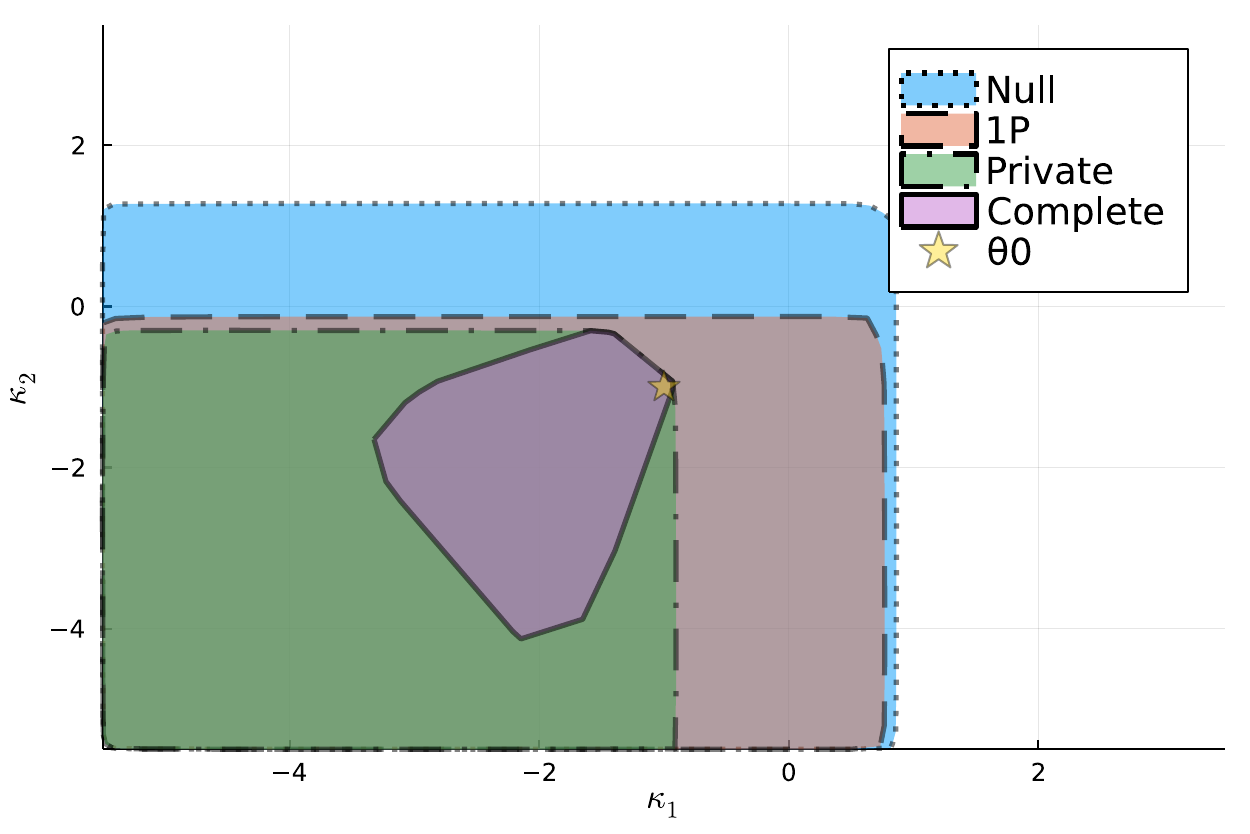}
\par\end{center}%
\end{minipage}
\end{figure}

Figure \ref{fig:Identified-sets}-(a) shows the BSE identified sets
obtained under different baseline information structures. The identified
sets shrink as the informational assumptions get stronger. We omit
the complete information case since $\Theta_{I}^{BSE}\left(S^{private}\right)=\Theta_{I}^{BSE}\left(S^{complete}\right)$.
Setting the baseline information structure as $S^{null}$ generates
an identified set that is quite permissive while using $S^{private}$
generates a tight identified set. Note that $\Theta_{I}^{BSE}\left(S^{null}\right)$
amounts to making no assumption on what the players minimally observe,
and $\Theta_{I}^{BSE}\left(S^{private}\right)$ is equal to the PSNE
identified set. Similarly, Figure \ref{fig:Identified-sets}-(b) plots
the BCE identified sets obtained under different baseline information
structures. It shows that stronger assumptions on information lead
to tighter identified sets. Assumptions on players' information play
a crucial role in determining the size of the identified set. In this
sense, imposing strong assumption on players' information may be far
from innocuous because it places strong restrictions for identification.

As stated in Theorem \ref{thm:Relationship between identified sets}.\ref{enu:Relationship between identified sets 3},
comparing Figure \ref{fig:Identified-sets}-(a) and \ref{fig:Identified-sets}-(b)
shows that, for any given baseline information structure, the corresponding
BSE identified set is a subset of the corresponding BCE identified
set. In our example, under the same informational assumption, the
BSE identified set can be substantially tighter than the BCE identified
set, illustrating the identifying power of leveraging observability
of opponents' actions in the equilibrium conditions.

\section{Estimation and Inference\label{sec:Estimation and Inference}}

We propose a computationally attractive approach for estimation and
inference. In Section \ref{subsec:A-Linear-Programming}, we show
that whether a candidate parameter enters the identified set can be
determined by solving a single linear feasibility program. In Section
\ref{subsec:A-Simple-Approach}, we show that this property can be
combined with the insights from \citet{horowitz2021inference} to
make construction of confidence sets simple and computationally tractable:
determining whether a candidate parameter enters the confidence set
amounts to solving a convex feasibility program. Finally, in Section
\ref{subsec:Implementation}, we provide some practical suggestions
for computational implementations.

\subsection{A Linear Programming Characterization\label{subsec:A-Linear-Programming}}

We provide a computationally attractive characterization of the identified
set. \citet*{syrgkanis2021inference} uses a similar characterization,
but with Bayes correlated equilibrium. Bayes stable equilibrium and
Bayes correlated equilibrium share similar computational property
since decision rules enter the equilibrium conditions linearly in
both cases. 

Let $\Theta_{I}\equiv\Theta_{I}^{BSE}\left(S\right)$ denote the sharp
identified set. Let $\partial u_{i}^{x,\theta}\left(a_{i}',a,\varepsilon_{i}\right)\equiv u_{i}^{x,\theta}\left(a_{i}',a_{-i},\varepsilon_{i}\right)-u_{i}^{x,\theta}\left(a_{i},a_{-i},\varepsilon_{i}\right)$
denote the gains from unilaterally deviating to $a_{i}'$ from outcome
$\left(a_{i},a_{-i}\right)$ given $\varepsilon_{i}$. Recall our
notation that $\sigma^{x}\in\Delta_{a\vert\varepsilon,t}$ if and
only if $\sigma_{a\vert\varepsilon,t}^{x}\geq0$ for all $a,\varepsilon,t$
and $\sum_{a\in\mathcal{A}}\sigma_{a\vert\varepsilon,t}^{x}=1$.
\begin{thm}[Linear programming characterization]
\label{thm:Identified set LP} Under Assumptions \ref{assu:baseline assumption for identification}
and \ref{assu:Identification under BSE}, $\theta\in\Theta_{I}$ if
and only if, for each $x\in\mathcal{X}$, there exists $\sigma^{x}\in\Delta_{a\vert\varepsilon,t}$
such that
\begin{enumerate}
\item (Obedience) For all $i\in\mathcal{I}$, $t_{i}\in\mathcal{T}_{i}$,
$a\in\mathcal{A}$, $a_{i}'\in\mathcal{A}_{i}$,
\begin{equation}
\sum_{\varepsilon\in\mathcal{E},t_{-i}\in\mathcal{T}_{-i}}\psi_{\varepsilon}^{x,\theta}\pi_{t\vert\varepsilon}^{x}\sigma_{a\vert\varepsilon,t}^{x}\partial u_{i}^{x,\theta}\left(a_{i}',a,\varepsilon_{i}\right)\leq0.\label{eq:LP IC}
\end{equation}
\item (Consistency) For all $a\in\mathcal{A}$,
\begin{equation}
\phi_{a}^{x}=\sum_{\varepsilon\in\mathcal{E},t\in\mathcal{T}}\psi_{\varepsilon}^{x,\theta}\pi_{t\vert\varepsilon}^{x}\sigma_{a\vert\varepsilon,t}^{x}.\label{eq:LP consistency}
\end{equation}
\end{enumerate}
\end{thm}
Theorem \ref{thm:Identified set LP} says that for any candidate $\theta\in\Theta$,
whether $\theta\in\Theta_{I}$ can be determined by solving a single
linear feasibility program. The first condition (\ref{eq:LP IC})
states that the nuisance parameter $\sigma^{x}$ should be a decision
rule that satisfies the Bayes stable equilibrium conditions. The second
condition (\ref{eq:LP consistency}) states that the observed conditional
choice probabilities must be consistent with those induced by the
equilibrium decision rule. Given a candidate $\theta$ as fixed, $\psi_{\varepsilon}^{x,\theta}$,
$\pi_{t\vert\varepsilon}^{x}$, $\partial u_{i}^{x,\theta}$, and
$\phi_{a}^{x}$ are known objects. Also note that $\sigma^{x}\in\Delta_{a\vert\varepsilon,t}$
represent constraints that are linear in $\sigma^{x}$. Then, since
the variables of optimization $\sigma^{x}$ enter the constraints
linearly, the program is linear.

Since our empirical framework obtains pure strategy Nash equilibrium
as a special case, the complete information pure strategy Nash equilibrium
identified set can be computed using linear programs as well. Let
$\Theta_{I}^{PSNE}$ be the sharp identified set obtained under the
pure strategy Nash equilibrium assumption and no assumption on the
equilibrium selection rule. As a corollary to Theorem \ref{thm:Relationship between identified sets}
and Theorem \ref{thm:Identified set LP}, whether $\theta\in\Theta_{I}^{PSNE}$
can also be determined via a single linear feasibility program. Thus,
Bayes stable equilibrium identified sets embed the pure strategy Nash
equilibrium identified set studied in \citet*{beresteanu2011sharpidentification}
and \citet{galichon2011setidentification} as a special case.
\begin{cor}[Linear programming characterization of PSNE identified set]
\label{cor:PSNE identified set} $\theta\in\Theta_{I}^{PSNE}$ if
and only if, for each $x\in\mathcal{X}$, there exists $\sigma^{x}\in\Delta_{a\vert\varepsilon}$
such that
\begin{enumerate}
\item (Obedience) For all $i\in\mathcal{I}$, $\varepsilon_{i}\in\mathcal{E}_{i}$,
$a\in\mathcal{A}$, $a_{i}'\in\mathcal{A}_{i}$,
\[
\sum_{\varepsilon_{-i}\in\mathcal{E}_{-i}}\psi_{\varepsilon}^{x,\theta}\sigma_{a\vert\varepsilon}^{x}\partial u_{i}^{x,\theta}\left(a_{i}',a,\varepsilon_{i}\right)\leq0.
\]
\item (Consistency) For all $a\in\mathcal{A}$,
\[
\phi_{a}^{x}=\sum_{\varepsilon\in\mathcal{E}}\psi_{\varepsilon}^{x}\sigma_{a\vert\varepsilon}^{x}.
\]
\end{enumerate}
\end{cor}
\begin{example*}[Continued]
 Suppose the econometrician wants to identify $\theta=\left(\kappa_{1},\kappa_{2}\right)\in\mathbb{R}^{2}$
based on the population choice probabilities $\phi=\left(\phi_{\left(0,0\right)},\phi_{\left(0,1\right)},\phi_{\left(1,0\right)},\phi_{\left(1,1\right)}\right)\in\mathbb{R}^{4}$.
Then $\theta\in\Theta_{I}^{PSNE}$ if and only if there exists $\sigma\in\Delta_{a\vert\varepsilon}$
such that
\begin{gather}
\sum_{\varepsilon_{-i}}\psi_{\varepsilon}\sigma_{a\vert\varepsilon}\left(\left(a_{i}'-a_{i}\right)\left(\kappa_{i}a_{-i}+\varepsilon_{i}\right)\right)\leq0,\quad\forall i,\varepsilon_{i},a_{i},a_{-i},a_{i}'\label{eq:PSNE feasibility}\\
\phi_{a}=\sum_{\varepsilon}\psi_{\varepsilon}\sigma_{a\vert\varepsilon},\quad\forall a.\nonumber 
\end{gather}
which is a linear feasibility program. $\blacksquare$
\end{example*}

\subsection{A Simple Approach to Inference\label{subsec:A-Simple-Approach}}

We leverage the insights from \citet{horowitz2021inference} and propose
a simple approach to inference on the structural parameters.\footnote{\citet{horowitz2021inference} describe methods for carrying out non-asymptotic
inference when the partially identified parameters are solutions to
a class of optimization problem. While we leverage the insights from
their work, we focus on asymptotic inference with multinomial proportion
parameters.} The key idea behind our approach is summarized as follows. In discrete
games, all information in the data is summarized by the conditional
choice probabilities, as apparent in Theorem \ref{thm:Identified set LP}.
The statistical sampling uncertainty arises only from the estimation
of the unknown population conditional choice probabilities, which
are multinomial proportion parameters. Then, if we control for the
sampling uncertainty associated with the estimation of the conditional
choice probabilities, we can conduct inference on the structural parameters
of interest. This strategy is feasible given that the number of multinomial
proportion parameters to estimate is small relative to the sample
size. Thus, we construct a confidence set for the conditional choice
probabilities, and translate inference on the conditional choice probabilities
to inference on the structural parameters using the characterizations
in Theorem \ref{thm:Identified set LP}.\footnote{A similar idea has been used by \citet{kline2016bayesian} who propose
a Bayesian method for inference. They leverage the idea that a posterior
on the reduced-form parameters (the conditional choice probabilities)
can be translated to posterior statements on $\theta$ using a known
mapping between them.}

Let $\phi\equiv\left(\phi^{x}\right)_{x\in\mathcal{X}}$ be the population
choice probabilities. Let us make the dependence of the identified
set on $\phi$ explicit by writing
\[
\Theta_{I}\equiv\Theta_{I}\left(\phi\right).
\]
In other words, the identified set is constructed by inverting the
mapping from the structural parameters to the conditional choice probabilities;
if we know $\phi$ accurately, then we can obtain the population identified
set.

When there is a finite number of observations, $\phi$ is unknown.
However, we are able to construct a confidence set for $\phi$ that
accounts for the sampling uncertainty. Let $\alpha\in\left(0,1\right)$.
We assume that the econometrician can construct a \emph{convex} confidence
set $\Phi_{n}^{\alpha}$ that covers $\phi$ with high probability
asymptotically.
\begin{assumption}[Convex confidence set for CCP]
\label{assu:Convex confidence set for CCP} Let $\alpha\in\left(0,1\right)$.
A set $\Phi_{n}^{\alpha}$ such that 
\[
\lim\inf_{n\to\infty}\text{Pr}\left(\phi\in\Phi_{n}^{\alpha}\right)\geq1-\alpha
\]
is available. Moreover, $\phi\in\Phi_{n}^{\alpha}$ can be expressed
as a collection of convex constraints.
\end{assumption}
Leading examples of $\Phi_{n}^{\alpha}$ are box constraints or ellipsoid
constraints; the former will be characterized by constraints that
are linear in $\phi^{x}$ and the latter will be characterized by
those quadratic in $\phi^{x}$. For example, we can construct simultaneous
confidence intervals for each $\phi_{a}^{x}\in\mathbb{R}$ such that
the probability of covering all $\left\{ \phi_{a}^{x}\right\} _{a\in\mathcal{A},x\in\mathcal{X}}$
simultaneously is asymptotically no smaller than $1-\alpha$.

Define the confidence set for the identified set as 
\begin{equation}
\widehat{\Theta}_{I}^{\alpha}\equiv\bigcup_{\tilde{\phi}\in\Phi_{n}^{\alpha}}\Theta_{I}\left(\tilde{\phi}\right).\label{eq:confidence set expression}
\end{equation}
By construction, if $\Phi_{n}^{\alpha}$ covers $\phi$ with high
probability, then $\widehat{\Theta}_{I}^{\alpha}$ covers $\Theta_{I}$
with high probability.
\begin{thm}[Inference]
\label{thm:ConfidenceSet}  Suppose $\Phi_{n}^{\alpha}$ satisfies
Assumption \ref{assu:Convex confidence set for CCP} and $\widehat{\Theta}_{I}^{\alpha}$
is constructed as (\ref{eq:confidence set expression}).
\begin{enumerate}
\item \label{enu:confidence set 1}$\lim\inf_{n\to\infty}\text{Pr}\left(\Theta_{I}\subseteq\widehat{\Theta}_{I}^{\alpha}\right)\geq1-\alpha$.
\item \label{enu:confidence set 2}For each $\theta$, determining $\theta\in\widehat{\Theta}_{I}^{\alpha}$
solves a convex program.
\end{enumerate}
\end{thm}
Theorem \ref{thm:ConfidenceSet}.\ref{enu:confidence set 1} follows
directly from (\ref{eq:confidence set expression}) and the assumption
on $\Phi_{n}^{\alpha}$. To understand Theorem \ref{thm:ConfidenceSet}.\ref{enu:confidence set 2},
note that $\theta\in\widehat{\Theta}_{I}^{\alpha}$ if and only if,
for all $x\in\mathcal{X}$, there exist $\sigma^{x}:\mathcal{E}\times\mathcal{T}\to\Delta\left(\mathcal{A}\right)$
and $\phi^{x}\in\Delta\left(\mathcal{A}\right)$ such that (\ref{eq:LP IC}),
(\ref{eq:LP consistency}), and $\phi\in\Phi_{n}^{\alpha}$ are satisfied.
Compared to the population program described in Theorem \ref{thm:Identified set LP},
which treated $\phi$ as known constants, we make $\phi$ part of
the optimization variables and impose convex constraints $\phi\in\Phi_{n}^{\alpha}$.
Since all equality constraints are linear in $\left(\sigma,\phi\right)$
and inequality constraints are convex in $\left(\sigma,\phi\right)$,
the feasibility program is convex (see \citet{boyd2004convexoptimization}).
Note that the computational tractability comes from the fact that
$\phi$ enters the restrictions in Theorem \ref{thm:Identified set LP}
in an additively separable manner; letting $\phi$ be part of the
optimization variable does not disrupt the linearity of the constraints
with respect to the variables of optimization.

Finally, we note that computation can be made faster by constructing
$\Phi_{n}^{\alpha}$ as linear constraints since then $\theta\in\widehat{\Theta}_{I}^{\alpha}$
can be determined via a linear program. In our empirical application,
we construct $\Phi_{n}^{\alpha}$ as simultaneous confidence intervals
for the multinomial proportion parameters $\phi$ using the results
in \citet{fitzpatrick1987quicksimultaneous}.\footnote{See Appendix \ref{subsec:Construction-of-Convex} for details. We
also provide Monte Carlo evidence that the proposed method has desirable
coverage probabilities even when $\mathcal{X}$ has many elements.}

\subsection{Implementation\label{subsec:Implementation}}

We propose a practical routine for obtaining the confidence set $\widehat{\Theta}_{I}^{\alpha}$.
Theorem \ref{thm:ConfidenceSet} says that for any candidate $\theta$,
we can determine whether $\theta\in\widehat{\Theta}_{I}^{\alpha}$
by solving a convex (feasibility) program. This feature is attractive,
but it only provides us a binary answer (``yes'' or ``no'').

As commonly done in existing works on partially identified game-theoretic
models (e.g., \citet{ciliberto2009marketstructure}, \citet*{syrgkanis2021inference},
\citet{magnolfi2021estimation}), we define a non-negative criterion
function $\widehat{Q}_{n}^{\alpha}\left(\theta\right)\geq0$ with
the property that $\widehat{Q}_{n}^{\alpha}\left(\theta\right)=0$
if and only if $\theta\in\widehat{\Theta}_{I}^{\alpha}$. The value
of $\widehat{Q}_{n}^{\alpha}\left(\theta\right)$ for each $\theta$
can be obtained by solving a convex program. The advantage of using
a criterion function is that the value of $\widehat{Q}_{n}^{\alpha}\left(\theta\right)$
gives us information on the distance between $\theta$ and the identified
set. Moreover, the gradients of the criterion functions provide information
on which directions to descend in order to spot a local minimum.

Let $\left\{ w^{x}\right\} _{x\in\mathcal{X}}$ be the set of strictly
positive weights for each bin $x\in\mathcal{X}$. The choice of weights
can be arbitrary although we will choose values proportional to the
number of observations at each bin $x$. Let $q^{x}\in\mathbb{R}$
and $q\equiv\left(q^{x}\right)_{x\in\mathcal{X}}$. Let $\widehat{Q}_{n}^{\alpha}\left(\theta\right)$
be the value of the following convex program. 
\begin{gather}
\min_{q,\sigma,\phi}\sum_{x\in\mathcal{X}}w^{x}q^{x}\quad\text{subject to }\label{eq:program Qhat}\\
\sum_{\varepsilon,t_{-i}}\psi_{\varepsilon}^{x,\theta}\pi_{t\vert\varepsilon}^{x}\sigma_{a\vert\varepsilon,t}^{x}\partial u_{i}^{x,\theta}\left(\tilde{a}_{i},a,\varepsilon_{i}\right)\leq q^{x},\quad\forall i,x,t_{i},a,\tilde{a}_{i}\nonumber \\
\phi_{a}^{x}=\sum_{\varepsilon,t}\psi_{\varepsilon}^{x,\theta}\pi_{t\vert\varepsilon}^{x}\sigma_{a\vert\varepsilon,t}^{x},\quad\forall a,x\nonumber \\
q^{x}\geq0,\ \sigma^{x}\in\Delta_{a\vert\varepsilon,t},\ \phi^{x}\in\Delta_{a},\quad\forall x\nonumber \\
\phi\in\Phi_{n}^{\alpha}.\nonumber 
\end{gather}

Intuitively, $q^{x}\geq0$ measures the minimal violation of the inequalities
necessary at bin $x$; when all equilibrium conditions can be satisfied,
the solver will drive the value of $q^{x}$ to zero.\footnote{This formulation uses the fact that $\max\left\{ z_{1},...,z_{K}\right\} $
can be obtained by solving $\min t$ subject to $z_{k}\leq t$ for
$k=1,...,K$.} The solution to (\ref{eq:program Qhat}) measures the weighted average
of the minimal violations of the equilibrium conditions required to
make $\theta$ compatible with data. Also note that the choice of
weights do not affect the results if the researcher is only interested
in the set of $\theta$'s whose criterion function values are exactly
zero.

The following summarizes the properties of the criterion function
approach.
\begin{thm}[Implementation]
\label{thm:Implementation} 
\begin{enumerate}
\item \label{enu:implementation 1}For any $\theta\in\Theta$, program (\ref{eq:program Qhat})
is feasible and convex.
\item \label{enu:implementation 2}$\widehat{Q}_{n}^{\alpha}\left(\theta\right)=0$
if and only if $\theta\in\widehat{\Theta}_{I}^{\alpha}$.
\item \label{enu:implementation 3}If the gradient $\nabla\widehat{Q}_{n}^{\alpha}\left(\theta\right)$
exists at $\theta$, it can be obtained as a byproduct to program
(\ref{eq:program Qhat}) via the envelope theorem.
\end{enumerate}
\end{thm}
In particular, Theorem \ref{thm:Implementation}.\ref{enu:implementation 3}
says that, due to the envelope theorem, we can obtain the gradients
for free when we evaluate the criterion function at each point (assuming
the analytic derivatives of $\psi^{x,\theta}$ and $u_{i}^{x,\theta}$
are available). In practice, we need to identify the minimizers of
$\widehat{Q}_{n}^{\alpha}\left(\theta\right)$ in order to numerically
approximate $\widehat{\Theta}_{I}^{\alpha}$. However, doing so by
conducting an extensive grid search over the whole parameter space
can be computationally costly especially when the dimension of $\theta$
is high. Due to Theorem \ref{thm:Implementation}.\ref{enu:implementation 3},
one can use gradient-based optimization algorithms to identify a minimizer
of the criterion function.\footnote{When program (\ref{eq:program Qhat}) has a manageable number of variables,
then the nested minimization problem $\min_{\theta}\widehat{Q}_{n}^{\alpha}\left(\theta\right)$
can be solved more efficiently as a single joint minimization problem
using a large-scale nonlinear solver \citep{su2012constrained}. We
use this approach for our empirical application in the next section.} The ability to quickly identify $\arg\min_{\theta}\widehat{Q}_{n}^{\alpha}\left(\theta\right)$
is advantageous since we can quickly test whether the identified set
is empty, or restrict the search to points near the minimizer.

For our empirical application, we use a heuristic approach to approximate
$\widehat{\Theta}_{I}^{\alpha}$. The idea is to identify a minimizer
of the criterion function and run a random walk process starting from
the minimizer in order to collect nearby points that have zero criterion
function values. This way we avoid the need to evaluate points that
are far from the identified set. See Appendix \ref{subsec:Random-Walk-Surface}
for details.

\section{Empirical Application: Entry Game by McDonald's and Burger King in
the US\label{sec:Empirical-Application}}

We apply our framework to study the entry game by McDonald's and Burger
King in the US using rich datasets. Entry competition in the fast
food industry fits our framework well due to two stylized facts. First,
the decisions on whether or not to operate outlets are highly persistent,
indicating that the firms' decisions are publicly observed. Tables
\ref{tab:Three-year transition probability of decisions} and \ref{tab:Transition-probability}
report the three-year transition probability of the firms' decisions
and the market outcomes $\left(a_{MD},a_{BK}\right)$ (where $a_{i}=1$
if firm $i$ is present in the market and $a_{i}=0$ otherwise), measured
for all urban census tracts (which correspond to our definition of
markets) in the contiguous US over 1997-2019. For instance, the probability
that McDonald's has an outlet in operation in a local market three
years later conditional on it having an outlet in operation today
is 0.95. Together with the assumption that the costs of revising decisions
are sufficiently low, the evidence supports the claim that firms'
decisions are best-responses to opponents' decisions that are readily
observed.\footnote{In the model, we assume that the costs of revising actions are zero.
We discuss the validity of the assumption in this setting in Appendix
\ref{sec:Adjustment-Costs-in}. }

\begin{table}[h]
{\small
\centering

\caption{\label{tab:Three-year transition probability of decisions}Three-year
Transition Probability of Decisions}

\begin{threeparttable}

{\renewcommand{\arraystretch}{1.2}
\begin{centering}
\begin{tabular}{c|cccc|cc}
\multicolumn{3}{c}{McDonald's} &  & \multicolumn{3}{c}{Burger King}\tabularnewline
$t\backslash t+3$ & Out & In &  & $t\backslash t+3$ & Out & In\tabularnewline
\cline{1-3} \cline{2-3} \cline{3-3} \cline{5-7} \cline{6-7} \cline{7-7} 
Out & 0.98 & 0.02 &  & Out & 0.99 & 0.01\tabularnewline
In & 0.05 & 0.95 &  & In & 0.08 & 0.92\tabularnewline
\end{tabular}
\par\end{centering}
}

\begin{tablenotes} \scriptsize
\item \emph{Notes:} Measured for urban tracts in the contiguous US, 1997-2019. 
\end{tablenotes}
\end{threeparttable}

}
\end{table}

\begin{table}[h]
{\small
\centering

\caption{\label{tab:Transition-probability}Three-year Transition Probability
of Market Outcomes $\left(a_{MD},a_{BK}\right)$}

\begin{threeparttable}

{\renewcommand{\arraystretch}{1.2}
\begin{centering}
\begin{tabular}{c|cccc}
$t$\textbackslash$t+3$ & $\left(0,0\right)$ & $\left(0,1\right)$ & $\left(1,0\right)$ & $\left(1,1\right)$\tabularnewline
\hline 
$\left(0,0\right)$ & 0.97 & 0.01 & 0.02 & 0.00\tabularnewline
$\left(0,1\right)$ & 0.09 & 0.87 & 0.00 & 0.04\tabularnewline
$\left(1,0\right)$ & 0.06 & 0.00 & 0.92 & 0.02\tabularnewline
$\left(1,1\right)$ & 0.00 & 0.04 & 0.08 & 0.88\tabularnewline
\end{tabular}
\par\end{centering}
}

\begin{tablenotes} \scriptsize
\item \emph{Notes:} Measured for urban tracts in the contiguous US, 1997-2019. 
\end{tablenotes}
\end{threeparttable}

}
\end{table}

Second, information asymmetries and information spillover from observing
others' decisions are common features in the industry. It is well-documented
that competitors take extra scrutiny over the locations where McDonald's
opens new outlets in order to take advantage of McDonald's leading
market research technology.\footnote{See \citet{ridley2008herding} and \citet{yang2020learning} who provide
anecdotal evidence on how competing firms learn about the profitability
of a location from entries of leading firms such as McDonald's and
Starbucks. For example, according to \emph{The Wall Street Journal},
``In the past, many restaurants... plopped themselves next to a McDonald's
to piggyback on the No. 1 burger chain's market research.'' \citep{leung2003aglutted}} Our notion of equilibrium accounts for this phenomenon.

Using the proposed framework, we estimate the entry game under different
baseline information structures in order to explore the role of informational
assumptions on identification. We also compare our results to those
obtained under Bayes correlated equilibrium, which also allows estimation
with weak assumptions on players' information. We then perform a counterfactual
policy exercise that studies how the market structures in Mississippi
food deserts respond after increasing access to healthy food.

\subsection{Data Description}

We combine multiple datasets to construct the final dataset for structural
estimation of the entry game. In the final dataset, the unit of observation
is a market (urban tract). Each observation contains information on
the firms' market entry decisions and the observable characteristics
of the firms and the market.

Our primary dataset comes from Data Axle Historical Business Database,
which contains a (approximately) complete list of fast-food chain
outlets operating in the US between 1997 and 2019 at an annual level.\footnote{This database contains location information for a detailed list of
business establishments in the US from 1997 to 2019. The provider
attempts to increase accuracy by using an internal verification procedure
after collecting data from multiple sources. The dataset is approximately
complete in the sense that the list is not free of error. However,
we compare the number of burger outlets in the data and the number
reported in external sources and confirm that the information is highly
accurate for the case of burger chains. See Appendix \ref{sec:Data-Appendix}
for details.} The advantage of this dataset is that it provides the address information
of the burger outlets across all regions of the US. The use of this
dataset to study strategic entry decisions is new.\footnote{We are not the first to study the entry game between McDonald's and
Burger King in the US. \citet{gayle2015choosing} uses 2011 cross-sectional
data hand-collected using the online restaurant locator on the brands'
websites. However, they define a local market as an ``isolated city''
that is more than 10 miles away from the closest neighboring city,
which is larger than our definition that uses a census tract. Moreover,
they focus on examining assumptions on the order of entries.}

Although we use panel data to investigate the persistence of decisions
over time, we use cross-section data to estimate the structural model.
The idea is to illustrate that the econometrician can use cross-sectional
data as a snapshot of the stable outcomes of the markets at some point
in time.\footnote{If we wanted to exploit the information available in panel data, we
would need to model the dependence of observations across time. However,
given that market environments usually seem to stay very stable over
time, it is not clear how to leverage the information for structural
estimation. For simplicity, we focus on analyzing a single cross-section
(which also represents a typical dataset available to researchers).} We use the 2010 cross-section since it was the last year for which
decennial census data were available. We describe the main features
of our dataset below. Further details on data construction are provided
in Appendix \ref{sec:Data-Appendix}. 

\subsubsection*{Market Definition}

Markets are defined as 2010 urban census tracts in the contiguous
US. A census tract is classified as urban if its geographic centroid
is in an \emph{urbanized area} defined by the Census. The final data
contain 54,944 markets. We code $a_{i}=1$ if firm $i$ had an outlet
operating in the market.\footnote{McDonald's (resp. Burger King) has more than one outlets in 1.5\%
(resp. 0.3\%) of the markets.} The unconditional probabilities of market outcomes are $\left(\hat{\phi}_{00},\hat{\phi}_{01},\hat{\phi}_{10},\hat{\phi}_{11}\right)=\left(0.74,0.06,0.15,0.05\right)$
where $\hat{\phi}_{a}$ is the sample frequency of outcome $a=\left(a_{MD},a_{BK}\right)$.

\subsubsection*{Exclusion Restrictions}

We use two firm-specific variables that have been used in existing
works: distance to headquarters (\citet{zhu2009marketstructure},
\citet{zhu2009spatial}, \citet{yang2012burgerking}) and own outlets
in nearby markets (\citet{toivanen2005marketstructure}, \citet{igami2016unobserved},
\citet{yang2020learning}). Variable distance to headquarter measures
the distance between the center of each market to the firms' respective
headquarters. The associated exclusion restriction is valid if the
cost of operating an outlet increases with its distance to own headquarter,
but is unrelated to the distance to opponents' headquarters. Variable
own outlets in neighboring markets is constructed by finding all outlets
in tracts that are adjacent to a given tract. The underlying assumption
is that an outlet's profit can be affected by an own-brand outlet
in a neighboring market, but not by a competing brand's outlet in
a neighboring market; competition with opponents occur only within
each market.

\subsubsection*{Summary Statistics}

Summary statistics are provided in Table \ref{tab:Summary-Statistics}.
Continuous variables are discretized to binary variables by using
cutoffs around their medians. Clearly, the entry probability of McDonald's
is higher. McDonald's is more likely to have an outlet present in
adjacent markets. The distance to headquarter is higher for Burger
King on average because Burger King has its headquarter in Florida
while McDonald's has its headquarter in Chicago.

\begin{table}[h]
{\footnotesize

\caption{\label{tab:Summary-Statistics}Summary Statistics}

\centering
\begin{threeparttable}
\begin{centering}
{
\def\sym#1{\ifmmode^{#1}\else\(^{#1}\)\fi}
\begin{tabular}{l*{1}{ccccc}}
\toprule
                    &\multicolumn{5}{c}{}                                            \\
                    &        Mean&     Std dev&         Min&         Max&           N\\
\midrule
\textbf{\emph{Decision variables}}&            &            &            &            &            \\
MD Entry            &       0.196&       0.397&        0.00&        1.00&       54940\\
BK Entry            &       0.106&       0.307&        0.00&        1.00&       54940\\
\textbf{\emph{Firm-specific variables}}&            &            &            &            &            \\
MD outlets present in nearby markets&       0.720&       0.449&        0.00&        1.00&       54940\\
BK outlets present in nearby markets&       0.483&       0.500&        0.00&        1.00&       54940\\
Long distance to MD HQ (>1.6K km)&       0.285&       0.451&        0.00&        1.00&       54940\\
Long distance to BK HQ (>1.6K km)&       0.712&       0.453&        0.00&        1.00&       54940\\
\textbf{\emph{Market environment variables}}&            &            &            &            &            \\
Many eating/drinking places (>7 stores)&       0.465&       0.499&        0.00&        1.00&       54940\\
High income per capita (>25K dollars)&       0.502&       0.500&        0.00&        1.00&       54940\\
Low access to healthy food&       0.856&       0.351&        0.00&        1.00&       54940\\
Food desert         &       0.334&       0.472&        0.00&        1.00&       54940\\
\bottomrule
\end{tabular}
}

\par\end{centering}
\begin{tablenotes} \scriptsize
\emph{Notes:} All variables are binary. Each observation corresponds to urban census tracts.
\end{tablenotes}
\end{threeparttable}

}
\end{table}

Market environment variables control for the determinants of profitability
that are common across firms. We obtain the following variables to
describe market environments. First, we have an indicator for whether
a tract has many eating or drinking places; the variable is obtained
from the National Neighborhood Data Archive (NaNDA) which provides
business activity information at the tract-level. Second, we have
an indicator for whether a tract has high income per capita; the variable
is from the census. Finally, from the Food Access Research Atlas,
we obtain indicators for whether a tract has low access to healthy
food and whether a tract is classified as a food desert. A tract is
classified as having low access to healthy food if at least 500 or
33 percent of the population lives more than 1/2 mile from the nearest
supermarket, supercenter, or large grocery store.\footnote{ USDA use supermarkets, supercenters, and large grocery stores that
offer a full range of food products\textemdash including fresh meat
and poultry, produce, dairy, dry and packaged foods, and frozen foods\textemdash to
calculate access to healthy food. To construct the list of stores,
USDA combines a list of stores authorized to accept Supplemental Nutrition
Assistance Program (SNAP) benefits and a list of stores from Trade
Dimensions TDLinx (see \citet{ver2012access}). The list of stores
flagged as healthy food providers by the USDA serves as a proxy for
access to healthy and affordable food but does not count other retailers
that might offer healthy options (e.g., convenience stores, drugstores,
dollar stores, military commissaries, and warehouse club stores).} A tract is classified as a food desert if it has low income and low
access to healthy food, where the criteria for low-income are from
the U.S. Department of Treasury's New Markets Tax Credit program.

The last rows of Table \ref{tab:Summary-Statistics} shows that 85\%
of all urban census tracts are classified as having low access to
healthy food and 33\% are classified as food deserts. In the counterfactual
analysis, we select food deserts in Mississippi and investigate the
impact of increasing access to healthy food on the strategic entry
decisions of the firms.

\subsection{Preliminary Analysis}

Before estimating the structural model, we examine the data patterns
using simple probit regressions. Each market $m$ contains binary
decisions of each firm $a_{im}\in\left\{ 0,1\right\} $ where $a_{im}=0$
if firm $i$ stays out in market $m$ and $a_{im}=1$ if $i$ stays
in. We pool the decisions of the firms in each market (so that the
unit of observation is $\left(i,m\right)$) and regress the binary
decisions on market characteristics. Table \ref{tab:Average-Marginal-Effects}
reports the average marginal effects computed from the regression
results.

\begin{table}[h]
{\scriptsize
\centering

\caption{\label{tab:Average-Marginal-Effects}Average Marginal Effects from
Simple Probit Models}

\begin{threeparttable}
\begin{centering}
\begin{tabular}{l*{3}{c}}
\toprule
                    &\multicolumn{1}{c}{(1)}&\multicolumn{1}{c}{(2)}&\multicolumn{1}{c}{(3)}\\
                    &\multicolumn{1}{c}{In}&\multicolumn{1}{c}{In}&\multicolumn{1}{c}{In}\\
\midrule
Own-brand outlets present in nearby markets&      -0.067&      -0.076&      -0.096\\
                    &     (0.002)&     (0.002)&     (0.002)\\
\addlinespace
Long distance to HQ (> 1.6K km)&      -0.083&      -0.083&      -0.010\\
                    &     (0.003)&     (0.003)&     (0.003)\\
\addlinespace
Many eating/drinking places (>7)&            &       0.203&       0.203\\
                    &            &     (0.002)&     (0.002)\\
\addlinespace
High income per capita (>25K dollars)&            &      -0.038&      -0.037\\
                    &            &     (0.002)&     (0.002)\\
\addlinespace
Low access to healthy food&            &       0.039&       0.041\\
                    &            &     (0.004)&     (0.004)\\
\addlinespace
McDonald's          &            &            &       0.109\\
                    &            &            &     (0.002)\\
\midrule
State Dummies       &         Yes&         Yes&         Yes\\
N                   &     107,042&     107,042&     107,042\\
\bottomrule
\end{tabular}

\par\end{centering}
\begin{tablenotes} \scriptsize
\item \emph{Notes:} Each observation corresponds to a firm-market pair. Standard errors, which are given in the parentheses, are clustered at the market-level. All variables are binary. 
\end{tablenotes}
\end{threeparttable}

}
\end{table}

Table \ref{tab:Average-Marginal-Effects} conveys three messages.
First, the presence of own outlets in neighboring markets and distance
to headquarter are negatively correlated with entry decisions. This
appears to be consistent with our prior that these variables have
a negative impact on potential profits. Second, the number of eating
and drinking places strongly affects the burger chains' entries. This
is presumably because districts with high concentration of food services
are also places with high traffic of people who eat out. Finally,
low access to healthy food is \emph{positively} correlated with entry
decisions. That is, the burger chains are more likely to enter a market
when there are fewer healthy substitutes for food.

While Table \ref{tab:Average-Marginal-Effects} provides a helpful
snapshot for what drives the chains' entry decisions, the estimates
are likely to be biased since they ignore the fact that firms' decisions
affect each other. Such consideration is crucial not only for estimating
the parameters of the model but also for studying a policy experiment.
In the next section, we estimate the entry game using Bayes stable
equilibrium as a solution concept.

\subsection{Entry Game Setup}

We posit a canonical entry game that extends the running example to
incorporate covariates in the payoff functions. Let us recall the
notation. We use $i=1,2$ to denote McDonald's and Burger King respectively.
In each market $m$, firm $i$ can choose a binary action $a_{im}\in\left\{ 0,1\right\} $
where $a_{im}=1$ if $i$ stays in and $a_{im}=0$ if $i$ stays out.
The payoff function is specified as
\[
u_{i}^{x_{m},\theta}\left(a_{im},a_{jm},\varepsilon_{im}\right)=a_{im}\left(\beta_{i}^{T}x_{im}+\kappa_{i}a_{jm}+\varepsilon_{im}\right).
\]
That is, the payoff from operating in the market is $\beta_{i}^{T}x_{im}+\kappa_{i}a_{jm}+\varepsilon_{im}$
where $x_{im}$ represents market covariates, $a_{jm}$ represents
whether the opponent is present, and $\varepsilon_{im}$ is firm
$i$'s payoff shock that is not observed by the econometrician. Each
$\varepsilon_{im}$ can include firm-specific payoff determinants
(e.g., customers' loyalty to brand and managerial ability) and common
payoff determinants (e.g., local food preference, local price level,
the degree of competition from other restaurants). We model $\left(\varepsilon_{1m},\varepsilon_{2m}\right)\in\mathbb{R}^{2}$
as being normally distributed with zero mean, unit variance, and correlation
coefficient $\rho\in[0,1)$.\footnote{For example, suppose that $\varepsilon_{im}\equiv\nu_{im}+\xi_{m}$
where $\nu_{im}\overset{iid}{\sim}N\left(0,\sigma_{\nu}^{2}\right)$
for $i=1,2$, and $\xi_{m}\overset{iid}{\sim}N\left(0,\sigma_{\xi}^{2}\right)$.
Then $\text{Var}\left(\varepsilon_{im}\right)=\sigma_{\nu}^{2}+\sigma_{\xi}^{2}$,
$\text{Cov}\left(\varepsilon_{1m},\varepsilon_{2m}\right)=\sigma_{\xi}^{2}$,
and $\text{Corr}\left(\varepsilon_{1m},\varepsilon_{2m}\right)=\sigma_{\xi}^{2}/\left(\sigma_{\nu}^{2}+\sigma_{\xi}^{2}\right)$.
Normalizing the variance of $\varepsilon_{im}$ to one scales the
coefficients $\beta_{i}$ and $\kappa_{i}$ to units equal to the
standard deviation of $\varepsilon_{im}$. Our approach of modeling
$\varepsilon_{im}$'s as jointly normally distributed with arbitrary
correlation follows \citet{magnolfi2021estimation}. \citet{ciliberto2009marketstructure}
models each $\varepsilon_{im}$ as a sum of independent firm-specific
and market-specific random shocks and estimates the associated covariance
matrix.} The payoff from staying out is normalized to zero. Our specification
of the payoff functions is quite standard in the literature.\footnote{A more flexible specification might add a richer set of covariates
or let the spillover effects $\kappa_{i}$ be a function of the observable
covariates as done in \citet{ciliberto2009marketstructure}. We keep
the specification parsimonious.}

We estimate the parameters under the baseline information assumptions
specified previously in Example \ref{exa:1}: $S^{null}$, $S^{1P}$,
and $S^{private}$. To recap, $S^{null}$ is the information structure
in which each player observes nothing; in $S^{1P}$, Player 1 observes
(only) $\varepsilon_{1}$ whereas Player 2 observes nothing; in $S^{private}$,
Player 1 observes $\varepsilon_{1}$ and Player 2 observes $\varepsilon_{2}$.

Under the Bayes stable equilibrium assumption, the baseline information
structures should be interpreted as specifying what the players \emph{minimally}
observe. Then estimating the model with $S^{null}$ as the baseline
information structure amounts to making no assumption on players'
information. On the other hand, if the baseline information structure
is set to $S^{private}$, then the identified set is robust to all
cases in which the players observe at least their payoff shocks. Finally,
setting the baseline information structure to $S^{1P}$ amounts to
assuming that McDonald's has good information about its payoff shocks
whereas Burger King might minimally have no information about its
payoff shock. This assumption relaxes the standard assumption on information
(namely the information structure is fixed at either $S^{private}$
or $S^{complete}$) and is consistent with the anecdotal evidence
that McDonald's is a leader in the market research technology.

\subsection{Estimation Results}

In order to keep the model parsimonious and reduce the computational
burden, we take some steps before estimation, which are described
as follows (see Appendix \ref{sec:Computational-Details} for further
details). First, we assume that the coefficients for common market-level
variables (eating places, income per capita, and low access to healthy
food) are identical across the two players.\footnote{This assumption is not without loss and can be refuted on the basis
that each chain might react differently to market environment. However,
we believe it is reasonable given that McDonald's and Burger King
are close substitutes to each other.} We also assume that the coefficients of the firm-specific variables
(distance to headquarter and the presence of own-brand outlets in
nearby markets) are non-positive. Second, while the benchmark distribution
of the latent variables $\left(\varepsilon_{1m},\varepsilon_{2m}\right)$
is continuous, we use discretized normal distribution for feasible
estimation. Third, we discretize each variable to binary bins; since
there are 7 variables in the covariates, this gives $2^{7}=128$ discrete
covariate bins. Conditional choice probabilities are non-parametrically
estimated using the observations within each bin. Fourth, to construct
confidence sets for the conditional choice probabilities, we used
simultaneous confidence bands based on the method described in \citet{fitzpatrick1987quicksimultaneous};
using simultaneous confidence bands makes the evaluation of the criterion
function a linear program.

\subsubsection{The Role of Informational Assumptions on Identification}

\begin{table}[h]
{\footnotesize
\centering

\caption{\label{tab:Estimation-Results} Bayes Stable Equilibrium Identified
Sets}

\begin{threeparttable}
\begin{centering}
\begin{tabular}{llll} \toprule 
 Baseline Information &  $S^{null}$ &  $S^{1P}$ & $S^{private}$    \\ [1ex] \hline 

 McDonald's Variables &  &  &  \\   \quad Spillover Effects & $[-1.83,1.62]$ & $[-0.89,-0.14]$ & \multicolumn{1}{c}{-} \\   \quad Constant & $[-1.64,0.32]$ & $[-1.46,-1.04]$ & \multicolumn{1}{c}{-} \\   \quad Nearby Outlets & $[-1.24,-0.00]$ & $[-0.56,-0.25]$ & \multicolumn{1}{c}{-} \\   \quad Distance to HQ & $[-1.23,-0.00]$ & $[-0.26,-0.00]$ & \multicolumn{1}{c}{-} \\   Burger King Variables &  &  &  \\   \quad Spillover Effects & $[-1.81,1.22]$ & $[-1.19,-0.25]$ & \multicolumn{1}{c}{-} \\   \quad Constant & $[-2.38,0.44]$ & $[-1.48,-0.76]$ & \multicolumn{1}{c}{-} \\   \quad Nearby Outlets & $[-1.44,-0.00]$ & $[-0.53,-0.00]$ & \multicolumn{1}{c}{-} \\   \quad Distance to HQ & $[-1.41,-0.00]$ & $[-0.52,-0.00]$ & \multicolumn{1}{c}{-} \\   Common Market-level Variables &  &  &  \\   \quad Eating Places & $[-0.31,1.87]$ & $[0.82,1.21]$ & \multicolumn{1}{c}{-} \\   \quad Income Per Capita & $[-1.02,0.75]$ & $[-0.54,-0.18]$ & \multicolumn{1}{c}{-} \\   \quad Low Access & $[-0.71,1.31]$ & $[0.25,0.54]$ & \multicolumn{1}{c}{-} \\   Correlation parameter $ \rho $  & $[0.00,0.99]$ & $[0.42,0.91]$ & \multicolumn{1}{c}{-} \\\hline   Number of Markets & 54940 & 54940 & 54940 \\

\bottomrule
\end{tabular}
\par\end{centering}
\begin{tablenotes}
\emph{Notes}: Table reports the projections of confidence sets obtained with nominal level $\alpha=0.05$. The identified set for $S^{private}$ not reported because it is empty. 
\end{tablenotes}
\end{threeparttable}

}
\end{table}

Table \ref{tab:Estimation-Results} reports projections of the 95\%
confidence sets obtained under the Bayes stable equilibrium assumption
with different baseline information structures. There are three main
findings related to the role of informational assumption. First, making
no assumption on players' information leads to an uninformative identified
set. The confidence set under $S^{null}$ is quite large, and we cannot
determine the signs of the parameters. Therefore, being utterly agnostic
about players' information does not give us enough identifying power
to draw meaningful conclusions.

Second, standard assumptions on information may be too strong. It
is quite standard to assume that each player $i$ observes (exactly)
$\varepsilon_{i}$ or $\left(\varepsilon_{i},\varepsilon_{-i}\right)$.
Setting baseline information structure as $S^{private}$ nests all
these cases. However, we find that the identified set under $S^{private}$
is empty, suggesting the possibility of misspecification.\footnote{Specifically, we consistently find that the minimum of the criterion
function under $S^{private}$ is strictly greater than zero. This
is also true even if we do not use sign constraints or reduce the
nominal level to a very low level (e.g., $\alpha=0.0001$).} Thus, assuming that each player observes at least their $\varepsilon_{i}$
may be too strong. Since the Bayes stable equilibrium identified set
under $S^{private}$ is equivalent to the pure strategy Nash equilibrium
identified set (see Theorem \ref{thm:Relationship between identified sets}.\ref{enu:Relationship between identified sets 2}),
the pure strategy Nash equilibrium assumption would also be rejected.\footnote{The emptiness of the identified set is not driven by the possibility
of non-existence of Bayes stable equilibrium. When the competition
effects parameters have the same signs, there exists at least one
pure strategy Nash equilibrium at each state, implying the existence
of a Bayes stable equilibrium. Of course, the emptiness of the identified
set might be due to misspecification in payoff functions, distribution
of errors, etc. Our statements are conditional on having these specifications
correct.}

Third, we find that setting the baseline information structure to
$S^{1P}$ can produce an informative identified set. Recall that the
identified set under $S^{1P}$ makes the assumption that McDonald's
has accurate information about its payoff shock, but Burger King's
information can be arbitrary. This assumption is consistent with the
anecdotal evidence that McDonald's has superior information on the
potential profitability of each market, and Burger King tries to free-ride
on McDonald's information by observing what McDonald's does.\footnote{ Note that Burger King's extraction of McDonald's information is
feasible when the errors are correlated. For instance, suppose that
local food taste variable $\xi$ enters both $\varepsilon_{1}$ and
$\varepsilon_{2}$ and that McDonald's observes $\varepsilon_{1}$
via its research technology. In a rational expectations equilibrium,
McDonald's decision reveals partial information about $\xi$, which
in turn Burger King can use to infer its $\varepsilon_{2}$. Such
refinement of information is not allowed in the static Bayes Nash
equilibrium framework.} Table \ref{tab:Estimation-Results} shows that, even if we substantially
relax the assumption on Burger King's information, we can determine
the signs of the most parameters. For example, we can see that burger
chains are more likely to enter in markets that have low access to
healthy food. We can also learn that the firms' payoff shocks are
highly correlated to each other.

In conclusion, we find that the informativeness of the identified
set crucially depends on the underlying assumption on players' information.
At least in our empirical application, it is difficult to draw a meaningful
economic conclusion without making assumptions on players' information.
On the other hand, under the maintained solution concept, the model
rejects the popular assumptions made in the literature, namely that
each firm $i$ observes at least its $\varepsilon_{i}$. A credible
intermediate case $S^{1P}$, which is consistent with our knowledge
of the market research technology in the fast food industry, delivers
strong identifying power.\footnote{\label{fn:R2MajC5}Following a reviewer's comment, we have also tried
an alternative specification that assumes $\varepsilon_{im}=\nu_{im}+\xi_{m}$
where $\nu_{im}\overset{iid}{\sim}N\left(0,\sigma_{\nu}^{2}\right)$
and $\xi_{m}\overset{iid}{\sim}N\left(0,\sigma_{\xi}^{2}\right)$.
In the estimation stage, we normalized to variance of $\varepsilon_{im}$
to one. We estimated the model under the assumption that each player
$i$ minimally observes $\nu_{im}$ (but remained agnostic to whether
firm $i$ observes $\xi_{m}$ or $\nu_{-im}$). This specification
is weaker than $S^{private}$, but it is neither stronger nor weaker
than $S^{1P}$ as McDonald's may not observe $\xi_{m}$ and Burger
King observes at least $\nu_{2m}$. We found that the identified set
is non-empty. Thus, an alternative specification can be used to relax
strong assumption on players' information. However, this approach
requires imposing an additional structure on the unobservables. Moreover,
it increases computational costs by increasing the dimensionality
of optimization problems that use discretized distributions. }

\subsubsection{Comparison to Bayes Correlated Equilibrium Identified Sets}

\begin{table}[h]
{\footnotesize
\centering

\caption{\label{tab:Estimation results with BCE}Bayes Correlated Equilibrium
Identified Sets}

\begin{threeparttable}
\begin{centering}
\begin{tabular}{llll} \toprule 
 Baseline Information &  $S^{null}$ &  $S^{1P}$ & $S^{private}$    \\ [1ex] \hline 

McDonald's Variables &  &  &  \\   \quad Spillover Effects & $[-4.83,1.92]$ & $[-4.85,-0.17]$ & $[-4.85,-2.11]$ \\   \quad Constant & $[-1.64,0.34]$ & $[-1.53,0.29]$ & $[-1.37,0.31]$ \\   \quad Nearby Outlets & $[-1.33,-0.00]$ & $[-1.11,-0.00]$ & $[-0.97,-0.00]$ \\   \quad Distance to HQ & $[-1.35,-0.00]$ & $[-1.10,-0.00]$ & $[-0.88,-0.00]$ \\   Burger King Variables &  &  &  \\   \quad Spillover Effects & $[-3.84,3.33]$ & $[-3.98,0.72]$ & $[-3.38,-1.03]$ \\   \quad Constant & $[-3.71,0.61]$ & $[-1.65,0.62]$ & $[-1.62,0.44]$ \\   \quad Nearby Outlets & $[-1.71,-0.00]$ & $[-1.23,-0.00]$ & $[-1.11,-0.00]$ \\   \quad Distance to HQ & $[-1.70,-0.00]$ & $[-1.03,-0.00]$ & $[-0.86,-0.00]$ \\   Common Market-level Variables &  &  &  \\   \quad Eating Places & $[-0.24,1.98]$ & $[0.51,1.76]$ & $[0.49,1.68]$ \\   \quad Income Per Capita & $[-1.32,0.84]$ & $[-1.16,0.14]$ & $[-1.08,0.11]$ \\   \quad Low Access & $[-0.59,1.49]$ & $[-0.37,1.31]$ & $[-0.28,1.07]$ \\   Correlation parameter $ \rho $  & $[0.00,0.99]$ & $[0.00,0.99]$ & $[0.00,0.97]$ \\\hline   BSE volume/BCE volume & 0.05036 & 0.00000 & - \\   Number of Markets & 54940 & 54940 & 54940 \\

\bottomrule
\end{tabular}
\par\end{centering}
\begin{tablenotes}
\emph{Notes}: Table reports the projections of confidence sets obtained with nominal level $\alpha=0.05$. BSE/BCE volume computed by taking products of projected intervals. 
\end{tablenotes}
\end{threeparttable}

}
\end{table}

We compare the Bayes stable equilibrium identified sets to the Bayes
correlated equilibrium identified sets studied in \citet{magnolfi2021estimation}.
The Bayes correlated equilibrium identified sets are reported in Table
\ref{tab:Estimation results with BCE}. We can readily see that the
Bayes correlated equilibrium assumption produces a much larger set
for each baseline information structure. Even when we set $S^{private}$
as the baseline information structure, it is not easy to learn the
signs of many parameters. For example, we cannot determine whether
low access to healthy food promotes or deters entries by the burger
chains.

Comparing Tables \ref{tab:Estimation-Results} and \ref{tab:Estimation results with BCE}
suggests that if the researcher is willing to accept the Bayes stable
equilibrium assumption, it can add significant identifying power while
providing the same kind of informational robustness as Bayes correlated
equilibria. At least in the context of our empirical application,
we believe it is reasonable to assume that McDonald's decisions that
we observe in the data represent best-responses to the \emph{observed}
decisions of Burger King and vice versa.

\subsection{Counterfactual Analysis: The Impact of Increasing Access to Healthy
Food on Market Structure}

Consumption of fast food is driven not only by consumers' taste for
fast food but also the availability of food substitutes in the neighborhood.
Following the recent surge of interest on the relationship between
food deserts and food consumption patterns (see, e.g., \citet{allcott2019fooddeserts}
and \citet{kolb2021retail}), we study the impact of accessibility
to healthy food on the entry decisions of fast-food chains. Specifically,
we consider a policy experiment to predict changes in market structure
in Mississippi food deserts after increasing access to healthy food
measured by supermarket entries.\footnote{Our policy experiment relies on consumers' substitution patterns between
fast-food restaurants and supermarkets (which include supercenters
and large grocery stores). Although shorter distance to providers
of healthy food does not necessarily translate to healthier diet (\citet{allcott2019fooddeserts}
and \citet{kolb2021retail}), easier access to supermarkets leads
consumers to lower their visits to fast-food chains due to reduced
travel costs and increased availability of alternative (healthy or
equally unhealthy but cheaper) food substitutes. Note that our indicator
for access to supermarket corresponds to having a supermarket within
1/2-mile distance which amounts to less than 10-minutes walking distance
while \citet{allcott2019fooddeserts} consider supermarket entries
within 10-15 minutes driving distance. } Mississippi is often called one of the ``hungriest'' states in
the US.\footnote{For example, Mississippi has been identified as the most food insecure
state in the country since 2010 according to Feeding America. See
\url{https://mississippitoday.org/2018/05/04/mississippi-still-the-hungriest-state/}.} Mississippi had 664 census tracts in 2010, and 329 of them are classified
as urban tracts, which correspond to our definition of markets. Out
of 329 urban tracts, 185 tracts (approximately 56\%) are classified
as food deserts, according to the U.S. Department of Agriculture.
According to the definition of food deserts, all of these tracts are
classified as having low access to healthy food.

We conduct a policy experiment as follows. We select the 185 tracts
classified as food deserts in Mississippi and then increase access
to healthy food. This amounts to changing the low access indicator
from one (low access) to zero (high access) in all these markets.
In reality, such policy would correspond to increasing healthy food
providers (large grocery stores, supermarkets, or supercenters) by
providing subsidies or tax breaks. We then recompute the equilibria
in these markets and report the weighted average of the bounds associated
with each measure of market structure.\footnote{Our counterfactual analysis corresponds to a partial equilibrium analysis.
We abstract away from considering how entry or exit in each market
can affect the burger chains' decisions in neighboring markets and
the responses of healthy food providers.} See Appendix \ref{subsec:Counterfactual-Analysis} for computational
details. 

We report the results of the counterfactual analysis in Table \ref{tab:-Policy-experiment:}.
The first column reports the estimates obtained from the data of the
185 markets corresponding to Mississippi food deserts. For example,
the probability of observing McDonald's enter the market in Mississippi
food deserts is 0.30, much larger than the unconditional probability
obtained using all markets, which was around 0.20.

\begin{table}
{\footnotesize
\centering

\caption{\label{tab:-Policy-experiment:} The Impact of Increasing Access to
Healthy Food in Mississippi Food Deserts}

\begin{threeparttable}

{\renewcommand{\arraystretch}{1.2}
\begin{centering}
\begin{tabular}{l l l l l l l l}  \toprule  
& & & \multicolumn{2}{c}{$BSE(S^{1P})$} & & \multicolumn{2}{c}{$BCE(S^{1P})$} \\
\cline{4-5} \cline{7-8} \addlinespace[0.1cm]
& Data & & Pre & Post & & Pre & Post
\\ [1ex] \hline 

Expected number of entrants & $ 0.47 $ &  & $[0.28,1.01]$ & $[0.15,0.79]$ &  & $[0.10,1.18]$ & $[0.03,1.17]$ \\   Probability of MD entry & $ 0.30 $ &  & $[0.11,0.32]$ & $[0.04,0.23]$ &  & $[0.00,0.71]$ & $[0.00,0.67]$ \\   Probability of BK entry & $ 0.17 $ &  & $[0.00,0.84]$ & $[0.00,0.72]$ &  & $[0.00,1.00]$ & $[0.00,1.00]$ \\   Probability of no entrant & $ 0.64 $ &  & $[0.15,0.74]$ & $[0.28,0.85]$ &  & $[0.00,0.90]$ & $[0.00,0.97]$ \\

\bottomrule
\end{tabular}
\par\end{centering}
}

\begin{tablenotes}
\emph{Notes}: Data column represents the sample estimates obtained using markets corresponding to Mississippi food deserts. Final bounds obtained by simulating equilibria at each parameter in the identified set, and then taken union over all bounds. Each number is obtained by taking a weighted average with weights proportional to the number of markets in each covariate bin.
\end{tablenotes}
\end{threeparttable}

}
\end{table}

The second and third columns report the bounds obtained before (``Pre''
has low access indicators set to one) and after the counterfactual
policy (``Post'' has low access indicators set to zero) using the
$S^{1P}$-Bayes stable equilibrium identified set. The bounds are
pretty wide because we have considered all parameters in the identified
set and made no assumption on the equilibrium selection. However,
they shift in the expected directions. For example, the bounds on
the expected number of entrants shift from $[0.28,1.01]$ to $[0.15,0.79]$.
Since the mean number of entrants in the data was $0.47$ and the
post-counterfactual bounds are $\left[0.15,0.79\right]$, the maximal
change we can expect is $0.15-0.47=-0.32$. In some cases, we can
make a stronger statement: while the unconditional probability of
observing McDonald's enter in data was 0.30, the upper bound in the
Post-regime decreases to 0.23, so we can expect that the probability
of McDonald's enter to decrease by \emph{at least} 0.07.

Our results suggest that meaningful counterfactual statements may
be made even with weak assumptions on players' information. The bounds
do not depend on specific assumptions on equilibrium selection and
admit all information structures that are expansions of the baseline
information structure.\footnote{Our predictions are conservative because we do not make any assumptions
on how the information structure or the equilibrium selection rule
might change after the counterfactual policy.} Hence our approach can also serve as a useful tool to conduct sensitivity
analysis for researchers who want to see whether their predictions
are driven by assumptions on equilibrium selection or what the players
know.

For comparison, in the last two columns, we report the counterfactual
results obtained using the $S^{1P}$-Bayes correlated equilibrium
identified set. One can readily see that the bounds are quite large
compared to the Bayes stable equilibrium counterpart. For example,
we cannot make any statement about the probability of Burger King's
entry after the counterfactual policy is implemented. Table \ref{tab:-Policy-experiment:}
shows that Bayes correlated equilibrium predictions can be too permissive,
especially when no assumption is imposed on what equilibrium might
be selected in the counterfactual world.

\section{Conclusion\label{sec:Conclusion}}

This paper presents an empirical framework for analyzing stable outcomes
with weak assumptions on players' information. We propose Bayes stable
equilibrium as a framework for analyzing stable outcomes, which appear
in various empirical settings. Our framework can be an attractive
alternative to existing methods for practitioners who want to work
with an empirical game-theoretic model and be robust to informational
assumptions. Furthermore, we believe the proposed computational algorithms
can also be helpful in similar settings, especially since reducing
computational burden remains a fundamental challenge in the literature.

We believe there are many exciting avenues for future research. First,
providing a non-cooperative foundation to our solution concepts remains
an open question. While we can imagine a dynamic adjustment process
that converges to stable outcomes, how to formalize this idea is yet
unclear. Second, it will be interesting to find reasonable ways of
imposing equilibrium selection. While Bayes stable equilibrium (or
Bayes correlated equilibrium) has the informational robustness property,
the set of predictions may be too large, limiting our ability to make
sharp predictions for counterfactual analysis. Finding ways to sharpen
predictions without sacrificing robustness to information will be
helpful. Third, our counterfactual analysis is limited to a partial
equilibrium analysis. It will be interesting to think about ways to
model the strategic interactions of healthy food providers and unhealthy
food providers together. Finally, there are other forms of informational
robustness that our model cannot handle but is empirically interesting.
For example, it might be natural to assume that McDonald's has superior
information relative to Burger King while being agnostic about their
specific information structure. Our model only specifies what the
players minimally observe and is agnostic about the relative information
across players. Studying alternative forms of informational robustness
and corresponding econometric frameworks should be interesting. 

\bibliographystyle{ecta}
\bibliography{StabilityReferences}

\newpage{}

\appendix

\part*{Appendix}

\section{Proofs\label{sec:Proofs}}

\subsection{Proof of Theorem \ref{thm:BSE and REE connection}}

Let $S^{*}$ be an expansion of $S$. Let $\delta:\mathcal{T}\times\tilde{\mathcal{T}}\to\Delta\left(\mathcal{A}\right)$
be an outcome function in $\left(G,S^{*}\right)$. We say that an
outcome function $\delta$ in $\left(G,S^{*}\right)$ \emph{induces}
a decision rule $\sigma:\mathcal{E}\times\mathcal{T}\to\Delta\left(\mathcal{A}\right)$
in $\left(G,S\right)$ if
\[
\sigma\left(a\vert\varepsilon,t\right)=\sum_{\tilde{t}}\lambda\left(\tilde{t}\vert\varepsilon,t\right)\delta\left(a\vert t,\tilde{t}\right)
\]
for each $a$ whenever $\text{Pr}\left(\varepsilon,t\right)>0$.
\begin{lem}
\label{lem:1}A decision rule $\sigma$ is a Bayes stable equilibrium
of $\left(G,S\right)$ if and only if, for some expansion $S^{*}$
of $S$, there is a rational expectations equilibrium of $\left(G,S^{*}\right)$
that induces $\sigma$.
\end{lem}
The proof of Lemma \ref{lem:1} closely follows the proof in Theorem
1 of \citet{bergemann2016bayescorrelated}. The only if $\left(\Rightarrow\right)$
direction is established by (i) letting the Bayes stable equilibrium
decision rule $\sigma$ a signal function that generates public signals
(recommendations of outcomes) for every given $\left(\varepsilon,t\right)$,
and (ii) constructing an outcome function $\delta$ as a degenerate
self-map that places unit mass on $a$ whenever $a$ is drawn from
$\sigma\left(\cdot\vert\varepsilon,t\right)$. Conversely, the if
$\left(\Leftarrow\right)$ direction is established by constructing
a decision rule by integrating out the players' signals from a given
outcome function.

\begin{proof}[Proof of Lemma \ref{lem:1}]
$\left(\Rightarrow\right)$ Suppose $\sigma$ is a Bayes stable equilibrium
of $\left(G,S\right)$. That is, 
\[
\sum_{\varepsilon,t_{-i}}\psi_{\varepsilon}\pi_{t\vert\varepsilon}\sigma_{a\vert\varepsilon,t}u_{i}\left(a,\varepsilon\right)\geq\sum_{\varepsilon,t_{-i}}\psi_{\varepsilon}\pi_{t\vert\varepsilon}\sigma_{a\vert\varepsilon,t}u_{i}\left(a_{i}',a_{-i},\varepsilon\right),\quad\forall i,t_{i},a,a_{i}'.
\]
We want to find an expansion $S^{*}$ of $S$ and a rational expectations
equilibrium outcome function $\delta$ in $\left(G,S^{*}\right)$
that induces $\sigma$. Construct an expansion $S^{*}$ of $S$ as
follows. With some abuse in notation, let $\lambda$ be a signal distribution
that generates a \emph{public} signal such that
\[
\lambda\left(\tilde{t}^{p}=a\vert\varepsilon,t\right)=\sigma\left(a\vert\varepsilon,t\right).
\]
where $\tilde{t}^{p}$ denotes a public signal.\footnote{More formally, the agents receive signals that are perfectly correlated,
i.e., $\lambda\left(\tilde{t}_{1}=a,...,\tilde{t}_{I}=a\vert\varepsilon,t\right)=\sigma\left(a\vert\varepsilon,t\right)$.} Let an outcome function be degenerate as follows:
\[
\delta\left(\tilde{a}\vert t,\tilde{t}^{p}=a\right)=\begin{cases}
1 & \text{if }\tilde{a}=a\\
0 & \text{if }\tilde{a}\neq a
\end{cases}.
\]
That is, when the players observe $\tilde{t}^{p}=a$ as a public signal,
the outcome function dictates that $a$ be played as an outcome of
the game. It remains to show that every outcome $a$ generated by
the outcome function $\delta$ is optimal to the players. The rational
expectations equilibrium condition is
\[
\sum_{\varepsilon,t_{-i}}\psi_{\varepsilon}\pi_{t\vert\varepsilon}\lambda_{\tilde{t}^{p}\vert\varepsilon,t}\delta_{\tilde{a}\vert t,\tilde{t}^{p}}u_{i}\left(\tilde{a},\varepsilon\right)\geq\sum_{\varepsilon,t_{-i}}\psi_{\varepsilon}\pi_{t\vert\varepsilon}\lambda_{\tilde{t}^{p}\vert\varepsilon,t}\delta_{\tilde{a}\vert t,\tilde{t}^{p}}u_{i}\left(\tilde{a}_{i}',\tilde{a}_{-i},\varepsilon\right),\quad\forall i,t_{i},\tilde{t}^{p},\tilde{a},\tilde{a}_{i}'
\]
But since $\lambda\left(\tilde{t}^{p}=a\vert\varepsilon,t\right)=\sigma\left(a\vert\varepsilon,t\right)$
and the inequality is trivially satisfied when $\tilde{t}^{p}\neq\tilde{a}$
(both sides become zero), the rational expectations equilibrium condition
reduces to 
\[
\sum_{\varepsilon,t_{-i}}\psi_{\varepsilon}\pi_{t\vert\varepsilon}\sigma_{a\vert\varepsilon,t}u_{i}\left(a,\varepsilon\right)\geq\sum_{\varepsilon,t_{-i}}\psi_{\varepsilon}\pi_{t\vert\varepsilon}\sigma_{a\vert\varepsilon,t}u_{i}\left(\tilde{a}_{i}',a_{-i},\varepsilon\right),\quad\forall i,t_{i},a,\tilde{a}_{i}'
\]
which holds by the assumption that $\sigma$ is a Bayes stable equilibrium
of $\left(G,S\right)$.

$\left(\Leftarrow\right)$ Suppose that $\delta$ is a rational expectations
equilibrium of $\left(G,S^{*}\right)$ and $\delta$ induces $\sigma$
in $\left(G,S\right)$. That is, we have
\[
\sum_{\varepsilon,t_{-i},\tilde{t}_{-i}}\psi_{\varepsilon}\pi_{t\vert\varepsilon}\lambda_{\tilde{t}\vert\varepsilon,t}\delta_{a\vert t,\tilde{t}}u_{i}\left(a,\varepsilon\right)\geq\sum_{\varepsilon,t_{-i},\tilde{t}_{-i}}\psi_{\varepsilon}\pi_{t\vert\varepsilon}\lambda_{\tilde{t}\vert\varepsilon,t}\delta_{a\vert t,\tilde{t}}u_{i}\left(a_{i}',a_{-i},\varepsilon\right),\quad\forall i,t_{i},\tilde{t}_{i},a,a_{i}'
\]
Integrating out $\tilde{t}_{i}$ from both sides gives
\begin{gather*}
\sum_{\varepsilon,t_{-i}}\psi_{\varepsilon}\pi_{t\vert\varepsilon}\left(\sum_{\tilde{t}}\lambda_{\tilde{t}\vert\varepsilon,t}\delta_{a\vert t,\tilde{t}}\right)u_{i}\left(a,\varepsilon\right)\geq\sum_{\varepsilon,t_{-i}}\psi_{\varepsilon}\pi_{t\vert\varepsilon}\left(\sum_{\tilde{t}}\lambda_{\tilde{t}\vert\varepsilon,t}\delta_{a\vert t,\tilde{t}}\right)u_{i}\left(a_{i}',a_{-i},\varepsilon\right),\quad\forall i,t_{i},a,a_{i}'\\
\Leftrightarrow\sum_{\varepsilon,t_{-i}}\psi_{\varepsilon}\pi_{t\vert\varepsilon}\sigma_{a\vert\varepsilon,t}u_{i}\left(a,\varepsilon\right)\geq\sum_{\varepsilon,t_{-i}}\psi_{\varepsilon}\pi_{t\vert\varepsilon}\sigma_{a\vert\varepsilon,t}u_{i}\left(a_{i}',a_{-i},\varepsilon\right),\quad\forall i,t_{i},a,a_{i}'
\end{gather*}
which is the Bayes stable equilibrium condition for $\sigma$ in $\left(G,S\right)$.
\end{proof}
The statement of the theorem then follows directly from Lemma \ref{lem:1}
because any decision rule $\sigma:\mathcal{E}\times\mathcal{T}\to\Delta\left(\mathcal{A}\right)$
in $\left(G,S\right)$ pins down the joint distribution on $\mathcal{E}\times\mathcal{T}\times\mathcal{A}$
(the prior distribution $\psi$ on $\mathcal{E}$ is fixed by $G$
and the signal distribution $\pi:\mathcal{E}\to\Delta\left(\mathcal{T}\right)$
is fixed by $S$). $\square$

\subsection{Proof of Corollary \ref{cor:CCP for BSE and REE}}

$\left(\subseteq\right)$ Take any $\phi\in\text{\ensuremath{\mathcal{P}_{a}^{BSE}\left(G,S\right)}}$.
By definition, there is a BSE $\sigma$ in $\left(G,S\right)$ that
induces $\phi$. By Theorem \ref{thm:BSE and REE connection}, there
exists an expansion $S^{*}$ of $S$ and a REE $\delta$ of $\left(G,S^{*}\right)$
that induces $\sigma$. Since $\delta$ induces $\sigma$ and $\sigma$
induces $\phi$, $\delta$ induces $\phi$. It follows that $\phi\in\bigcup_{S^{*}\succsim_{E}S}\mathcal{P}_{a}^{REE}\left(G,S^{*}\right)$.

$\left(\supseteq\right)$ Take any $\phi\in\bigcup_{S^{*}\succsim_{E}S}\mathcal{P}_{a}^{REE}\left(G,S^{*}\right)$.
By definition, there exists some $S^{*}\succsim_{E}S$ and a REE $\delta$
of $\left(G,S^{*}\right)$ such that $\delta$ induces $\phi$, (i.e.,
$\phi_{a}=\sum_{\varepsilon,t,\tilde{t}}\psi_{\varepsilon}\pi_{t\vert\varepsilon}\lambda_{\tilde{t}\vert\varepsilon,t}\delta_{a\vert t,\tilde{t}}$
for all $a\in\mathcal{A}$). Since $S^{*}\succsim_{E}S$ and $\delta$
is a REE of $\left(G,S^{*}\right)$, by Theorem \ref{thm:BSE and REE connection},
$\delta$ induces a decision rule $\sigma$ in $\left(G,S\right)$
that is a BSE of $\left(G,S\right)$. Since $\delta$ induces $\sigma$,
it follows that $\sigma$ induces $\phi$. Therefore, we have $\phi\in\text{\ensuremath{\mathcal{P}_{a}^{BSE}\left(G,S\right)}}$.
$\square$

\subsection{Proof of Corollary \ref{cor:more info leads to tighter set}}

Take $\sigma\in\mathcal{P}_{t,\varepsilon,a}^{BSE}\left(G,S\right)$.
We want to show that $\sigma\in\mathcal{P}_{\varepsilon,t,a}^{BSE}\left(G,S'\right)$.
From Theorem \ref{thm:BSE and REE connection}, we have $\mathcal{P}_{\varepsilon,t,a}^{BSE}\left(G,S\right)=\bigcup_{\tilde{S}\succsim_{E}S}\mathcal{P}_{\varepsilon,t,a}^{REE}\left(G,\tilde{S}\right)$,
so there exists some $S^{*}$ such that $\sigma\in\mathcal{P}_{\varepsilon,t,a}^{REE}\left(G,S^{*}\right)$.
But since $S^{*}\succsim_{E}S\succsim_{E}S'$, we have $\mathcal{P}_{\varepsilon,t,a}^{REE}\left(G,S^{*}\right)\subseteq\bigcup_{\tilde{S}\succsim_{E}S'}\mathcal{P}_{\varepsilon,t,a}^{REE}\left(G,\tilde{S}\right)$,
so it follows that 
\[
\sigma\in\bigcup_{\tilde{S}\succsim_{E}S'}\mathcal{P}_{\varepsilon,t,a}^{REE}\left(G,\tilde{S}\right)=\mathcal{P}_{\varepsilon,t,a}^{BSE}\left(G,S'\right)
\]
which is what we wanted. $\square$

\subsection{Proof of Theorem \ref{thm:Relationship to PSNE}}
\begin{enumerate}
\item We first prove the first statement:\\
$\left(\Rightarrow\right)$ Since $\delta$ is a REE of $\left(G,S^{complete}\right)$,
it satisfies
\[
\psi_{\varepsilon}\delta_{a\vert\varepsilon}u_{i}\left(a,\varepsilon\right)\geq\psi_{\varepsilon}\delta_{a\vert\varepsilon}u_{i}\left(a_{i}',a_{-i},\varepsilon\right),\quad\forall i,\varepsilon,a,a_{i}'.
\]
Fix any $\varepsilon^{*}\in\mathcal{E}$ such that $\psi_{\varepsilon^{*}}>0$
(with the full support assumption, we have $\psi_{\varepsilon}>0$
for all $\varepsilon$). Consider any $a^{*}\in\mathcal{A}$ such
that $\delta$ places a positive mass at $\varepsilon^{*}$, i.e.,
$\delta_{a^{*}\vert\varepsilon^{*}}>0$. Since $\psi_{\varepsilon^{*}}\delta_{a^{*}\vert\varepsilon^{*}}>0$,
the REE condition reduces to
\[
u_{i}\left(a^{*},\varepsilon^{*}\right)\geq u_{i}\left(a_{i}',a_{-i}^{*},\varepsilon^{*}\right),\quad\forall i,a_{i}'
\]
 which is exactly the PSNE condition of $a^{*}$ at state $\varepsilon^{*}$.

$\left(\Leftarrow\right)$ Suppose that $\delta:\mathcal{E}\to\Delta\left(\mathcal{A}\right)$
is constructed in a way such that $\delta_{a\vert\varepsilon}>0$
implies that $a$ is a PSNE outcome at $\varepsilon$. Since any on-path
outcome $a$ at $\varepsilon$ is a PSNE at $\varepsilon$, it immediately
follows that the outcome is optimal to each player who observes $\left(a_{i},a_{-i}\right)$
and $\varepsilon$, satisfying the REE condition. \\
The second statement follows by observing that an outcome function
and a decision rule are equivalent (i.e., $\delta\left(a\vert\varepsilon\right)\equiv\sigma\left(a\vert\varepsilon,t=\varepsilon\right)$)
when the signal distribution is degenerate ($\pi\left(t=\varepsilon\vert\varepsilon\right)=1$).
In this case, the Bayes stable equilibrium conditions reduces to the
rational expectations equilibrium conditions. $\square$
\item The first statement is proven as follows:\\
($\Leftarrow$) Let $\delta:\mathcal{E}\to\Delta\left(\mathcal{A}\right)$
be a REE of $\left(G,S^{complete}\right)$. By definition, we have
\[
\psi_{\varepsilon}\delta_{a\vert\varepsilon}u_{i}\left(a,\varepsilon_{i}\right)\geq\psi_{\varepsilon}\delta_{a\vert\varepsilon}u_{i}\left(a_{i}',a_{-i},\varepsilon_{i}\right),\quad\forall i,\varepsilon,a,a_{i}'
\]
Integrating both sides with respect to $\varepsilon_{-i}$ gives
\[
\sum_{\varepsilon_{-i}}\psi_{\varepsilon}\delta_{a\vert\varepsilon}u_{i}\left(a,\varepsilon_{i}\right)\geq\sum_{\varepsilon_{-i}}\psi_{\varepsilon}\delta_{a\vert\varepsilon}u_{i}\left(a_{i}',a_{-i},\varepsilon_{i}\right),\quad\forall i,\varepsilon_{i},a,a_{i}'
\]
which is exactly the REE condition for $\left(G,S^{private}\right)$.

$\left(\Rightarrow\right)$ Conversely, let $\delta:\mathcal{E}\to\Delta\left(\mathcal{A}\right)$
be a REE of $\left(G,S^{private}\right)$. To show that $\delta$
is a REE of $\left(G,S^{complete}\right)$, by Theorem \ref{thm:Relationship to PSNE}.\ref{enu:REE PSNE 1},
it is enough to show that for each $\varepsilon$, $\delta_{a\vert\varepsilon}>0$
implies that $a$ is a PSNE of $\Gamma_{\varepsilon}$. \\
Since $\delta$ is a REE of $\left(G,S^{private}\right)$, by definition,
we have
\begin{gather*}
\sum_{\varepsilon_{-i}}\psi_{\varepsilon}\delta_{a\vert\varepsilon}u_{i}\left(a,\varepsilon_{i}\right)\geq\sum_{\varepsilon_{-i}}\psi_{\varepsilon}\delta_{a\vert\varepsilon}u_{i}\left(a_{i}',a_{-i},\varepsilon_{i}\right),\quad\forall i,\varepsilon_{i},a,a_{i}'\\
\Leftrightarrow\varphi\left(a,\varepsilon_{i}\right)u_{i}\left(a,\varepsilon_{i}\right)\geq\varphi\left(a,\varepsilon_{i}\right)u_{i}\left(a_{i}',a_{-i},\varepsilon_{i}\right),\quad\forall i,\varepsilon_{i},a,a_{i}'
\end{gather*}
where $\varphi\left(a,\varepsilon_{i}\right):=\sum_{\varepsilon_{-i}}\psi_{\varepsilon}\delta_{a\vert\varepsilon}$.
\\
Now fix $\varepsilon$ and consider any $a$ such that $\delta_{a\vert\varepsilon}>0$.
But $\delta_{a\vert\varepsilon}>0$ implies $\varphi\left(a,\varepsilon_{i}\right)>0$
which in turn implies that 
\[
u_{i}\left(a,\varepsilon_{i}\right)\geq u_{i}\left(a_{i}',a_{-i},\varepsilon_{i}\right),\quad\forall i,a_{i}'
\]
which is exactly the PSNE condition of $a$ at $\varepsilon$. \\
The second statement follows similarly as above; under $S^{private}$,
outcome functions and decision rules are equivalent since players'
signals exhaust information about the state of the world.$\square$
\end{enumerate}

\subsection{Proof of Theorem \ref{thm:Equivalence of identified sets}}

Let $S\equiv\left(S^{x}\right)_{x\in\mathcal{X}}$ and $\tilde{S}\equiv\left(\tilde{S}^{x}\right)_{x\in\mathcal{X}}$.
Let $\tilde{S}\succsim_{E}S$ if and only if $\tilde{S}^{x}\succsim_{E}S^{x}$
for each $x\in\mathcal{X}$. We want to show
\[
\Theta_{I}^{BSE}\left(S\right)=\bigcup_{\tilde{S}\succsim_{E}S}\Theta_{I}^{REE}\left(\tilde{S}\right).
\]
Note that 
\begin{equation}
\Theta_{I}^{BSE}\left(S\right)\equiv\left\{ \theta\in\Theta:\ \forall x\in\mathcal{X},\ \phi^{x}\in\mathcal{P}_{a}^{BSE}\left(G^{x,\theta},S^{x}\right)\right\} \label{eq:equivalence1}
\end{equation}
and
\begin{align}
\bigcup_{\tilde{S}\succsim_{E}S}\Theta_{I}^{REE}\left(\tilde{S}\right) & \equiv\bigcup_{\tilde{S}\succsim_{E}S}\left\{ \theta\in\Theta:\ \forall x\in\mathcal{X},\ \phi^{x}\in\mathcal{P}_{a}^{REE}\left(G^{x,\theta},\tilde{S}^{x}\right)\right\} \nonumber \\
 & =\left\{ \theta\in\Theta:\ \forall x\in\mathcal{X},\ \phi^{x}\in\bigcup_{\tilde{S}^{x}\succsim_{E}S^{x}}\mathcal{P}_{a}^{REE}\left(G^{x,\theta},\tilde{S}^{x}\right)\right\} .\label{eq:equivalence2}
\end{align}

By Corollary \ref{cor:CCP for BSE and REE}, for any given $\theta\in\Theta$
and $x\in\mathcal{X}$, we have
\begin{equation}
\mathcal{P}_{a}^{BSE}\left(G^{x,\theta},S^{x}\right)=\bigcup_{\tilde{S}^{x}\succsim_{E}S^{x}}\mathcal{P}_{a}^{REE}\left(G^{x,\theta},\tilde{S}^{x}\right).\label{eq:equivalence3}
\end{equation}
That (\ref{eq:equivalence1}) and (\ref{eq:equivalence2}) are equal
follows from (\ref{eq:equivalence3}), which is what we wanted. $\square$

\subsection{Proof of Theorem \ref{thm:Relationship between identified sets}}
\begin{enumerate}
\item Let $G$ be an arbitrary basic game. We suppress the covariates $x$
since they do not play a role. Let $S^{1}$ and $S^{2}$ be arbitrary
information structures such that $S^{1}\succsim_{E}S^{2}$. It is
enough to show that a BSE in $\left(G,S^{1}\right)$ always induces
a BSE in $\left(G,S^{2}\right)$ because it will imply that the set
of feasible CCPs in $\left(G,S^{1}\right)$ is a subset of the feasible
CCPs in $\left(G,S^{2}\right)$. But this directly follows from Corollary
\ref{cor:more info leads to tighter set}. $\square$ 
\item The statement follows from Theorem \ref{thm:Relationship to PSNE}.
In particular, note that when pure strategy Nash equilibrium is the
relevant solution concept, the decision rule (or the outcome function)
simply represents an arbitrary equilibrium selection mechanism; no
assumption is placed on the equilibrium selection rule. Since the
set of probability distributions over $\mathcal{A}$ on each realization
of $\varepsilon$ is the same across Bayes stable equilibria and pure
strategy Nash equilibria, the resulting identified set of parameters
must be identical. $\square$
\item The statement follows from Theorem \ref{thm:BSE and BCE}. Theorem
\ref{thm:BSE and BCE} says that for any $\left(G,S\right)$, if a
decision rule $\sigma$ in $\left(G,S\right)$ is a Bayes stable equilibrium
of $\left(G,S\right)$, then it is a Bayes correlated equilibrium
of $\left(G,S\right)$. This implies that we will have $\mathcal{P}_{a}^{BSE}\left(G,S\right)\subseteq\mathcal{P}_{a}^{BCE}\left(G,S\right)$
for any $\left(G,S\right)$, which leads to the statement. $\square$
\end{enumerate}

\subsection{Proof of Theorem \ref{thm:ConfidenceSet}}
\begin{enumerate}
\item The first statement follows directly from construction:
\begin{align*}
\text{Pr}\left(\Theta_{I}\subseteq\widehat{\Theta}_{I}^{\alpha}\right) & =\text{Pr}\left(\Theta_{I}\left(\phi\right)\subseteq\bigcup_{\tilde{\phi}\in\Phi_{n}^{\alpha}}\Theta_{I}\left(\tilde{\phi}\right)\right)\geq\text{Pr}\left(\phi\in\Phi_{n}^{\alpha}\right)
\end{align*}
(The inequality follows from the possibility that there may exist
$\bar{\phi}\neq\phi$ such that $\bar{\phi}\in\Phi_{n}^{\alpha}$
but $\Theta_{I}\left(\phi\right)\subseteq\Theta_{I}\left(\bar{\phi}\right)$.)
Taking the limits on both sides gives the desired result. $\square$
\item The second statement follows from the fact that $\phi$ enters the
population program (see Theorem \ref{thm:Identified set LP}) in an
additively separable manner, and that $\phi\in\Phi_{n}^{\alpha}$
represents a set of convex constraints. To see this, note that $\theta\in\widehat{\Theta}_{I}^{\alpha}$
if and only if the following program is feasible: For each $x\in\mathcal{X}$,
find $\sigma^{x}\in\Delta_{a\vert\varepsilon,t}$ and $\phi^{x}\in\Delta_{a}$
such that
\begin{gather*}
\sum_{\varepsilon,t_{-i}}\psi_{\varepsilon}^{x,\theta}\pi_{t\vert\varepsilon}^{x}\sigma_{a\vert\varepsilon,t}^{x}\partial u_{i}^{x,\theta}\left(a_{i}',a,\varepsilon_{i}\right)\leq0,\quad\forall i,t_{i},a,a_{i}'\\
\phi_{a}^{x}=\sum_{\varepsilon,t}\psi_{\varepsilon}^{x,\theta}\pi_{t\vert\varepsilon}^{x}\sigma_{a\vert\varepsilon,t}^{x},\quad\forall a,x\\
\phi\in\Phi_{n}^{\alpha}.
\end{gather*}
That is, compared to the population program which treats $\phi$ as
known, we let $\phi$ be a variable of optimization and add convex
constraints $\phi\in\Phi_{n}^{\alpha}$. Under the assumption that
$\phi\in\Phi_{n}^{\alpha}$ represents convex constraints, the above
program is convex. $\square$
\end{enumerate}

\subsection{Proof of Theorem \ref{thm:Implementation}}
\begin{enumerate}
\item First, let use show that (\ref{eq:program Qhat}) is always feasible
for any $\theta$. Pick any $\bar{\phi}\in\Phi_{n}^{\alpha}$. For
any $\bar{\phi}$, we can find a $\bar{\sigma}$ satisfying $\bar{\phi}_{a}^{x}=\sum_{\varepsilon,t}\psi_{\varepsilon}^{x,\theta}\pi_{t\vert\varepsilon}^{x}\sigma_{a\vert\varepsilon,t}^{x}$
for all $a,x$. Finally, there exists a non-negative vector of $\left\{ q_{x}\right\} _{x\in\mathcal{X}}$
such that $\sum_{\varepsilon,t_{-i}}\psi_{\varepsilon}^{x,\theta}\pi_{t\vert\varepsilon}^{x}\sigma_{a\vert\varepsilon,t}^{x}\partial u_{i}^{x,\theta}\left(\tilde{a}_{i},a,\varepsilon_{i}\right)\leq q_{x}$
for all $i,x,t_{i},a,\tilde{a}_{i}$. Therefore, the feasible set
of $\left(q,\sigma,\phi\right)$ is always non-empty. Second, convexity
of program (\ref{eq:program Qhat}) follows from the fact that all
the constraints are linear in $\left(q,\sigma,\phi\right)$ and that
$\phi\in\Phi_{n}^{\alpha}$ represents a set of convex constraints.
$\square$
\item It is straightforward to show that $\hat{Q}_{n}^{\alpha}\left(\theta\right)=0$
if and only if $\theta\in\widehat{\Theta}_{I}^{\alpha}$. If $\hat{Q}_{n}^{\alpha}\left(\theta\right)=0$,
then it must be that $q_{x}^{*}=0$ for all $x\in\mathcal{X}$, implying
that $\theta\in\widehat{\Theta}_{I}^{\alpha}$. Conversely, if $\theta\in\widehat{\Theta}_{I}^{\alpha}$,
then we can get $\hat{Q}_{n}^{\alpha}\left(\theta\right)=0$ by plugging
in $q_{x}=0$ for all $x\in\mathcal{X}$. $\square$
\item Finally, we can obtain $\nabla\widehat{Q}_{n}^{\alpha}\left(\theta\right)$
as a byproduct to the convex program using the envelope theorem. $\square$
\end{enumerate}
\newpage{}

\part*{Supplementary Materials (Online Appendix)}

\section{Computational Details\label{sec:Computational-Details}}

\subsection{Discretization of Unobservables \label{subsec:Discretization-of-Unobservables}}

Our approach to econometric analysis requires a discrete approximation
to the distribution of payoff shocks, which are often assumed to be
continuous. We follow a discretization approach similar to that taken
in \citet{magnolfi2021estimation}, which requires finding a finite
set of representative points on the support, and assigning appropriate
probability mass on each point of the discretized support. The only
difference is that \citet{magnolfi2021estimation} uses equally spaced
quantiles of the distribution of $\varepsilon_{i}$'s to find the
discretized support whereas we use the approach introduced in \citet{kennan_note_2006}
to find the discretized support.

First, to discretize the space of each $\varepsilon_{i}\in\mathbb{R}$,
we adopt the recommendations by \citet{kennan_note_2006}, which have
been used in several works, e.g., \citet{kennan2011theeffect}, \citet{lee2019onthe},
and \citet{aizawa2020equilibrium}. Let us briefly describe the procedure.
Let $F_{0}$ be the true continuous distribution of a scalar random
variable $\varepsilon_{i}$ with support $\mathcal{E}_{0}$. Suppose
we want to find an $N$-point discrete approximation to $F_{0}$.
Specifically, we want to find a pair $\left(\mathcal{E},F\right)$
where $\mathcal{E}$ contains $N$ points and $F$ describes the probability
mass on each of the $n$ points. How should we choose $\mathcal{E}$
and $F$?

\citet{kennan_note_2006} characterizes the ``best'' discrete approximation
$\left(\mathcal{E},F\right)$ to $\left(\mathcal{E}_{0},F_{0}\right)$,
measured in $L^{p}$ norm (for any $p>0$) when the researcher can
choose $N$ points. We restate the proposition introduced in \citet{kennan_note_2006}.
\begin{prop*}[Kennan 2006]
 The best $N$-point approximation $F$ to a given distribution $F_{0}$
has equally-weighted support points $\mathcal{E}\equiv\left\{ x_{j}^{*}\right\} _{j=1}^{N}$
given by
\[
F\left(x_{j}^{*}\right)=\frac{2j-1}{2N}
\]
for $j=1,...,N$.
\end{prop*}
Following the proposition, we discretize unobservables as follows.
In a two-player game with binary actions, we take the benchmark distribution
of firm $i$'s random shock $\varepsilon_{i}$ to be the standard
normal distribution. We fix the number of grid points $N$ (we use
$N=10$ for empirical application) and find $\mathcal{E}_{i}\equiv\left\{ x_{j}^{*}\right\} _{j=1}^{N}$
as described above. Then we take the Cartesian product of $\mathcal{E}_{1}$
and $\mathcal{E}_{2}$ to set the discrete support of $\left(\varepsilon_{1},\varepsilon_{2}\right)$.
In the baseline case where $\varepsilon_{1}$ is uncorrelated with
$\varepsilon_{2}$, we construct the discretized prior distribution
$\psi$ as an $N\times N$ matrix whose entries are constant at $\frac{1}{N\times N}$.
Thus, $\psi\left(\varepsilon_{1},\varepsilon_{2}\right)=\frac{1}{N\times N}$
for any $\left(\varepsilon_{1},\varepsilon_{2}\right)\in\mathcal{E}\equiv\mathcal{E}_{1}\times\mathcal{E}_{2}$.
For example, when each $\varepsilon_{i}$ is approximated with $N=20$
points, we have $20^{2}=400$ points in $\mathcal{E}$ with $\psi$
assigning mass $1/400$ to each point in $\mathcal{E}$.

Second, to capture correlated unobservables, we apply weights to each
point in $\mathcal{E}$ where the weights are generated using the
density of the Gaussian copula. Specifically, we find the weight at
each point $\varepsilon=\left(\varepsilon_{1},\varepsilon_{2}\right)\in\mathcal{E}$
to be proportional to the density of bivariate Gaussian copula evaluated
at the point with correlation matrix $R=\left[\begin{array}{cc}
1 & \rho\end{array};\begin{array}{cc}
\rho & 1\end{array}\right]$. In the special case $\rho=0$, the approach applies uniform weights
to each point on $\mathcal{E}$, and we return to the case where $\psi$
has constant mass on every point on $\mathcal{E}$. Extension to the
case with more than two players is straightforward.

In Figure \ref{fig:True-Correlation-vs.}, we plot the true correlation
coefficient against the estimated correlation coefficient obtained
using the discretization approach with $N_{E}=10$. The figure shows
that discretized distribution has estimated correlation coefficient
slightly smaller than the true (input) correlation coefficient $\rho$.

\begin{figure}[h]
\caption{\label{fig:True-Correlation-vs.}True Correlation vs. Estimated Correlation
($N_{E}=10$)}

\centering{}\includegraphics[scale=0.35]{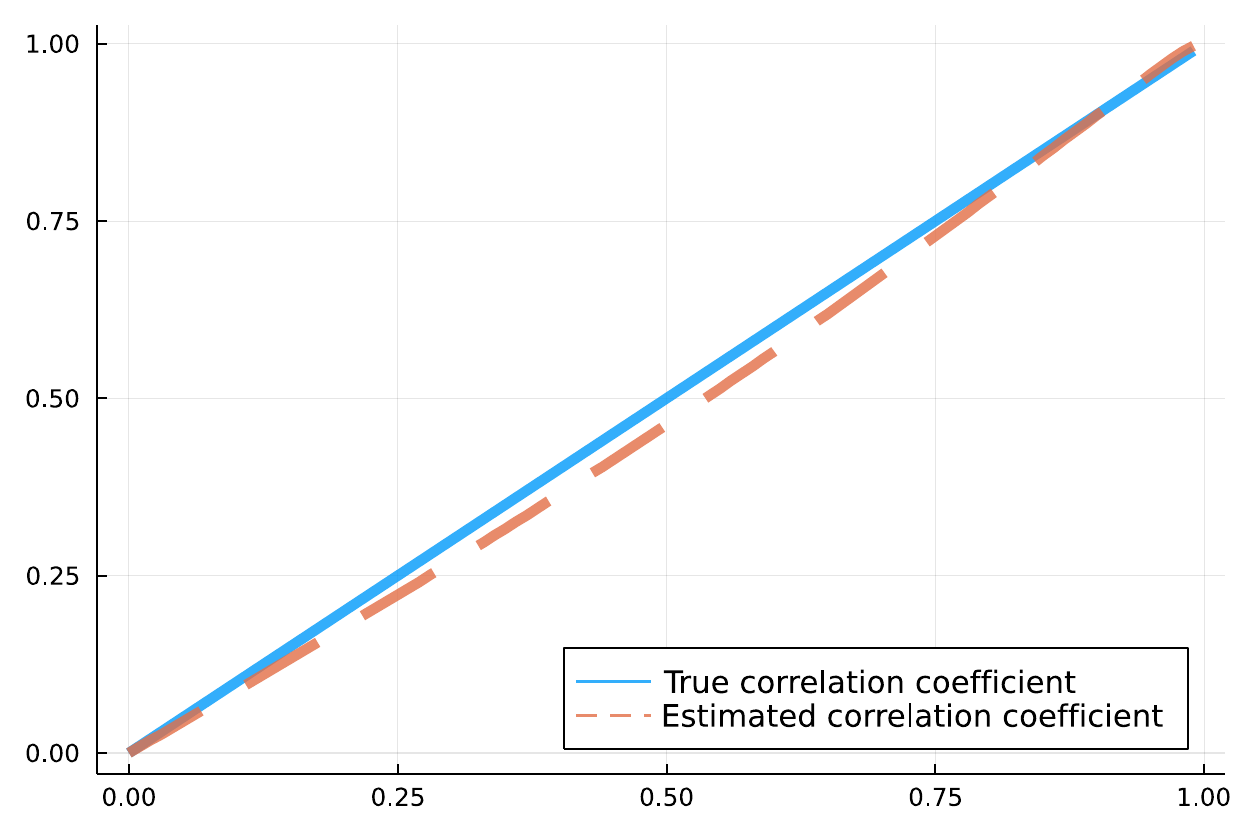}
\end{figure}

Note that whereas \citet{kennan_note_2006} shows an ``optimal''
way of discretizing the support of a univariate random variable, we
do not have such optimality result for a multivariate case. Thus,
our approach should be understood as being heuristic.

\subsubsection{Maximal Error from Discrete Approximation}

Given that our approach relies on discrete approximations (as done
in \citet*{syrgkanis2021inference} and \citet{magnolfi2021estimation}),
a natural question is how accurate the approximation is. We provide
a simple numerical evidence which supports the claim that the approximation
error is at most mild.

Consider a two-player entry game with payoff $u_{i}\left(a_{i},a_{j},\varepsilon_{i}\right)=a_{i}\left(\kappa_{i}a_{j}+\varepsilon_{i}\right)$.
We generate observed choice probability data at $\left(\kappa_{1},\kappa_{2}\right)=\left(-0.5,-0.5\right)$
using a continuous distribution $\varepsilon_{i}\overset{\text{iid}}{\sim}N\left(0,1\right)$,
and symmetric equilibrium selection probability. The population choice
probability is $\left(\phi_{00},\phi_{01},\phi_{10},\phi_{11}\right)\approx\left(0.25,0.3274,0.3274,0.0952\right)$.

If we use the discrete approximation procedure described above, how
much error can there be? Our measure of discrepancy is the solution
to
\begin{gather*}
\min_{t\in\mathbb{R},\sigma\in\Delta_{a\vert\varepsilon}}t\quad\text{subject to}\\
\sum_{\varepsilon_{-i}}\psi_{\varepsilon}\sigma_{a\vert\varepsilon}\partial u_{i}\left(\tilde{a}_{i},a,\varepsilon_{i}\right)\leq t,\quad\forall i,\varepsilon_{i},a,\tilde{a}_{i}\\
\sum_{\varepsilon}\psi_{\varepsilon}\sigma_{a\vert\varepsilon}-\phi_{a}\leq t,\quad\forall a\\
\phi_{a}-\sum_{\varepsilon}\psi_{\varepsilon}\sigma_{a\vert\varepsilon}\leq t,\quad\forall a
\end{gather*}
The solution $t^{*}$ measures the maximal relaxation required for
the equilibrium conditions and the consistency conditions. If $t^{*}=0$,
there is no approximation error. In general, we can expect $t^{*}>0$.
Let $N_{E}$ be the number of grid points used for approximating $N\left(0,1\right)$.
(We use $N_{E}=10$ for $\varepsilon_{1}$ and $\varepsilon_{2}$
in our empirical application which produces $10^{2}=100$ points for
the support of $\psi$.)

\begin{figure}[h]
\caption{\label{fig:Discrete-approximation-error}Discrete approximation error}

\centering{}\includegraphics[scale=0.35]{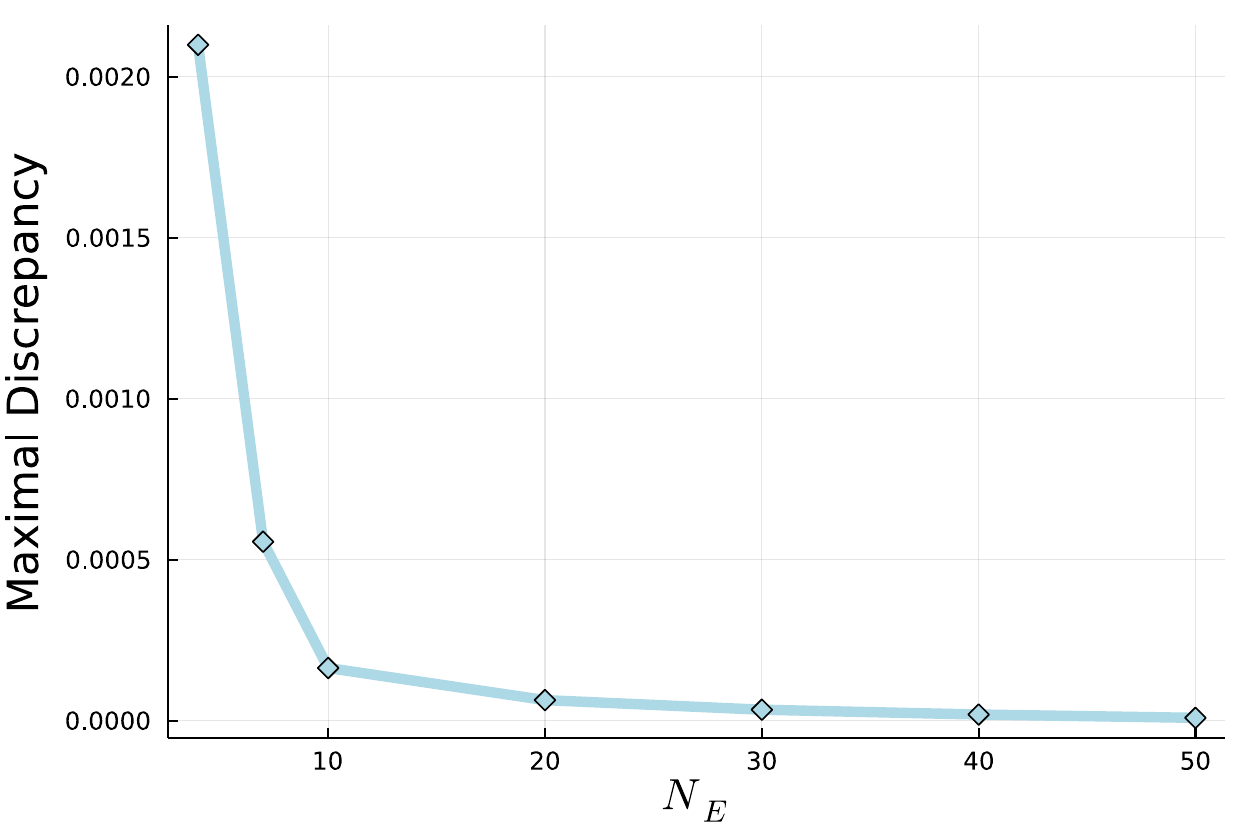}
\end{figure}

Figure \ref{fig:Discrete-approximation-error} plots $t^{*}$ (``maximal
discrepancy'') against $N_{E}$. The figure shows that the discrepancy
is decreasing in $N_{E}$ and at most modest after $N_{E}=10$. Since
we construct confidence sets for the conditional choice probabilities
when we do inference, it is likely that the approximation error will
be controlled together. For this reason, it seems quite unlikely that
discretization error will contaminate the estimation results.

\subsection{Construction of Convex Confidence Sets for Conditional Choice Probabilities\label{subsec:Construction-of-Convex}}

In this section, we describe a simple approach to constructing confidence
sets for the conditional choice probabilities, which we use for the
empirical application. We construct simultaneous confidence intervals
based on \citet{fitzpatrick1987quicksimultaneous}. The basic idea
is to construct confidence intervals for each multinomial proportion
parameter so that the confidence set for the conditional choice probabilities
can be characterized as a set of constraints that are \emph{linear}
in the population conditional choice probabilities. While there are
many ways of constructing simultaneous confidence bands for a vector
of means (e.g., see \citet{olea2019simultaneous} and the references
therein), we follow \citet{fitzpatrick1987quicksimultaneous} because
it provides a very simple approach to constructing simultaneous confidence
intervals for multinomial proportion parameters.\footnote{Specifically, \citet{fitzpatrick1987quicksimultaneous} shows a particular
simultaneous confidence intervals for multinomial proportion parameters
that are extremely easy to construct and characterizes the asymptotic
coverage probabilities.}

Let $\mathcal{X}$ be a finite set of covariates and $\vert\mathcal{X}\vert$
its cardinality. Let $\phi_{a}^{x}\in\mathbb{R}$ be the population
choice probability of outcome $a\in\mathcal{A}$ at bin $x\in\mathcal{X}$.
At each bin $x$, the conditional choice probabilities $\phi^{x}\equiv\left(\phi_{a}^{x}\right)_{a\in\mathcal{A}}\in\mathbb{R}^{\vert\mathcal{A}\vert}$
represent the proportion parameters of a multinomial distribution.
The entire vector of conditional choice probabilities is denoted $\phi\equiv\left(\phi^{x}\right)_{x\in\mathcal{X}}\in\mathbb{R}^{\vert\mathcal{A}\vert\times\vert\mathcal{X}\vert}$.
Let $n^{x}\in\mathbb{Z}$ be the number of observations at each bin
$x$, and let $n\equiv\sum_{x\in\mathcal{X}}n^{x}$ be the total number
of observations in the data.

Our strategy is described as follows. Our objective is to construct
a confidence set $\Phi_{n}^{\alpha}$ that covers $\phi$ with probability
at least $1-\alpha$ asymptotically where $\alpha\in\left(0,1\right)$.
To do so, we will construct a confidence set $\Phi_{n^{x}}^{x,\beta_{\alpha}}$
at each bin $x$ that covers $\phi^{x}$ with probability at least
$1-\beta_{\alpha}$ asymptotically where $\beta_{\alpha}=1-\left(1-\alpha\right)^{1/\vert\mathcal{X}\vert}$
($\beta_{\alpha}$ arises from applying the Šidák correction for testing
$\vert\mathcal{X}\vert$ number of independent hypotheses with family-wise
error rate $\alpha$; note that the samples in each bin $x$ are independent
from each other when the data are generated from independent markets).
Next, we will construct $\Phi_{n}^{\alpha}$ by taking intersections
of $\Phi_{n^{x}}^{x,\beta_{\alpha}}$ across $x$; making the coverage
probability for $\phi^{x}$ at each $x$ be no less than $1-\beta_{\alpha}$
ensures that the overall coverage probability for $\phi$ is no less
than $1-\alpha$. Moreover, if, for each $x$, $\Phi_{n^{x}}^{x,\beta_{\alpha}}$
can be represented by a set of constraints linear in $\phi^{x}$,
then $\Phi_{n}^{\alpha}$ will be represented by a set of constraints
linear in $\phi$ by construction.

At each $x\in\mathcal{X}$, we define the confidence set for $\phi^{x}$
as follows. Let $\hat{\phi}_{a}^{x}\equiv n_{a}^{x}/n^{x}\in\mathbb{R}$
be the nonparametric frequency estimator of $\phi_{a}^{x}$ where
$n_{a}^{x}\in\mathbb{Z}$ is the number of observations with outcome
$a$ at bin $x$. Then construct $\Phi_{n^{x}}^{x,\beta_{\alpha}}$
as:
\begin{equation}
\Phi_{n^{x}}^{x,\beta_{\alpha}}\equiv\left\{ \phi^{x}:\ \phi_{a}^{x}\in\hat{\phi}_{a}^{x}\pm\frac{z\left(\beta_{\alpha}/4\right)}{2\sqrt{n^{x}}},\quad\forall a\in\mathcal{A}\right\} ,\label{eq:FS1}
\end{equation}
where $z\left(\tau\right)\in\mathbb{R}$ denotes the upper $100\left(1-\tau\right)\%$
quantile of the standard normal distribution.\footnote{Although the intervals may include values lower than 0 or higher than
1, we impose the condition that $\phi_{a}^{x}\in\left[0,1\right]$
for each $a,x$ and $\sum_{a}\phi_{a}^{x}=1$ for each $x$ in the
optimization problem.} Note that $\Phi_{n^{x}}^{x,\beta_{\alpha}}$ consists of $\vert\mathcal{A}\vert$
number of confidence intervals.

Finally, we define a confidence region for $\phi$ as:
\begin{equation}
\Phi_{n}^{\alpha}\equiv\left\{ \phi:\ \phi^{x}\in\Phi_{n^{x}}^{x,\beta_{\alpha}},\quad\forall x\in\mathcal{X}\right\} .\label{eq:FS2}
\end{equation}

The following proposition states that, under regular conditions, $\Phi_{n}^{\alpha}$
constructed as (\ref{eq:FS2}) has the desired asymptotic coverage
probabilities for the population conditional choice probabilities
$\phi$.
\begin{prop}
\label{prop:SCI using FP}Let $\Phi_{n}^{\alpha}$ be defined as (\ref{eq:FS2}).
Suppose that samples are independent across $x\in\mathcal{X}$, and
$n^{x}\to\infty$ for each $x\in\mathcal{X}$ as $n\to\infty$. If
$\alpha$ is sufficiently low or $\vert\mathcal{X}\vert$ is sufficiently
large so that $\beta_{\alpha}\leq0.032$, we have 
\[
\lim_{n\to\infty}\text{Pr}\left(\phi\in\Phi_{n}^{\alpha}\right)\geq1-\alpha.
\]
\end{prop}
To prove Proposition \ref{prop:SCI using FP}, we use Theorem 1 of
\citet{fitzpatrick1987quicksimultaneous} as a lemma. The lemma characterizes
the asymptotic lower bounds on the coverage probabilities of $\Phi_{n^{x}}^{x,\beta_{\alpha}}$
for $\phi^{x}$ when the intervals of form (\ref{eq:FS1}) are used.
\begin{lem}[\citet{fitzpatrick1987quicksimultaneous} Theorem 1]
 Let $\Phi_{n^{x}}^{x,\beta_{\alpha}}$ be defined as (\ref{eq:FS1}).
Then 
\[
\lim_{n^{x}\to\infty}\text{Pr}\left(\phi^{x}\in\Phi_{n^{x}}^{x,\beta_{\alpha}}\right)\geq\mathcal{L}\left(\beta_{\alpha}\right)
\]
where 
\[
\mathcal{L}\left(\beta_{\alpha}\right)=\begin{cases}
1-\beta_{\alpha}, & \text{if }\beta_{\alpha}\leq0.032\\
6\Phi\left(\frac{3z\left(\beta_{\alpha}/4\right)}{\sqrt{8}}\right)-5, & \text{if \ensuremath{0.032\leq\beta_{\alpha}\leq0.3}}
\end{cases}.
\]
\end{lem}
Now let us prove Proposition \ref{prop:SCI using FP}. The proof uses
the fact that (i) the samples are independent across $x\in\mathcal{X}$,
(ii) $\Phi_{n^{x}}^{x,\beta_{\alpha}}$ covers $\phi^{x}$ with probability
no less than $\beta_{\alpha}$ asymptotically, and (iii) $\beta_{\alpha}$
is chosen in a way that ensures the overall coverage probability for
$\phi$ becomes no less than $1-\alpha$ asymptotically (Šidak correction).
\begin{proof}
We have
\begin{align}
\text{Pr}\left(\phi\in\Phi_{n}^{\alpha}\right) & =\text{Pr}\left(\phi^{x}\in\Phi_{n^{x}}^{x,\beta_{\alpha}},\quad\forall x\in\mathcal{X}\right)\nonumber \\
 & =\prod_{x\in\mathcal{X}}\text{Pr}\left(\phi^{x}\in\Phi_{n^{x}}^{x,\beta_{\alpha}}\right)\label{eq:FSproof1}
\end{align}
where (\ref{eq:FSproof1}) follows from the independence across $x\in\mathcal{X}$.
Given that $\beta_{\alpha}$ is sufficiently small, taking the limit
gives
\begin{align}
\lim_{n\to\infty}\prod_{x\in\mathcal{X}}\text{Pr}\left(\phi^{x}\in\Phi_{n^{x}}^{x,\beta_{\alpha}}\right) & =\prod_{x\in\mathcal{X}}\lim_{n^{x}\to\infty}\text{Pr}\left(\phi^{x}\in\Phi_{n^{x}}^{x,\beta_{\alpha}}\right)\label{eq:FSproof2}\\
 & \geq\prod_{x\in\mathcal{X}}\left(1-\beta_{\alpha}\right)\label{eq:FSproof3}\\
 & =\left(1-\beta_{\alpha}\right)^{\vert\mathcal{X}\vert}\nonumber \\
 & =\left(1-\left\{ 1-\left(1-\alpha\right)^{1/\vert\mathcal{X}\vert}\right\} \right)^{\vert\mathcal{X}\vert}\label{eq:FSproof4}\\
 & =1-\alpha.\nonumber 
\end{align}
where (\ref{eq:FSproof2}) follows from the product rule of limits,
(\ref{eq:FSproof3}) follows from \citet{fitzpatrick1987quicksimultaneous}
Theorem 1, and (\ref{eq:FSproof4}) follows from the definition of
$\beta_{\alpha}$.
\end{proof}
The main advantage of using \citet{fitzpatrick1987quicksimultaneous}
is its simplicity. The method is easily applicable even when there
are zero count cells, i.e., $n_{a}^{x}=0$ for some $a\in\mathcal{A}$
and $x\in\mathcal{X}$. Zero count cells often occur when the sample
size is small and may require some correction if other popular approaches
(e.g., normal approximation for each $\phi_{a}^{x}$ taken as a Bernoulli
parameter) were used. The simultaneous confidence bands can be conservative,
but retains a linear structure, which is computationally attractive.
\begin{example}
Suppose there are two bins $\mathcal{X}=\left\{ l,h\right\} $, and
that the number of observations at each bin is $n^{l}=400$ and $n^{h}=600.$
Suppose that $\mathcal{A}=\left\{ 00,01,10,11\right\} $ so that $\phi^{x}=\left(\phi_{00}^{x},\phi_{01}^{x},\phi_{10}^{x},\phi_{11}^{x}\right)$
and that we obtained $\hat{\phi}^{l}=\left(0.1,0.1,0.4,0.4\right)$
and $\hat{\phi}^{h}=\left(0.2,0.3,0.3,0.2\right)$ using nonparametric
frequency estimators at each bin. If $\alpha=0.05$, then $\beta_{\alpha}=1-\left(1-\alpha\right)^{1/2}=0.0253$.
Then $z\left(\beta_{\alpha}/4\right)=z\left(1-0.0253/4\right)=2.4931$.
Finally, since $z\left(\beta_{\alpha}/4\right)/\left(2\sqrt{400}\right)=0.0623$
and $z\left(\beta_{\alpha}/4\right)/\left(2\sqrt{600}\right)=0.0509$,
our $\Phi_{n}^{\alpha}$ is defined by the following inequalities:
\begin{align*}
\hat{\phi}_{a}^{l}-0.0623 & \leq\phi_{a}^{l}\leq\hat{\phi}_{a}^{l}+0.0623,\quad\forall a\in\mathcal{A}\\
\hat{\phi}_{a}^{h}-0.0509 & \leq\phi_{a}^{h}\leq\hat{\phi}_{a}^{h}+0.0509,\quad\forall a\in\mathcal{A}.
\end{align*}
$\blacksquare$
\end{example}

\subsubsection{Monte Carlo Experiment}

We conduct Monte Carlo experiments to examine whether the simultaneous
confidence bands have correct coverage probabilities and confirm that
the approach works well. Let $\mathcal{X}=\left\{ 1,2,...,N_{\mathcal{X}}\right\} $
be a finite set indices (covariates). The following constitutes a
single trial. We randomly generate a probability vector $\phi^{x}\in\mathbb{R}^{4}$
for $x=1,...,N_{\mathcal{X}}$ by taking a 4-dimensional uniform vector
and normalize the vector so that it sums to one. Then, at each $x\in\mathcal{X}$,
we generate a random sample by taking a draw from a multinomial distribution
with parameter $\left(n^{x},\phi^{x}\right)$ where $n^{x}$ is the
number trials. Finally, we test whether the simultaneous confidence
bands, constructed as described above, covers $\phi^{x}$. We repeat
this procedure for $100,000$ times and find the coverage probability.

\begin{table}[h]
\caption{\label{tab:Monte-Carlo-Experiment:} Coverage Probability of Simultaneous
Confidence Bands from Simulation}

{\footnotesize
\begin{centering}
\begin{tabular}{cccccccccccc}
\hline 
 & \multicolumn{5}{c}{(A) $\alpha=0.05$} &  & \multicolumn{5}{c}{(B) $\alpha=0.01$}\tabularnewline
$N_{\mathcal{X}}\backslash n^{x}$ & 100 & 200 & 500 & 1000 & 10000 &  & 100 & 200 & 500 & 1000 & 10000\tabularnewline
\cline{1-6} \cline{2-6} \cline{3-6} \cline{4-6} \cline{5-6} \cline{6-6} \cline{8-12} \cline{9-12} \cline{10-12} \cline{11-12} \cline{12-12} 
4 & 0.9697 & 0.9707 & 0.9713 & 0.9744 & 0.9837 &  & 0.9950 & 0.9948 & 0.9957 & 0.9956 & 0.9975\tabularnewline
10 & 0.9735 & 0.9731 & 0.9748 & 0.9754 & 0.9854 &  & 0.9955 & 0.9954 & 0.9957 & 0.9960 & 0.9978\tabularnewline
50 & 0.9760 & 0.9760 & 0.9777 & 0.9797 & 0.9885 &  & 0.9958 & 0.9962 & 0.9962 & 0.9968 & 0.9981\tabularnewline
100 & 0.9779 & 0.9788 & 0.9791 & 0.9811 & 0.9886 &  & 0.9959 & 0.9961 & 0.9964 & 0.9969 & 0.9982\tabularnewline
200 & 0.9776 & 0.9783 & 0.9794 & 0.9816 & 0.9902 &  & 0.9964 & 0.9962 & 0.9966 & 0.9971 & 0.9984\tabularnewline
\hline 
\end{tabular}
\par\end{centering}
}
\end{table}

Table \ref{tab:Monte-Carlo-Experiment:} reports the results of the
Monte Carlo experiment. It shows that the confidence sets obtain desired
coverage probabilities although they can be conservative. We conclude
that the proposed approach works well.

\subsection{Random Walk Surface Scanning Algorithm\label{subsec:Random-Walk-Surface}}

Let $\Theta_{I}$ be the identified set of parameters. The identified
set is defined as the level set
\[
\Theta_{I}\equiv\left\{ \theta\in\Theta:Q\left(\theta\right)\leq0\right\} 
\]
where $Q\left(\theta\right)$ is a non-negative valued criterion function.
(To obtain the confidence set, simply replace $Q\left(\theta\right)$
with $\widehat{Q}_{n}^{\alpha}\left(\theta\right)$.) Except for special
cases (e.g., when $\Theta_{I}$ is convex), we need to approximate
$\Theta_{I}$ by collecting a large number of points in $\Theta_{I}$.
A naive approach is to conduct an extensive grid search: draw a fine
grid on the parameter space $\Theta$ (e.g., by taking quasi-Monte
Carlo draws) and evaluate the criterion function at all point on the
grid. However, a naive grid search can be computationally burdensome
especially when the dimension of $\theta$ is large.

In our setup, Theorem \ref{thm:Implementation} says that we can get
the gradient information for free due to the envelope theorem. That
is, once we evaluate $Q\left(\theta\right)$ at any $\theta$, we
can get $\nabla Q\left(\theta\right)$ as well. Exploiting the gradient
information allows us to find a minimizer of $Q\left(\theta\right)$
far more efficiently because we can use gradient-based optimization
algorithms (e.g., gradient descent or (L-)BFGS) as opposed to gradient-free
algorithms. However, since we need to find \emph{all} minimizers of
$Q\left(\theta\right)$, solving $\min_{\theta}Q\left(\theta\right)$
is insufficient.

We propose a heuristic approach. First, we identify $\theta^{0}=\arg\min_{\theta}Q\left(\theta\right)$
by using a gradient-based optimization algorithm. Second, we iteratively
explore the neighbors of the identified set by running a random walk
process from $\theta^{0}$ and accepting points at which the criterion
function is zero-valued. Being able to quickly identify a point in
the identified set gives a considerable advantage over grid search
algorithms because we do not have to explore points that are ``far''
from the identified set. The required assumption is that $\Theta_{I}$
is a connected set.

We use a random walk surface scanning algorithm described as follows.
Let $\theta^{0}=\arg\min_{\theta}Q\left(\theta\right)$ be the identified
parameter and assume that $Q\left(\theta^{0}\right)=0$ (otherwise
the identified set is empty). From $\theta^{0}$, we take a random
candidate
\[
\tilde{\theta}^{1}\leftarrow\theta^{0}+\eta
\]
where $\eta\sim N\left(0,\sigma_{\eta}^{2}\right)$ is a vector of
random shocks. We then evaluate $Q\left(\tilde{\theta}^{1}\right)$
and check whether the value is equal to zero. If $Q\left(\tilde{\theta}^{1}\right)=0$,
we accept the candidate $\tilde{\theta}^{1}$ and let $\theta^{1}\leftarrow\tilde{\theta}^{1}$.
If $Q\left(\tilde{\theta}^{1}\right)>0$, then we draw a new $\tilde{\theta}^{1}$
until we find a point that is accepted. Iterating this process generates
a random sequence of points $\theta^{0},\theta^{1},\theta^{2},...$
that ``bounces'' inside the level set $\Theta_{I}$. We iterate
this process until we find a large number of points in $\Theta_{I}$.

To control the step size, we let $\sigma_{\eta}$ adjust adaptively.
Specifically, if a candidate point is accepted, we increase $\sigma_{\eta}$
before a new draw is taken to make the search more aggressive. If
a candidate point is rejected, we decrease $\sigma_{\eta}$ to make
the search more conservative (a lower bound can be placed to prevent
excessively small step size).

\subsection{Counterfactual Analysis\label{subsec:Counterfactual-Analysis}}

In this section, we explain the implementation details for counterfactual
analysis. Let us first lay out the counterfactual prediction problem.
Let us call the game before and after the counterfactual policy pre-game
and post-game respectively. Suppose we have a counterfactual policy
that changes the pre-game $\left(G^{pre},S\right)$ to post-game $\left(G^{post},S\right)$
(we assume that a counterfactual policy only changes the payoff-relevant
primitives, but not the information structure). In our application,
we assume the counterfactual policy changes the covariates from $x^{pre}$
to $x^{post}$ so that the payoff function changes from $u_{i}^{pre}\left(a,\varepsilon_{i};\theta\right)\equiv u_{i}^{x^{pre},\theta}\left(a,\varepsilon_{i}\right)$
to $u_{i}^{post}\left(a,\varepsilon_{i};\theta\right)\equiv u_{i}^{x^{post},\theta}\left(a,\varepsilon_{i}\right)$.
We assume that the prior distribution $\psi$ and the baseline information
structure $S$ do not change.

Let $h:\mathcal{A}\times\mathcal{E}\to\mathbb{R}$ be the counterfactual
objective of interest, which is a function of realized state of the
world and action profiles (see examples provided below). At a fixed
$x\in\mathcal{X}$, we can obtain the lower/upper bounds on the expected
value of $h$ by finding the equilibria that will minimize/maximize
the expected value of $h$:
\begin{gather}
\min/\max_{\sigma^{x}\in\Delta_{a\vert\varepsilon,t}}\sum_{\varepsilon,t,a}\psi_{\varepsilon}^{x}\pi_{t\vert\varepsilon}^{x}\sigma_{a\vert\varepsilon,t}^{x}h\left(a,\varepsilon\right)\quad\text{subject to }\label{eq:counterfactual1}\\
\sum_{\varepsilon,t_{-i}}\psi_{\varepsilon}^{x}\pi_{t\vert\varepsilon}^{x}\sigma_{a\vert\varepsilon,t}^{x}\partial u_{i}^{x,\theta}\left(\tilde{a}_{i},a,\varepsilon\right)\leq0,\quad\forall i,t_{i},a,\tilde{a}_{i}.\nonumber 
\end{gather}
Note that (\ref{eq:counterfactual1}) is a linear program.

We now connect the characterizations to the empirical application.
Let $\mathcal{X}^{pre}$ be the set of covariates corresponding to
the food deserts in Mississippi; there can be multiple values of $x^{pre}\in\mathcal{X}^{pre}$
because there are multiple markets with different observable covariates.
By the definition of food deserts, all Mississippi food deserts have
covariates with the low access to healthy food indicator equal to
1. For each market $m$, we define the counterfactual covariates as
the vector obtained by changing the low access indicator from 1 to
0. For example, if $x^{pre}=\left(x^{highipc},x^{lowaccess}\right)=\left(1,1\right)$
for a particular market, we set $x^{post}=\left(1,0\right)$. This
changes the game since the players' payoff functions are changed.
Then the set of covariates for the post-regime $\mathcal{X}^{post}$
is constructed by taking each $x^{pre}\in\mathcal{X}^{pre}$ and changing
the low access indicator from 1 to 0.

We use four measures of market structure:

\begin{table}[h]
\centering{}%
\begin{tabular}{cc}
\hline 
Counterfactual objective & $h\left(a,\varepsilon\right)$\tabularnewline
\hline 
Number of entrants & $1\times\left(\mathbb{I}\left\{ a=\left(0,1\right)\right\} +\mathbb{I}\left\{ a=\left(1,0\right)\right\} \right)+2\times\mathbb{I}\left\{ a=\left(1,1\right)\right\} $\tabularnewline
McDonald's entry & $\mathbb{I}\left\{ a=\left(1,0\right)\right\} +\mathbb{I}\left\{ a=\left(1,1\right)\right\} $\tabularnewline
Burger King entry & $\mathbb{I}\left\{ a=\left(0,1\right)\right\} +\mathbb{I}\left\{ a=\left(1,1\right)\right\} $\tabularnewline
No entry & $\mathbb{I}\left\{ a=\left(0,0\right)\right\} $\tabularnewline
\hline 
\end{tabular}
\end{table}

Suppose $\theta$ is given. At a fixed covariate $x$, we can obtain
bounds on the expected value of $h$ by solving (\ref{eq:counterfactual1}).
However, since $\mathcal{X}^{pre}$ is non-singleton, we find the
weighted average of the bounds. Let $\left\{ w^{x}\right\} _{x\in\mathcal{X}^{pre}}$
be the weights at each covariate vector where $w^{x}$ is proportional
to the number of markets corresponding to Mississippi food deserts
in covariate bin $x\in\mathcal{X}^{pre}$; we scale the weights so
that $\sum_{x\in\mathcal{X}^{pre}}w^{x}=1$. The weighted average
on $h$ can be found by solving:
\begin{gather*}
\min/\max_{\sigma}\sum_{x\in\mathcal{X}^{post}}w^{x}\sum_{\varepsilon,t,a}\psi_{\varepsilon}^{x}\pi_{t\vert\varepsilon}^{x}\sigma_{a\vert\varepsilon,t}^{x}h\left(a,x\right)\quad\text{subject to }\\
\sum_{\varepsilon,t_{-i}}\psi_{\varepsilon}^{x}\pi_{t\vert\varepsilon}^{x}\sigma_{a\vert\varepsilon,t}^{x}\partial u_{i}^{x,\theta}\left(\tilde{a}_{i},a,\varepsilon\right)\leq0,\quad\forall x\in\mathcal{X}^{post},i,t_{i},a,\tilde{a}_{i}\\
\sigma^{x}\in\Delta_{a\vert\varepsilon,t},\quad\forall x\in\mathcal{X}^{post}
\end{gather*}
The bounds for the pre-counterfactual regime can be found by replacing
$\mathcal{X}^{post}$ with $\mathcal{X}^{pre}$.

Finally, since $\Theta_{I}$ is set-valued, we repeat the above process
for each $\theta$ in $\Theta_{I}$ and take the union of the bounds.
Since there is a large number of points in $\Theta_{I}$, to save
computation time, we use $k$-means clustering on $\Theta_{I}$ to
find a set of points that approximate $\Theta_{I}$ (we choose $k$
equal to 2000 or larger and compare the projection of the original
set to the projection of the approximating set to see if the approximation
is accurate).

\subsection{Overview of the Implementation\label{subsec:Overview-of-the}}

We provide a brief overview of how we obtain the confidence sets in
the empirical application section. To prepare data for structural
estimation, we use \texttt{Stata} to obtain discretized bins of covariates.
We use estimate the conditional choice probabilities via nonparametric
frequency estimator. We also compute the number of observations in
each bin $x\in\mathcal{X}$ (which are inputs to constructing simultaneous
confidence intervals for the CCPs) and define weights at each $x$
(which are inputs to criterion function) as being proportional to
the number of observations. The final dataset has $\vert\mathcal{X}\vert$
rows, where each row contains vector of covariate values corresponding
to bin $x$, CCP estimates $\hat{\phi}_{a}^{x}$ for each outcome
$a\in\mathcal{A}$, and the number of observations at the covariate
bin. We then export the data to \texttt{Julia} where all computations
for structural estimation are done.

To prepare feasible optimization programs, we discretize the space
of shocks using the approach described in Section \ref{subsec:Discretization-of-Unobservables}.
We then declare optimization program using \texttt{JuMP} interface
\citep{dunning2017jumpa}.\footnote{The main advantages of \texttt{JuMP} are its ease of use and its automatic
differentiation feature which does not require the researcher to provide
first- and second-order derivatives.} We construct the simultaneous confidence sets for the conditional
choice probabilities using the approach described in \ref{subsec:Construction-of-Convex}.
This makes evaluation of the criterion functions $\widehat{Q}_{n}^{\alpha}\left(\theta\right)$
for each point $\theta\in\Theta$ a linear program. We use \texttt{Gurobi}
to solve linear programs.

To approximate the confidence set $\widehat{\Theta}_{I}^{\alpha}$,
we need to collect many points in $\Theta$ that satisfy the condition
$\hat{Q}_{n}^{\alpha}\left(\theta\right)=0$. Collecting these points
are done by the random walk surface scanning algorithm described in
Section \ref{subsec:Random-Walk-Surface}. To use this approach, it
is important to quickly identify an initial point $\theta^{0}$ such
that $\hat{Q}_{n}^{\alpha}\left(\theta^{0}\right)=0$ by solving $\min_{\theta}\hat{Q}_{n}^{\alpha}\left(\theta\right)$.
This can be done efficiently by using gradients of $\hat{Q}_{n}^{\alpha}\left(\theta\right)$
obtained by the envelope theorem (see Theorem \ref{thm:Implementation}).
We recommend using many initial points to increase the chance of convergence,
and decreasing the tolerance for optimality conditions ($\Vert\nabla\hat{Q}_{n}^{\alpha}\left(\theta\right)\Vert<\varepsilon^{tol}$)
for higher accuracy. We use \texttt{Knitro} to solve nonlinear programs.

Specifically, we identify $\arg\min_{\theta}\hat{Q}_{n}^{\alpha}\left(\theta\right)$
by solving the minimization problem in two steps:
\[
\min_{\theta}\hat{Q}_{n}^{\alpha}\left(\theta\right)=\min_{\theta^{\rho}}\left(\min_{\theta^{u}}\hat{Q}_{n}^{\alpha}\left(\theta^{u};\theta^{\rho}\right)\right)
\]
where $\theta^{u}$ represent the parameters associated with the payoff
functions and $\theta^{\rho}$ represent the correlation coefficient
parameter for the distribution of payoff shocks. In the outer loop,
we search for the minimum over a grid of $\theta^{\rho}$ on $\left[0,1\right]$.
In the inner loop, taking $\theta^{\rho}$ as fixed, we solve $\min_{\theta^{u}}\hat{Q}_{n}^{\alpha}\left(\theta^{u};\theta^{\rho}\right)$
by minimizing (\ref{eq:program Qhat}) \emph{jointly} with $\theta^{u}$.
Solving the nested optimization program (the outer loop minimizes
over $\theta^{u}$ and the inner loop minimizes over $q,\sigma,\phi$)
as a single optimization program is faster when the number of variables
is manageable; this is similar to the key idea of \citet{su2012constrained}.
Although we can obtain $\psi^{x,\theta^{\rho}}$ in closed form so
that the minimization problem can be solved jointly in $\left(\theta^{u},\theta^{\rho}\right)$,
we chose to divide the minimization problem as above because $\psi^{x,\theta^{\rho}}$
can be highly non-linear in $\theta^{\rho}$.

\part*{}

\section{Data Appendix\label{sec:Data-Appendix}}

This section describes the datasets used for our empirical application,
which studies the entry game between McDonald\textquoteright s and
Burger King in the US. The following table provides an overview of
the datasets used in this paper.

{\footnotesize
\centering
\setstretch{1.1}
\begin{center}
\begin{longtable}[c]{>{\raggedright}p{4cm}>{\raggedright}p{12cm}}
\hline 
\textbf{Dataset Name} & \textbf{Description}\tabularnewline
\hline 
Data Axle (Infogroup) Historical Business Database & Proprietary; accessed via Wharton Research Data Services \url{https://wrds-www.wharton.upenn.edu/}
using institutional subscription.\footnote{Wharton Research Data Services (WRDS) was used in preparing part of
the data set used in the research reported in this paper. This service
and the data available thereon constitute valuable intellectual property
and trade secrets of WRDS and/or its third-party suppliers.} Data Axle (formerly known as Infogroup) is a data analytics marketing
firm that provides digital and traditional marketing data on millions
of consumers and businesses. Address-level records on business entities
operating in the US are available for 1997-2019 at the annual level.
We obtain the addresses of burger outlets in operation, which in turn
are translated into tract-level entry decisions for each calendar
year using the census shapefiles.\tabularnewline
\hline 
US Census Shapefiles & Accessible from \url{https://www.census.gov/geographies/mapping-files/time-series/geo/tiger-line-file.html}.
Used to get 2010 census tract boundaries. Shapefiles are needed to
find tract IDs corresponding to each physical store given their location
coordinates.\tabularnewline
\hline 
Longitudinal Tract Data Base (LTDB) & Accessible from \url{https://s4.ad.brown.edu/projects/diversity/researcher/bridging.htm}.
LTDB provides tract-level demographic information (from the census)
for 1970-2010 harmonized to 2010 tract boundaries. We obtain population
and income per capita for year 2000 and 2010 from here.\tabularnewline
\hline 
National Neighborhood Data Archive (NaNDA) & Accessible from \url{https://www.openicpsr.org/openicpsr/nanda}.
NaNDA provides measures of business activities at each tract. We obtain
the number of eating and drinking places for year 2010 at the tract
level. Other variable such as the number of grocery stores (per square
miles), the number of super-centers, and the number of retail stores
are available.\tabularnewline
\hline 
Food Access Research Atlas & Accessible from \url{https://www.ers.usda.gov/data-products/food-access-research-atlas/}.
Food Access Research Atlas provides information on whether a census
tract has limited access to supermarkets, super-centers, grocery stores,
or other sources of healthy and affordable food. We obtain indicators
for ``low access to healthy food'' and ``food deserts'' at the
tract level for year 2010. A census tract is classified as a food
desert if it is identified as having low access to healthy food and
low income. A census tract is classified as low-access tract if at
least 500 people or at least 33 percent of the population is greater
than 1/2 mile from the nearest supermarket, supercenter, or large
grocery store for an urban area or greater than 10 miles for a rural
area.\footnote{An alternative measure uses 1 mile radius for urban area. Using the
1 mile radius measure does not change the qualitative conclusion of
our empirical analysis.} The criteria for identifying a census tract as low-income are from
the Department of Treasury's New Markets Tax Credit (NMTC) program.\tabularnewline
\hline 
\end{longtable}
\par\end{center}

}

\subsection{Data Construction}

We merge multiple datasets to construct the final sample used for
structural estimation. The details are described as follows.

\subsubsection*{Panel data at tract-year level}

Although we use 2010 cross-section for estimation of the structural
model, we construct a panel dataset at a tract-year level to track
the openings and closings of fast-food outlets in the US. We make
the sample period run from 1997 to 2019, corresponding to the period
for which business location data from Data Axle Historical Business
Database are available.

We define the units for \emph{markets} as 2010 census tracts designated
by the US Census Bureau. (We define potential markets as 2010 urban
tracts. See below for the definition of urban tracts.) Year 2010 was
selected since it was the latest year for which the decennial census
data was available when we started the empirical analysis. For all
years in the sample period, we fix markets as 2010 census tracts;
although census tract boundaries change slightly every decade, we
fixed the boundaries for consistency across time.

To construct tract-level data, we first download the 2010 census shapefiles
from the US Census to obtain the list of all 2010 census tracts (there
are 74,134 tracts defined for the 2010 decennial census in the US
and its territories). Next, we exclude all tracts outside the contiguous
US: Alaska, Hawaii, American Samoa, Guam, Northern Mariana Islands,
Puerto Rico, and the Virgin Islands. We drop these regions since the
data generating process (specifically how the game depends on observable
market characteristics) is likely to differ from the rest.

Using the market-year panel data as a ``blank sheet'', we append
relevant variables that include the firms' entry decisions in each
tract for a given year and observable tract characteristics such as
population. At this stage, we can create a variable distance to headquarter
by measuring the distance between the location of a firm's headquarter
and the centroid of a tract (McDonald's and Burger King have their
headquarters in Chicago and Florida respectively).

In the final dataset used for the empirical application, we restrict
attention to 2010 urban census tracts (i.e., we drop all rural tracts).
A census tract is defined as urban if its population-weighted centroid
is in an ``urban area'' as defined in the Census Bureau's \emph{urbanized
area} definition; a census tract is rural if not urban. We obtain
the urban tract indicator from the Food Access Research Atlas.

\subsubsection*{Coding Entry Decisions}

The primary source of data for our empirical application is Data Axle\textquoteright s
\emph{Historical Business Data}base. The dataset contains the list
of local business establishments operating in the US over 1997-2019
at an annual level. Each establishment is assigned a unique identification
number which can be used to construct establishment-level panel data.
In addition, the dataset contains information such as company name,
parent company, location of the establishment in coordinates, number
of employees, industry codes.

We first need to download the entire list of burger outlets that were
in operation. We download raw data from Wharton Research Data Services
(WRDS) using the qualifier \textquotedblleft SIC code=58\textquotedblright{}
(retail eating places). We then identify relevant burger chains using
company (brand) names and their parent number. In principle, each
burger chain should have a unique parent number by the data provider.
For example, all McDonald\textquoteright s outlets have parent number
\textquotedblleft 001682400\textquotedblright . Ideally, one can identify
all burger chains that belong to a brand using their names and parent
numbers. However, there are some errors due to misclassifications,
which makes identifying all relevant burger chains more difficult.
For example, McDonald\textquoteright s outlets will have different
company names such as \textquotedblleft MC DONALD\textquoteright S\textquotedblright ,
\textquotedblleft MCDONALDS\textquotedblright , and \textquotedblleft MC
DONALD\textquotedblright . In addition, some McDonald\textquoteright s
outlets have parent numbers missing for some subset of years, or some
establishments have duplicate observations.\footnote{The main hurdle in constructing establishment-level panel data is
the following. Each establishment is assigned a unique ``ABI number''
which allows the analyst to track how the establishment operates over
time. However, we found that some establishments had their ABIs changing
over time or one establishment had duplicate observations with different
ABI numbers assigned. When we inquired the original data provider
support team about why this issue might be arising, they responded
that it seems to be errors generated in the data recording stage.}

To overcome this issue, we rely on the coordinates information to
identify unique establishments. Since the same establishment can have
different coordinates assigned over time depending on which point
of place is used to measure the coordinates, we put each establishment
in blocks approximately 250 meters in height and width. The idea is
to put all observations whose coordinates are very close to each other
in a single bin. Then we assign a unique establishment id to the stores
in each block, i.e., we treat them as corresponding to a single store.
We find that while it is challenging to avoid minor classification
errors, the total number of burger chains outlets identified by our
procedure closely is very close to the total number of outlets reported
by other sources (e.g., reports in Statista \url{https://www.statista.com/}).
Identifying unique establishments allows the construction of establishment-level
panel data, which can be used to track firm entries and exits in each
market.

The final step is to reshape the establishment-level panel data to
market-level data to tabulate the number of burger chains operating
in each market-year pair. We accomplish this with the help of \texttt{Stata}'s
geocoding function, which helps identify census tract id\textquoteright s
corresponding to each coordinate (location of establishments). We
then tabulate the number of outlets by each brand at a year-tract
level.

In each market, we code entry decisions as binary variables. There
were very few cases of a firm having more than one outlet in a single
tract (approximately 1.5\% of markets for McDonald's and 0.3\% for
Burger King). We also construct a firm-specific variable \emph{own
outlets in nearby markets}. This variable records the number of own-brand
outlets operating in adjacent markets (neighboring markets that share
the same borders). For example, if for market $m$, McDonald\textquoteright s
nearby outlets are 2, it means that there were a total of 2 outlets
operating in markets adjacent to market $m$. We constructed this
variable with the help of a dataset downloaded from \emph{Diversity
and Disparities project} website that provides the list of 2010 census
tracts and adjacent tracts.\footnote{Accessible from \url{https://s4.ad.brown.edu/Projects/Diversity/Researcher/Pooling.htm}.}

\subsubsection*{Market Characteristics}

We obtain tract-level characteristics from multiple sources described
in the table above. All of these datasets provide variables at tract-level
for the year 2010. We append tract-level characteristics to the main
dataset that has entry decisions and firm-specific variables at tract-level.

\section{Adjustment Costs\label{sec:Adjustment-Costs-in}}

\subsection{Adjustment Costs in the Model}

Throughout the paper, we have assumed that adjustment costs (e.g.,
sunk entry costs) are zero. The assumption is very common for static
entry models (see, e.g., Bresnahan and Reiss \citeyearpar{bresnahan_entry_1990,bresnahan_entry_1991},
\citet{ciliberto2009marketstructure}, \citet{ciliberto2021superstar},
\citet{ciliberto2021marketstructure}, and \citet{magnolfi2021estimation}).
The assumption is also commonly used in other econometric frameworks,
e.g., matching, network formation, and consumer choice. Whether the
zero-adjustment costs assumption is appropriate or realistic depends
on the research question, empirical setting, and the empirical model.
If the researcher believes that adjustment costs (might) drive the
observed decisions, then the researcher should model the adjustment
costs. For instance, adjustment costs may play an important role in
firms' decision timings in a dynamic model. However, in static entry
games, static profit functions are often interpreted as a reduced-form
representation of long-run profit associated with the firms' decisions.
In this case, one-time sunk payment associated with changing the operating
status (e.g., from ``staying out'' to ``staying in'') is likely
to be small relative to long-run profit. Although adjustment costs
can be introduced in the model, the zero-adjustment costs assumption
helps us motivate the main topic of this paper and simplify the exposition.
Finally, we remark that we cannot identify adjustment costs with typical
cross-sectional data that do not provide information on how firms
switched from one state to another.

We note that our main results up to the econometrics section are not
affected by the zero-adjustment costs assumption. If there are adjustment
costs, then we can treat the realized action profile as a state variable
and introduce adjustment costs as a function of a player's current
action and new action. Introducing adjustment costs do not affect
the form of outcome functions, which serve as players' model of the
relationship between signal profiles and outcomes. Thus, the definition
of rational expectations equilibrium is unaffected; the presence of
adjustment costs only makes deviation from the realized outcome more
difficult. The relationship between Bayes stable equilibrium and rational
expectations equilibrium (Theorem \ref{thm:BSE and REE connection})
also remains intact. However, it should be clear that introducing
adjustment costs in a static model requires assumptions about the
timing and speed of firms' actions as well as the size of the adjustment
costs relative to long-run profits. Also note that an outcome function
would be silent on how the players reached a certain outcome in the
presence of adjustment costs.

\subsection{Adjustment Costs in the Fast-Food Industry}

In the empirical section, we assume that the burger chains can revise
their actions costlessly. One-time sunk costs incurred for revising
a decision (from ``in'' to ``out'' or vice versa) are likely negligible
relative to the discounted sum of payoffs over long horizon, especially
when the one-time payments are not significant and firms tend to adhere
to a certain decision for a long period. In fact, the one-time payments
associated with adjusting the operating status of a restaurant seem
non-significant in the fast-food industry.\textcolor{blue}{{} }

When a fast-food outlet opens or closes, the store has to incur one-time
sunk costs that would not be incurred in the absence of adjustments.
For example, a store that is opening has to pay legal fees for obtaining
licenses to operate a restaurant. A store that intends to close might
be constrained by terms of contract at least in the short run. The
total amount of adjustment costs varies, but our investigation of
fast-food restaurants suggests that adjustment costs are generally
small relative to the long-run profits in the industry.\footnote{See, e.g., \citet{meyer2013open} for an in-depth discussion on opening
and operating costs of restaurants in the US.} While the details of the cost structure of McDonald's and Burger
King are proprietary, their franchise disclosure documents provide
a rough idea of the relative magnitude of sunk entry costs. The 2021
McDonald's Franchise Disclosure Document (FDD) reports that the average
annual sales volume of domestic (US) traditional McDonald's restaurants
was \$3,487,000 (the average profit per store is not reported since
rents vary considerably depending on the outlet's location). Operating
expenses include food costs, labor costs, rent/lease, loan payments,
equipment leases, fixed salaries, taxes, advertising fee, promotion,
utilities, insurance, office supplies/postage/shipping, etc. A typical
opening costs for a restaurant include those for construction, remodeling,
licenses/permits, and professional/legal services. The FDD reports
that expenses on signs, seating, equipment, and decor typically ranges
between \$1,000,000 and \$1,600,000. However, a subset of these expenses
is not sunk as they include purchase of capital that can be resold
or used in other outlets. Exit costs are less well-documented. However,
when a fast-food restaurant decides to close, it is likely that land,
buildings, equipment and furniture can be sold at a fair market value
or used in other outlets. Thus, it seems reasonable to assume that
adjustment costs do not play an important role in shaping firms' long
run decision.

The profit function parameter estimates from \citet{aguirregabiria2020identification}
also support our assumption that the adjustment costs are negligible.
In their empirical application, they study the dynamic entry game
played by McDonald's and Burger King in the UK during the period 1991\textendash 1995.
The estimates of structural payoff parameters, reported in Table 9
of their paper, show that the sunk entry costs are small relative
to the estimated annual variable profits, implying that sunk entry
costs are likely to be negligible in the long run. They do not estimate
exit costs because there were no exists during the period they study.

\section{Outcome Functions\label{sec:Outcome-Functions}}

Rational expectations equilibrium assumes that players refine their
information via a commonly agreed outcome function. In this section,
we further examine the role of outcome functions by allowing each
player to hold a different outcome function and studying the implications.
We discuss minimal conditions that must be imposed on outcome functions
in equilibrium and how outcome functions relate to players' ability
to process information. 

\subsection{The Role of Outcome Functions}

In the baseline version of the rational expectations equilibrium,
it is assumed that all players agree on a single outcome function
$\delta:\mathcal{T}\to\Delta\left(\mathcal{A}\right)$. In this case,
$\delta$ plays multiple roles. First, $\delta$ represents the true
data generating process in a reduced form. While the model is silent
on how $\delta$ comes about, the use of $\delta$ reflects the fact
that some relationship between exogenous signals and endogenous outcomes
exists. Second, $\delta$ represents players' subjective model of
how signals relate to actions. The rational expectations equilibrium
assumes that players' models are correct. Third, $\delta$ represents
players' information updating rule. Each player use $\delta$ to assess
if deviation from a given outcome can be profitable. Clearly, such
assumption on information updating is strong as it ignores how the
players might have updated their information along the interaction
process before stable actions are determined and realized, but it
substantially simplifies the characterization of the equilibrium.

\subsection{Heterogeneous Outcome Functions}

Relaxing the symmetry assumption and allowing players to hold heterogenous
outcome functions disentangles the roles. Suppose that each player
$i$ uses outcome function $\delta_{i}:\mathcal{T}\to\Delta\left(\mathcal{A}\right)$.
Each $\delta_{i}$ represents $i$'s subjective model of how players'
information relates to market outcomes. We maintain the ``rational
expectations'' assumption\textemdash that each player refines her
information after observing $a$ using Bayes' rule with $\delta_{i}$\textemdash since
the players can observe opponents' actions at stable outcomes. However,
we relax the baseline assumption that $\delta_{i}$'s be identical
across players. 

\subsubsection{Data Generating Process}

To say that $\delta_{i}$ is correct or incorrect requires reference
to the true data generating process, so we need to introduce $\delta^{*}:\mathcal{T}\to\Delta\left(\mathcal{A}\right)$
to denote the true data generating process. Each $\delta_{i}$ represents
player $i$'s subjective model, whereas $\delta^{*}$ represents the
objective model. If $\delta_{i}\neq\delta^{^{*}}$, then player $i$'s
model is misspecified. 

\subsubsection{Minimal Consistency Requirements}

Although $\delta_{i}$'s are subjective, any reasonable notion of
equilibrium will require that $\delta_{i}$'s cannot be too arbitrary.
First, we should require that $\delta_{i}$'s generate no deviation
incentives. This requires evaluating whether $i$ finds deviation
incentives at every information set that can be realized with positive
probability under $\delta^{*}$. Second, a player's subjective belief
cannot contradict information available at equilibrium. The true data
generating process $\delta^{*}$ should be consistent with the actual
behavior induced by $\delta_{i}$'s. Moreover, the empirical distribution
over a player's information set should be consistent with the distribution
implied by $\delta_{i}$ since the player can detect inconsistency
otherwise.\footnote{That players' subjective models can be supported as long as the players'
models do not contradict their observations is interesting because
it is reminiscent of the self-confirming equilibrium of \citet{fudenberg_self-confirming_1993}.} Complete characterization of the equilibrium conditions with heterogeneous
beliefs depends on the analyst's assumptions, but it is clear that
any reasonable concept of equilibrium will constrain players' subjective
beliefs from being too arbitrary because rational players will use
information available at equilibrium to refine their beliefs.

\subsubsection{Processing Information}

In rational expectations equilibrium, $\delta_{i}$ represents a channel
through which player $i$ refines her information upon observing market
outcomes. Given that $\delta_{i}$ is subjective, it is natural to
ask whether $\delta_{i}$'s can be used to reflect differences in
the ability to process information. For example, can we capture differences
in McDonald's and Burger King's ability to process information by
setting $\delta_{i}$'s appropriately? It turns out that this is quite
difficult because the model imposes minimal level of sophistication
on players, which in turn constrains $\delta_{i}$'s from being too
arbitrary in equilibrium. While the analyst may consider setting $\delta_{i}$
so as to mimic limited ability to update information (e.g., by setting
$\delta_{i}:\mathcal{T}\to\Delta\left(\mathcal{A}\right)$ to a constant
function or controlling the degree of misspecification), such approach
is constrained by the consistency requirements discussed above (i.e.,
that $\delta_{i}$ cannot contradict what $i$ can observe or learn
in equilibrium). Thus, in the rational expectations equilibrium framework,
the assumptions on players' rationality limits the analyst's freedom
in manipulating $\delta_{i}$ to capture potential heterogeneity in
players' ability to process information.

\subparagraph*{}

\end{document}